\documentclass[runningheads]{llncs}
\usepackage{main/style}
\usepackage{fullpage}
\usepackage[T1]{fontenc}
\usepackage{color}

\spnewtheorem{fact}{Fact}{\bfseries}{\itshape}

\crefname{fact}{fact}{facts}
\Crefname{fact}{Fact}{Facts}

\newcommand{\circuitscale}{0.85}

\makeatletter
\let\lncs@origmaketitle\maketitle
\renewcommand{\maketitle}{%
  \begingroup
    \let\addcontentsline\@gobblethree 
    \lncs@origmaketitle
  \endgroup
}
\makeatother

\begin{document}
\title{Scalable, quantum-accessible, and adaptive pseudorandom quantum state and pseudorandom function-like quantum state generators}
\titlerunning{Scalable, quantum-accessible, and adaptive PRS \& PRFS}

%
\author{Rishabh~Batra\inst{1}\orcidID{0000-0002-5829-2304} \and
Zhili~Chen\inst{1}\orcidID{0009-0002-2155-5817} \and
Rahul~Jain\inst{1,2,3}\orcidID{0000-0002-3649-6576} \and
YaoNan~Zhang\inst{1}\orcidID{0009-0009-3760-9729}}

\authorrunning{R. Batra et al.}

\institute{Centre for Quantum Technologies, Singapore
\\
\email{\{rishabhbatra,chen.zhili,zhang.yaonan\}@u.nus.edu}\\
\and
Department of Computer Science, National University of Singapore\\
\email{rahul@comp.nus.edu.sg}\\
\and
MajuLab, UMI 3654, Singapore
}
\maketitle              

\begin{abstract}

We show new constructions for pseudorandom quantum states (PRS) and pseudorandom function-like quantum state (PRFS) generators satisfying
\textbf{scalability}, which means the security parameter can be much larger than the number of qubits,
\textbf{quantum accessibility}, which means the adversary can provide quantum input,
and \textbf{adaptivity}, which means the adversary can query it adaptively.

We present an isometric procedure to prepare quantum states that can be arbitrarily random (i.e., the trace distance from the Haar-random state can be arbitrarily small for the true random case, or the distinguishing advantage can be arbitrarily small for the pseudorandom case).
This naturally gives the first construction for scalable, quantum-accessible, and adaptive PRFS assuming quantum-secure one-way functions.
Compared to prior PRFS works, we use a stronger definition of quantum accessibility in which the adversary can be ancilla-assisted, i.e., the input state may not be pure and could be entangled with other quantum registers.
Thus, our result also gives the first (fully) quantum-accessible PRFS.

Our PRFS construction implies various primitives, including long-input PRFS, short-input PRFS, short-output PRFS, non-adaptive PRFS, and classically-accessible adaptive PRFS \cite{AQY21,AGQY22}. This new construction may be helpful in simplifying the microcrypt zoo.

\end{abstract}

\newpage

\begingroup
\let\clearpage\relax
\setcounter{tocdepth}{3}
\hypersetup{hidelinks}
\renewcommand{\contentsname}{}
\tableofcontents
\endgroup


\vspace{3em}

\section{Introduction}
Randomness has played an indispensable role in various fields of computer science including cryptography, algorithms, and complexity theory. 
For practical applications, however, using true randomness may be very expensive and difficult to produce. Storing the complete description of a truly random object may be very inefficient. For example, to merely specify a binary truly random function on $n$-bit inputs, we require an exponentially large description. 

An object is called pseudorandom if we can generate it efficiently, but an efficient distinguisher cannot distinguish it from a truly random object. Using pseudorandomness enables us to design efficient protocols that require only short random seeds. In classical cryptography, pseudorandomness has been fundamental for the construction of efficient protocols that are secure against attacks by efficient adversaries. 

Let $U_n$ represent the uniform distribution on $n$ bits. A pseudorandom generator (PRG) is an efficiently computable function, $\PRG: \{0,1\}^n \rightarrow \{0,1\}^{l(n)}$ such that  $l(n)>n$ and $\PRG(U_n)$ is computationally indistinguishable from $U_{l(n)}$. A one-way function (OWF) is a function that is efficiently computable and hard to invert. 
The existence of one-way functions is a widely used computational assumption in classical cryptography. A pseudorandom function family (PRF) is a collection of efficiently computable functions that are computationally indistinguishable from a truly random function. It is known that in the classical world, OWF, PRG, and PRF are equivalent to each other~\cite{Goldreich_2001}.

\subsection*{Motivations}
In the quantum world, however, the story is not as straightforward. Ji, Liu, and Song~\cite{JLS18} first introduced the notion of pseudorandom quantum states (PRS), which are efficiently generatable pure quantum states that are computationally indistinguishable from truly random (Haar-random) quantum states — even when an adversary has access to multiple copies of the state. PRS can be thought of as a generalization of PRGs with quantum outputs. Similarly, pseudorandom function-like quantum state (PRFS) generators, first introduced by Ananth, Qian, and Yuen \cite{AQY21}, are the quantum analogue of PRF.

Morimae and Yamakawa \cite{Morimae_2022} showed that quantum commitments can be based on PRS (or PRFS).  It is known that PRS (and PRFS) can exist even if BQP = QMA (relative to a quantum oracle)~\cite{Kre} or if P = NP (relative to a classical oracle)~\cite{KSTQ}, which does not allow for the existence of one-way functions (relative to these oracles). Hence, these are potentially weaker objects than OWF, from which we can construct quantum commitments, oblivious transfer, and multi-party computation~\cite{BCQ,OT_GLSV,OT_BCKM,AQY21}.
Therefore, there is a natural interest in novel and efficient constructions of PRS and PRFS.    


The following are some desirable properties of the PRS and PRFS constructions:
\begin{itemize}
    \item \textbf{Scalability:} A PRS (or PRFS) generator is \textit{scalable} if the security parameter $\lambda$ and the output size $n$ (and the input size $m$ for PRFS) can be chosen independently.
    In other words, the distinguishing advantage can be arbitrarily small as long as we choose a sufficiently large $\lambda$, regardless of the state size.
    {Scalability} is not automatic in the quantum setting.
    Classically, an $n$-bit (pseudo)random string allows the use of any $k$-bit substring as a $k$-bit (pseudo)random string with $k < n$.
    However, in the quantum case, even if we already have an $n$-qubit Haar-random state, the first $k$ qubits (while tracing out the rest) are typically very different from a $k$-qubit Haar-random state. 
    Therefore, knowing how to create a longer pseudorandom object does not automatically provide a shorter one.
    Bouaziz-Ermann and Muguruza~\cite{BEM24} show that quantum pseudorandomness cannot be shrunk in a black-box way. 
    \item \textbf{Adaptivity:} For PRFS, it is desirable to have \textit{adaptive} security, i.e., the adversary can query it adaptively. The first PRFS construction~\cite{AQY21} guaranteed only selective (non-adaptive) security. 
    \item \textbf{Quantum accessible/superposition query:} A PRFS is quantum accessible if the adversary can provide quantum (superposition) input (potentially entangled with ancilla registers) into it.
    Ideally, a PRFS query should keep any global pure state (including in ancilla registers) a pure state; that is, there is no loss of information on performing the query. Such a PRFS generator would be an \textit{isometry}.
\end{itemize}



\begin{remark}[Isometry vs quantum-accessibility]
We would like to clarify that the quantum accessible property is the same as the isometric property for PRFS.
That is, ideal map W should be the isometry as mentioned in our \Cref{def:prfs}.
$$\sum \alpha_x 
    \ket{x}  \mapsto \sum \alpha_x 
    \ket{x} 
    \ket{\psi_x}.
$$
If it is not, then some input pure superposition query state would become a mixed state after the action of the ideal map. This would lose (to an extent) the (distinguishing) advantage that the adversary wanted to derive using the pure superposition query.

As such, it is not needed that the efficient construction V should be an isometry; it just needs to be indistinguishable from W by a QPT adversary. However, V being an isometry may have additional advantages of its own. For example, in our construction of PRS, V is an isometry as opposed to the construction by \cite{BS20}. This difference crucially allows us to extend it to a PRFS construction in a natural fashion by using separate randomness for each input $x$.


\end{remark}

\subsection*{Our contribution}
In this work, we present an isometric and easy-to-implement procedure to prepare quantum states that can be arbitrarily random (i.e., the trace distance from the Haar-random state can be arbitrarily small for the true random case, or the distinguishing advantage can be arbitrarily small for the pseudorandom case).
This procedure provides \textbf{a new method for scalable and isometric PRS}.

\begin{theorem}[Informal]
    Assuming a quantum-secure PRF, there exists a scalable and isometric PRS.
\end{theorem}

Using independent randomness for different inputs in our isometric state-generation procedure yields \textbf{the first construction for scalable, quantum-accessible, and adaptive PRFS}.
\begin{theorem}[Informal]
    Assuming a quantum-secure PRF, there exists a scalable, quantum accessible, and adaptive PRFS.
\end{theorem}

\begin{figure}[ht]
\resizebox{\textwidth}{!}{
 \begin{tikzpicture}[
    transform shape,
    node distance = 1.5em,
    circ/.style = {ellipse, draw, align=center, minimum size=1cm},
    arr/.style  = {-{Latex[length=2mm]}, thick}
]

\node[circ] (owf) {OWF};

\node[circ, below=of owf] (sprfs) {\textbf{Scalable, Quantum-Accessible, and Adaptive PRFS}};

\node[circ, below left=2em and 3em of sprfs] (a) {(Selective) Short Input PRFS};
\node[circ, below right=0 and 5em of a] (b) {Long Input PRFS};
\node[circ, above right=0 and 8em of b] (c) {Short Output PRFS};
\node[align=center, right=6em of c] (dots) {...};

\draw[arr] (owf) -- (sprfs);

\draw[arr] (sprfs) -- (a);
\draw[arr] (sprfs) -- (b);
\draw[arr] (sprfs) -- (c);
\draw[arr] (sprfs) -- (dots);

\node[circ, below=2em of a] (si-imply) {PRS, SB-QCOM, ...};
\node[circ, below=2em of b] (li-imply) {CCA1-qPKE with quantum ciphers, \\ ...};
\node[circ, below=2em of c] (so-imply) {PD-PRF, Pseudo-encryption ...};

\draw[arr] (a) -- (si-imply);
\draw[arr] (b) -- (li-imply);
\draw[arr] (c) -- (so-imply);

\end{tikzpicture}
}
\caption{Our PRFS construction implies various other primitives.}
\label{fig:imply-others}
\end{figure}

Combining with the fact that quantum-secure OWF implies quantum-secure PRF~\cite{Zha21}, our work shows that quantum-secure OWF implies scalable, quantum-accessible, and adaptive PRFS. 
Our PRFS construction implies various primitives, including long-input PRFS, short-input PRFS, short-output PRFS, non-adaptive PRFS, and classical-accessible adaptive PRFS \cite{AQY21,AGQY22}.
This gives a unified construction of PRFS, and may be helpful in simplifying the microcrypt zoo
\footnote{\url{https://sattath.github.io/microcrypt-zoo/}.}.

Since short-input PRFS implies SB-QCOM and quantum pseudo-encryption \cite{AGQY22}, long-input PRFS implies CCA1-qPKE with quantum ciphers \cite{BGH+23}, short-output PRFS implies PD-PRF \cite{ALY23}, and so on, our new construction can directly imply all these primitives. If a black-box separation between OWF and scalable PRFS can be shown in the future, this would suggest a potential cryptography world without OWFs, which encompasses a richer set of primitives than one starting from PRS.

\subsection*{Previous works and our novelty}

\begin{enumerate}
    \item {\bf Novelty over Ji, et al.~\cite{JLS18} and Brakerski and Shmueli~\cite{BS19} (PRS)}: PRS was first introduced by \cite{JLS18}, and they provide the first construction by applying a random phase to each computational basis element:
    \[
        \ket{\phi_k} = \frac 1 {\sqrt{N}} \sum_{z=0}^{N-1} \left(e^{2\pi i / N}\right)^{f_k(z)} \ket{z}.
    \]
    Later, \cite{BS19} proved that the binary phase state is also pseudorandom:
    \[
        \ket{\phi_k} = \frac 1 {\sqrt{N}} \sum_{z=0}^{N-1} (-1)^{f_k(z)} \ket{z}.
    \]
    These elegant constructions guarantee that $t$-copies of such a state are close to $t$-copies of the Haar-random state for polynomially bounded $t$ when the number of qubits is sufficiently large. However, since these constructions are actually very different from the true Haar-random state, they are not scalable: There exists a size-dependent distance (distinguishing advantage) that cannot be reduced. Thus, the size $n$ is required to be in $\Omega(\log \lambda)$ for security parameter $\lambda$.
    Our construction, on the other hand, is scalable and has no such restriction.

    \item {\bf Novelty over Brakerski and Shmueli~\cite{BS20} (scalable PRS)}: The first scalable PRS generator construction is given by \cite{BS20}.
    They prepare a state whose coefficients follow a normalized Gaussian random vector by \emph{quantum rejection sampling}. This approach prevents them from obtaining an isometric PRS. Their construction does not give a pure state output and has some entangled junk that needs to be traced out, and hence, it is not clear how to extend such a PRS to the construction of a quantum-accessible PRFS.
    
    Our construction is also based on this well-known fact that the distribution of each coefficient in a Haar-random state follows a complex Gaussian distribution (up to normalization), but we use a very different approach.
    Our key step is similar to the Grover-Rudolph algorithm \cite{GR02}: We run a Beta sampling algorithm in superposition and apply controlled rotations qubit by qubit. Then classical results relating Beta distribution, Gamma distribution, Chi-squared distribution, and Gaussian distribution give us the desired state.
    This gives us an isometric construction, and naturally leads to our PRFS construction.

    \item {\bf Novelty over Ananth et al. \cite{AQY21,AGQY22} (PRFS)}: The initial PRFS construction introduced by \cite{AQY21} guaranteed only selective security. Subsequently, Ananth, Gulati, Qian, and Yuen~\cite{AGQY22} gave new constructions for classically-accessible and quantum-accessible adaptive PRFS (secure against pure input), which are also not scalable.
    We present the first PRFS construction which is simultaneously quantum-accessible, adaptive, and scalable.
    
    \item {\bf Novelty over Lu et al. \cite{Scramblers} (PRSS)}: A new primitive called quantum pseudorandom scramblers (PRSS) was introduced by \cite{Scramblers}, which also provides a scalable and isometric PRS generation procedure.
    In their Appendix, they show that PRSS can imply classically-accessible selective PRFS (without invoking OWF). They also claim adaptive and quantum-accessible security for logarithmic-size inputs, using the fact that superposition queries provide no additional advantages when the output state is known for every input string.

    Our focus differs: They show that their new primitive implies (limited) PRFS in a black-box way. We directly construct PRFS that satisfies all the desired properties from a quantum-secure PRF. The resulting scalable, quantum-accessible, and adaptive PRFS implies all prior limited PRFS primitives.
    Using independent randomness for different inputs is natural, but security against quantum queries is not automatic. We explicitly allow the adversary to use quantum ancilla registers (see more details below). We bound the {\em diamond-norm distance} via the operator-norm distance between the corresponding isometries, without a dimension-dependent loss in the input size.

    Subsequent to our work appearing on ArXiv, in a personal communication from the authors of~\cite{Scramblers}, it was brought to our notice that, similar to the construction in this work, assuming in addition a quantum-secure PRF, one can sample a key for each $x$ by the PRF and obtain a scalable PRFS for arbitrarily long inputs.
    This construction is secure against adversaries using classical ancillary bits via straightforward triangle inequalities. However, it is currently unclear whether the construction remains secure when the ancillary qubits may be entangled with the input register.
    \footnote{
        Chuhan Lu, Minglong Qin, Fang Song, Penghui Yao, Mingnan Zhao.
        Personal communication.
    }.
    
    \item {\bf Novelty over all of the above:} For a quantum-accessible pseudorandom object, we naturally expect the adversary to be ancilla-assisted, i.e., the input may not be pure, but entangled with some other register held by the adversary. The recent constructions of pseudorandom unitaries (PRUs) by Ma and Huang \cite{PRU} indeed allow the adversary to use ancilla registers. No other work allows this in PRFS to the best of our knowledge.
    We give \textbf{the first (fully) quantum-accessible PRFS}.
    
    Since scalable PRU is still open, to the best of our knowledge, our result is the first quantum pseudorandom object that is scalable and secure against ancilla-assisted quantum adversaries.

    \item {\bf Difference with PRU:} PRU is only known to imply a PRFS with $m \leq n$. In the classical world, given a PRF ensemble $\{f_k: \{0, 1\}^{l(\lambda)} \to \{0, 1\}^{r(\lambda)} \}$, one can straightforwardly derive another PRF ensemble with different input/output lengths $\{f_k: \{0, 1\}^{l'(\lambda)} \to \{0, 1\}^{r'(\lambda)} \}$. Consequently, we typically do not distinguish between PRF and PRF with specific input/output lengths.
    However, this is not true for PRFS.

    Non-adaptive PRU was given by Metger et al.~\cite{nonaptive_PRU}, and adaptive PRU by \cite{PRU}. However, constructing a scalable PRU remains an open question. Since gluing small PRUs has recently become a popular technique for constructing low-depth PRU~\cite{log_dep_PRU,const_dep_PRU}, scalability appears to be a crucial property to improve the result \footnote{Currently known constructions either require the size of each small PRU to be $\omega(\log n)$, or use log-size small PRUs but achieve only inverse-polynomial security.}.

    Furthermore, even given a scalable PRU, it is unclear how to derive a scalable PRFS. This is similar to the classical case: A PRP over a small domain can be distinguished from PRF.

\item {\bf Difference with Grover-Rudolph \cite{GR02}:} To prepare the target state, we develop a new technique inspired by the Grover-Rudolph algorithm, yet distinct in its approach.
In the Grover-Rudolph algorithm, if a quantum state with given coefficients is to be created, qubit-by-qubit rotations are done by computing the ratio between the sum of the coefficients. In our case, we are creating a random quantum state with normalized Gaussian vector coefficients by performing qubit-by-qubit rotations using Beta sampling algorithms, feeding prefixes in superposition as randomness seeds.

This provides a general framework for efficiently sampling a random quantum state, as long as the distribution of the ratio between partial coefficient sums is independent and efficiently samplable. This may prove useful for other applications in quantum cryptography and algorithms.

\end{enumerate}

In \Cref{tab:prs-comparison} and \Cref{tab:prfs-comparison}, we compare our work with prior PRS and PRFS constructions.

\begin{table}[!ht]
\centering
\footnotesize
\caption{Comparison with prior PRS constructions.}
\label{tab:prs-comparison}
\begin{tabular}{||c c c c||}
 \hline
 Work by & PRS Construction & Scalable & Isometry \\ [0.5ex] 
 \hline\hline
 \cite{JLS18} & Random phase & No & Yes \\ 
 \hline
 \cite{BS19} & Binary phase & No & Yes \\
 \hline
 \cite{BS20} & Quantum rejection sampling & Yes & No \\
 \hline
 \cite{Scramblers} & Kac's walk & Yes & Yes \\
 \hline
 This work & Qubit-by-qubit rotation & Yes & Yes \\ 
 \hline
\end{tabular}
\end{table}

\begin{table}[!ht]
\centering
\footnotesize
\caption{Comparison with prior PRFS constructions.}
\label{tab:prfs-comparison}
\begin{tabular}{||c c c c c||} 
 \hline
 Work by & Scalable & Adaptive & Quantum-accessible & Ancilla-assisted adversary\\ [0.5ex] 
 \hline\hline
 \cite{AQY21} & No & No & No & Not applicable \\ 
 \hline
 \cite{AGQY22} & No & Yes & Yes & Unknown \\
 \hline
 \cite{Scramblers}(Lemma 16) & Yes & No & No & Not applicable \\
 \hline
 \cite{Scramblers}(Lemma 17) & No & Yes & Yes & Unknown \\
 \hline
 This work & Yes & Yes & Yes & Yes \\ 
 \hline
\end{tabular}
\end{table}
 
\subsection*{Future work}

Using ideas similar to our construction of PRS and PRFS, it may be possible to construct more scalable quantum pseudorandom primitives, such as pseudorandom isometries and pseudorandom unitaries, which is a major open question.

As shown by Ananth et al.~\cite{AGQY22}, using the technique of ``verifiable tomography'', they are able to obtain a bit-commitment protocol with classical communication from a scalable PRS. For this application, the scalability of the PRS is important. It is conceivable that, using similar arguments, the technique of verifiable tomography may help in getting classical communication analogues of some other primitives, e.g., MAC with quantum tags, qPKE with quantum ciphers, etc., from our scalable PRFS. These are currently obtained from long-input PRFS.

Another interesting topic is the relationship between OWF and scalable PRFS (see \Cref{fig:imply-others}).
As mentioned before, if a black-box separation can be shown, this would suggest a potential cryptography world without OWF, which encompasses a richer set of primitives than one starting from PRS.
On the other hand, if we can show that OWF and scalable PRFS are equivalent, this would significantly reshape our understanding of quantum cryptographic primitives.

\section{Technical overview}

An asymptotically random state (ARS) generator is a statistical notion of PRS that generates states that are statistically indistinguishable from Haar-random states but requires (quantum) access to an exponential amount of randomness. Then, replacing the exponentially large random string with a quantum-secure pseudorandom function (PRF), a PRS construction naturally follows from an ARS construction.

We now provide an outline of how we achieve a scalable ARS using a $\C \to \C^{2^n}$ isometric procedure, given quantum access to a random function.
The output state is exactly the Haar-random state in the idealized setting with infinite precision and a perfect sampling algorithm, and has negligible error with finite precision.
We use a similar technique to construct scalable and adaptive asymptotically random function-like quantum state (ARFS) generators and PRFS generators.

\subsection{Random amplitudes quantum state}

From \Cref{fact:haar}
we know that for a Haar-random state, its coefficients follow the distribution of the normalized standard Gaussian random vector (\Cref{def:gaussian-vector}). We want to construct a state
\[
    \ket{\phi} = \frac 1 {\sqrt{\sum_{z \in [N]} A_z^2 + B_z^2}} \sum_{z \in [N]} (A_z + i B_z) \ket{z},
\]
where $N=2^n$ the size of Hilbert space, and $A_z, B_z \sim \cN(0, 1)$ i.i.d.
To do this, we first consider the amplitudes $\sqrt{C_z} = \sqrt{A_z^2 + B_z^2}$ and create the following state (we call it the random amplitudes quantum state):
\[
    \ket{\xi} = \frac 1 {\sqrt{\sum_{z \in [N]} C_z}} \sum_{z \in [N]} \sqrt{C_z} \ket{z}.
\]
Each $C_z$ follows the $\chi^2_2$ distribution if $A_z, B_z \sim \cN(0, 1)$.
Once we can prepare such a state, we can obtain the Haar-random state by applying a random phase $e^{2\pi i U_z}$, where $U_z \sim \cU(0, 1)$, to each $\ket{z}$ (\Cref{lma:phase-gaussian}):
\[
    \ket{\phi} = \frac 1 {\sqrt{\sum_{z \in [N]} C_{z}}} \sum_{z \in [N]} \sqrt{C_z} e^{2\pi i U_z} \ket{z}. 
\]
Merging all the $\ket{z}$ starting with $0$ and all the $\ket{z}$ starting with $1$, respectively, we can represent $\ket{\xi}$ as
\[
    \ket{\xi} = \frac 1 {\sqrt{W_0 + W_1}} ( \sqrt{W_0} \ket{0}\ket{\xi_0} + \sqrt{W_1} \ket{1}\ket{\xi_1}),
\]
where $W_0 = \sum_{z: z[0] = 0} C_z$ and $W_1 = \sum_{z: z[0] = 1} C_z$.
If each $C_i \sim \chi^2_2$,
then $W_0, W_1 \sim \chi^2_{2^n}$.
This shows a very straightforward way to make the ``weight'' of the first qubit correct: Sample $w_0, w_1$ from $\chi^2_{2^n}$, and rotate the first qubit:
\[
    \ket{0} \to \frac 1 {\sqrt{w_0 + w_1}} (\sqrt{w_0} \ket{0} + \sqrt{w_1} \ket{1}). 
\]

We can continue to make the ``weight'' of the first two qubits correct.
Let $W_{ab} = \sum_{z: z[0, 1] = ab} C_z$ for $a, b \in \{0, 1\}$.
Similarly, we represent $\ket{\xi_0}$ as
\[
    \ket{\xi_0} = \frac 1 {\sqrt{W_{00} + W_{01}}} ( \sqrt{W_{00}} \ket{0}\ket{\xi_{00}} + \sqrt{W_{01}} \ket{1}\ket{\xi_{01}}),
\]
and represent $\ket{\xi_1}$ as
\[
    \ket{\xi_1} = \frac 1 {\sqrt{W_{10} + W_{11}}} ( \sqrt{W_{10}} \ket{0}\ket{\xi_{10}} + \sqrt{W_{11}} \ket{1}\ket{\xi_{11}}).
\]
Since each $C_z \sim \chi^2_2$, we have $W_{a b} \sim \chi^2_{2^{n-1}}$. We can sample $w_{00}, w_{01}, w_{10}, w_{11}$ from $\chi^2_{2^{n-1}}$ under the condition that $w_{00} + w_{01} = w_0$ and $w_{10} + w_{11} = w_1$, and rotate the second qubit controlled on the first qubit:
\begin{align*}
    &\frac 1 {\sqrt{w_0 + w_1}} (\sqrt{w_0} \ket{0} + \sqrt{w_1} \ket{1}) \ket{0} \\
    &\to
    \frac 1 {\sqrt{w_{00} + w_{01} + w_{10} + w_{11}}} (\sqrt{w_{00}} \ket{00} + \sqrt{w_{01}} \ket{01} + \sqrt{w_{10}} \ket{10} + \sqrt{w_{11}} \ket{11}).
\end{align*}

\Cref{fact:gamma-beta} tells us if $W_0, W_1 \sim \chi^2_{2^n}$ independently, then $W_0 / (W_0 + W_1) \sim \dBeta(2^{n-1}, 2^{n-1})$ and is independent of $W_0 + W_1$. This makes the above task easier. The overall process (for making the ``weight'' of the first two qubits correct) is as follows:
\begin{enumerate}
    \item Initialize $\ket{0}^{\otimes n}$.
    \item Sample $b \gets \dBeta(2^{n-1}, 2^{n-1})$ and rotate the first qubit:
    \[
        \ket{0} \to \sqrt{b} \ket{0} + \sqrt{1-b} \ket{1} .
    \]
    \item Sample $b_0, b_1 \gets \dBeta(2^{n-2}, 2^{n-2})$ and rotate the second qubit controlled on the first qubit:
    \[
        \ket{00} \to \ket{0}(\sqrt{b_0}\ket{0} + \sqrt{1-b_0}\ket{1}),
    \]
    \[
        \ket{10} \to \ket{1}(\sqrt{b_1}\ket{0} + \sqrt{1-b_1}\ket{1}).
    \]
\end{enumerate}

Doing this bit by bit, we can obtain the random amplitudes (see \Cref{lma:beta-ramp} for details).
If we directly apply the above procedure, we need to sample $2^t$ Beta random variables and do $2^t$ controlled rotations for the $t$-th qubit, but we can do this in superposition, which makes the procedure efficient.
This is inspired by the Grover-Rudolph algorithm \cite{GR02}.
While the Grover-Rudolph algorithm uses each prefix to compute the conditional probability mass, we use the prefix as a random seed to sample an independent Beta random variable.

Assume that we have a deterministic classical Beta sampling algorithm $\alg_B(k, z)$ which samples $b \gets \dBeta(2^k, 2^k)$ with random seed $z$ (with finite precision). We use its reversible quantum version $U_B^{(k)}: \ket{z}\ket{0} \mapsto \ket{z}\ket{b_z}$ to compute an independent Beta random variable $b_z$ (with finite precision) for each prefix $z$ and store in an ancilla qubit, and then apply a rotation controlled on this ancilla qubit.
These are described in \Cref{alg:ramp}.
We need two things for such $\alg_B(k, z)$: A random function to convert the random seed $z$ to randomness $r_z$, and a Beta sampling algorithm which uses randomness $r_z$.

In the Appendix, we show a classical sampling algorithm we used for the random amplitudes procedure. While this algorithm is well-known and widely used in practice (for example, in the NumPy library), 
we provide a rigorous proof of its correctness and show that the error is negligible when truncating to finite precision.

\begin{figure}[!ht]
\resizebox{\textwidth}{!}{    
\begin{tikzpicture}[
    grow'=right,level distance=110pt, sibling distance=12pt,
    every node/.style={font=\footnotesize}, every tree node/.style={font=\normalsize},
    edge from parent/.style={draw, ->}
]
\Tree [.$\ket{000}$
    \edge node[auto=left]{$\sqrt{b_{0, \emptyset}}$}; [.$\ket{000}$
        \edge node[auto=left]{$\sqrt{b_{1, 0}}$}; [.$\ket{000}$
            \edge node[auto=left]{$\sqrt{b_{2, 00}}$}; $\ket{000}$
            \edge node[auto=right]{$\sqrt{1 - b_{2, 00}}$}; $\ket{001}$ ]
        \edge node[auto=right]{$\sqrt{1 - b_{1, 0}}$}; [.$\ket{010}$
            \edge node[auto=left]{$\sqrt{b_{2, 01}}$}; $\ket{010}$
            \edge node[auto=right]{$\sqrt{1-b_{2, 01}}$};$\ket{011}$ ] ]
    \edge node[auto=right]{$\sqrt{1-b_{0, \emptyset}}$}; [.$\ket{100}$
        \edge node[auto=left]{$\sqrt{b_{1, 1}}$}; [.$\ket{100}$
            \edge node[auto=left]{$\sqrt{b_{2, 10}}$}; $\ket{100}$
            \edge node[auto=right]{$\sqrt{1-b_{2, 10}}$}; $\ket{101}$ ]
        \edge node[auto=right]{$\sqrt{1-b_{1, 1}}$};[.$\ket{110}$
            \edge node[auto=left]{$\sqrt{b_{2, 11}}$}; $\ket{110}$
            \edge node[auto=right]{$\sqrt{1-b_{2, 11}}$}; $\ket{111}$ ] ]
]
\end{tikzpicture}
}
\caption{Random amplitudes procedure for 3 qubits.}
\label{fig:ramp-3-qubits}
\end{figure}

\subsection{(Pseudo)random quantum state}

Once we have the above random amplitudes state, denoted as $\ket{\xi} = \sum_{z=0}^{2^n-1} \alpha_z \ket{z}$,
we can get an ARS by applying a random phase to each $\ket{z}$.
We compute a uniformly random $u_z \in [2^\lambda]$ for each $z$ on an ancilla register, then apply the phase shift gate $P(2 \pi / 2^k)$ on $k$-th qubit of $u_z$ for all $k$, and finally uncompute $u_z$.
\begin{align*}
    \ket{\xi} \ket{0} =
    \sum_{z=0}^{2^n-1} \alpha_z \ket{z} \ket{0}
    &\to
    \sum_{z=0}^{2^n-1} \alpha_z \ket{z} \ket{u_z} \\
    &\to
    \sum_{z=0}^{2^n-1} \alpha_z \ket{z} \left(e^{2 \pi i u_z / 2^\lambda} \ket{u_z}\right)\\
    &\to 
    \sum_{z=0}^{2^n-1} \alpha_z e^{2 \pi i u_z / 2^\lambda} \ket{z} \ket{0}.
\end{align*}
To compute these $u_z$'s, we need quantum access to a random function.

In the idealized setting ($\lambda \to \infty$), the coefficient vector $(\alpha_z e^{2 \pi i u_z / 2^\lambda})_z$ is a normalized complex standard Gaussian random vector that corresponds to a Haar-random state.
However, it requires considerable effort (deferred to the Appendix) to show that, with finite precision, the trace distance is negligible.
Thus, for any adversary who queries this oracle polynomially many times, the distinguishing advantage is negligible (\Cref{thm:ars}). Finally, by replacing the random function with a quantum-secure PRF, we construct a PRS (\Cref{thm:prs}).

\subsection{(Pseudo)random function-like quantum state generator}

The next step is to modify the above ARS algorithm to accept an additional quantum state input $\ket{\phi} = \sum_x \alpha_x \ket{x}$ (More generally, the input could be $\rho = \Tr_E [\sum_x \alpha_x \ket{x}_X \ket{\theta_x}_E]$).
Now, the random function $f$ takes not only the original input but also $x$ as input. In other words, for classical input $\ket{x}$, $f(\cdot)$ is replaced by $f(x \| \cdot)$. Since the additional input state may be entangled with other registers, we need to bound the \textbf{diamond-norm distance}. We do this via the operator-norm distance between the corresponding isometries. Their block-diagonal structure makes this operator norm the maximum of the state-vector distances over $x$, so the bound has no dimension-dependent loss in the input size $m$. Let the output state of such a modified ARS procedure be $\ket{\psi_x}$ on input $\ket{x}$; then, for a superposition input $\ket{\phi}=\sum_x\alpha_x\ket{x}$, the entire output state becomes $\sum_x \alpha_x \ket{x} \ket{\psi_x}$.

Since we know that $f(x \| \cdot)$ for each $x$ is also an independent random function, each $\ket{\psi_x}$ is an independent Haar-random state, in the idealized setting. For finite precision, we choose the rounded Beta sampler's output precision independently of its statistical accuracy, so the union bound over $x$ affects only the latter and does not increase the state-generation precision parameter.
This isometry $\sum_x \alpha_x \ket{x} \mapsto \sum_x \alpha_x \ket{x} \ket{\psi_x}$ is a truly random function-like quantum state generator.
For finite precision and an imperfect sampling algorithm, the error in diamond norm distance is negligible, which bounds the maximal distinguishing probability of two quantum channels. Thus, for any adversary who queries this oracle polynomially many times, the distinguishing advantage is negligible (\Cref{thm:arfs}). Finally, by replacing the random function with a quantum-secure PRF, we construct a PRFS (\Cref{thm:prfs}).

\subsection*{Organization}
In \Cref{sec:Preliminaries}, we introduce the preliminaries and notations required for this work.
In \Cref{sec:ARS}, we show an isometric algorithm whose outputs can be arbitrarily close to Haar-random states, given (quantum) access to a truly random function.
In \Cref{sec:PRS}, we replace the truly random function with a quantum-secure pseudorandom function in the above algorithm, and achieve a scalable pseudorandom quantum state generator.
In \Cref{sec:ARFS}, we modify the above algorithm to accept an additional input state and achieve a scalable and adaptive (pseudo)random function-like quantum state generator.
In the Appendix, we show a classical algorithm for sampling from the rounded Beta distribution (with negligible statistical distance), which is used in the above algorithm, and provide proofs of some lemmas and theorems.

\section{Preliminaries}\label{sec:Preliminaries}
In this section, we introduce various notations, functions, and definitions we will require for our proofs.

\subsection{Notations}

Given $n \in \N^+$, we denote $\{0,1,\ldots,n-1\}$ by $[n]$. For $x \in \{0, 1\}^n$, we use $x[t]$ to denote the $t$-th bit of $x$, and use $x[1,\cdots, t]$ to denote the first $t$ bits of $x$. Specifically, if $t=0$, then $x[1, \cdots, t]$ is the empty string.

For a real number $x$, we use $\floor{x}$ for rounding down; i.e., the largest integer that does not exceed $x$. We use $\round{x}$ for rounding to the nearest integer (with rounding half up).
For a complex number $z$, we denote the real part and the imaginary part of $z$ by $\Re(z)$ and $\Im(z)$, respectively.

For a finite set $\mathcal{X}$, the notation $x \leftarrow \mathcal{X}$ signifies that the value $x$ is sampled uniformly randomly from $\mathcal{X}$. Given a distribution $D$ over $\mathcal{X}$, we use the notation $X \sim D$ to indicate that $X$ is a random variable distributed according to $D$.
In this paper, we always use uppercase letters for random variables and lowercase letters for their realizations/observed values, e.g., $X \sim D$ and $x \gets D$.
We use $\cU(a,b)$ for the continuous uniform distribution on $(a, b)$, and we use $X \sim \cU_n$ to denote that $X$ is a uniformly random variable over $\{0,1\}^n$.
We usually use "$:=$" for assignment, but for clarity, we use "$\gets$" when assigning the output of a probabilistic algorithm, e.g., $k \gets \KeyGen(1^\lambda)$, as this is equivalent to sampling from the distribution of outputs.

A function $\varepsilon(\lambda)$ is negligible if for all constant $c > 0$, $\varepsilon(\lambda) < \lambda^{-c}$ for large enough $\lambda$.
We denote the collection of polynomially bounded functions in $\lambda$ by $\poly(\lambda)$ and denote the collection of negligible functions in $\lambda$ by $\negl(\lambda)$. For $n(\lambda)$ a function of the security parameter $\lambda$, we may use the shorthand $n = n(\lambda)$ for brevity after it is defined.

Given sets $\cX$ and $\cY$, denote $\mathcal{Y}^{\mathcal{X}}$ as the set of all functions $\set{f : \mathcal{X} \rightarrow \mathcal{Y}}$.

Let $\cH$ be a Hilbert space. We denote $S(\mathcal{H})$ as the set of unit vectors/pure states in $\cH$. We use $L(\cH)$ to denote the set of linear operators from $\cH$ to $\cH$. We use $D(\cH)$ to denote the set of density operators acting on $\mathcal{H}$.
Given an operator $P \in L(\mathcal{H})$, its trace norm is denoted by $\norm{P}_1 \defeq \Tr\left(\sqrt{P^{\dag}P}\right)$. For a vector $\vec{v} \in \cH$, we denote its Euclidean norm as $\norm{\vec{v}}_2$.
For a linear map $A: \cH_A \to \cH_B$, its operator norm is denoted by $\norm{A}_{\mathrm{op}} \defeq \max_{\norm{\vec{v}}_2=1} \norm{A\vec{v}}_2$.

A linear map $V: \C^n \to \C^m$ is an \textit{isometry} if $V^{\dag}V = \id_n$. An isometry $V$ is a \textit{unitary} when $n = m$.

\subsection{Probability distributions}

\begin{definition}[Statistical distance]
    Given two distributions $D_X$ and $D_Y$ over a finite set $S$, denote random variables $X \sim D_X$ and $Y \sim D_Y$. The statistical distance of $D_X$ and $D_Y$ is defined as
\[
\Delta(D_X,D_Y) \defeq \frac{1}{2}\sum_{s \in S}\abs{\Pr[X = s] - \Pr[Y = s]} = \max_{T \subseteq S} \abs{\Pr[X \in T] - \Pr[Y \in T]} .
\]
\end{definition}

\begin{definition}[Normal distribution]
    Given $\mu \in \R$ and $\sigma \in \R^+$, the normal (Gaussian) distribution $\mathcal{N}(\mu, \sigma^2)$ is a real-valued continuous probability distribution, with PDF
$$
    f(x) = \frac{1}{\sqrt{2\pi}\sigma} e^{-(x-\mu)^2/(2\sigma^2)},
    x \in \R.
$$
\end{definition}

The \textit{standard normal distribution} is the special case $\mathcal{N}(0, 1)$.

\begin{fact}[Box-Muller transform \cite{BM58}] \label{fact:box-muller}
    Suppose $X \sim \cU(0,1)$ and $Y \sim \cU(0,1)$ independently. Let
$$
Z_0 = \sqrt{-2 \ln X } \cos(2 \pi Y),
$$
and
$$
Z_1 = \sqrt{-2 \ln X } \sin(2 \pi Y).
$$
Then $Z_0 \sim \cN(0,1)$ and $Z_1 \sim \cN(0,1)$ independently.
\end{fact}

\begin{fact}[Mill's Inequality {\cite[p.~123]{CB02}}]\label{fact:mills}
    Let $X \sim \cN(0,1)$, then
    \[
        \Pr[\abs{X} > t]
            \leq {\sqrt{\frac 2 \pi}} \frac{e^{- \frac {t^2} {2}}}{t}
            < \frac{e^{- \frac {t^2} {2}}}{t}.
    \]
\end{fact}

\begin{definition}[Gamma distribution]
    Given $\alpha,\theta \in \R^+$, the \textit{Gamma distribution}, denoted by $\dGamma(\alpha, \theta)$, is a continuous probability distribution with PDF
\[
    f(x) = \frac{1}{\Gamma(\alpha) \theta^\alpha} x^{\alpha - 1} e^{-x/\theta},
    x \in [0, \infty),
\]
where $\Gamma(\alpha)$ is the Gamma function, that is,
\[
    \Gamma(\alpha) = \int_0^\infty t^{\alpha - 1} e^{-t} \, dt.
\]
\end{definition}

\begin{fact}[Stirling's formula {\cite[Theorem 5.7.4]{DS12}}]\label{fact:stir}
    \[
        \Gamma(\alpha) = \sqrt{\frac {2 \pi} \alpha} \paren{\frac \alpha e}^\alpha \paren{1 + O\paren{\frac 1 \alpha}} .
    \]
\end{fact}

\begin{definition}[Beta distribution]
    Given $\alpha,\beta \in \R^+$, the Beta distribution, denoted by $\dBeta(\alpha, \beta)$, is a continuous probability distribution with PDF
\[
    f(x) = \frac{1}{\mathrm{B}(\alpha, \beta)} x^{\alpha - 1} (1 - x)^{\beta - 1},
    x \in [0, 1],
\]
where $\mathrm{B}(\alpha, \beta) = \frac{\Gamma(\alpha) \Gamma(\beta)}{\Gamma(\alpha + \beta)}$ is the Beta function.
\end{definition}

\begin{fact}[{\cite[Theorem 5.8.4]{DS12}}]\label{fact:gamma-beta}
For two random variables $X, Y$,
    \begin{align*}
    &X \sim \dGamma(\alpha, \theta), Y \sim \dGamma(\beta, \theta) \text{ independently} \\
    &\Leftrightarrow
    X/(X+Y) \sim \dBeta(\alpha, \beta), X+Y \sim \dGamma(\alpha+\beta, \theta) \text{ independently}.
    \end{align*}
\end{fact}

\begin{definition}[Chi-squared distribution]\label{def:chi-sq}
    Given $k \in \N$, the \textit{Chi-squared distribution} $\chi^2_k$ is the distribution of the sum of the squares of $k$ independent standard normal random variables.
\end{definition}

\begin{fact}[{\cite[Section 8.2]{DS12}}]\label{fact:chi-sq-gamma}
    The Chi-squared distribution is a special case of the Gamma distribution. Specifically, $\chi^2_k = \dGamma(\alpha=k/2, \theta=2)$.
\end{fact}

\begin{fact}[Transformation theorem, one-dimensional case {\cite[Theorem 3.8.4]{DS12}}]\label{fact:trans}
    Denote the PDF of the random variable $X$ as $f_X(x)$.
    For a monotonic and differentiable function $h$, the PDF of $Y = h(X)$, denoted as $f_Y(y)$, is:
    \[
        f_Y(y) = \frac {f_X(x)} {\abs{h'(x)}},
    \]
    where $x = h^{-1}(y)$, $h'(x)$ is the derivative of $h(x)$.
\end{fact}

\begin{fact}[Rejection sampling \cite{vN51}]\label{fact:rej-samp}
Rejection sampling is a well-known method to sample from a target distribution $X$ with probability density $f(x)$, given the ability to sample from another distribution $Y$ with probability density $g(x)$. Let $M$ be a constant such that $f(x) \leq M g(x)$ for all $x$, then we can sample $x$ from $Y$ and accept it with probability $\frac {f(x)} {M g(x)}$. The distribution of the accepted samples is exactly $X$, and the success probability is $\frac{1}{M}$.
\end{fact}

\subsection{Quantum cryptography}

\paragraph{Single-qubit gates.} In this paper, we use the phase shift gate
\[
    P(\theta) = \begin{bmatrix}
        1 & 0 \\
        0 & e^{i\theta}
    \end{bmatrix},
\]
and the $y$-axis rotation gate
\[
    R_Y(\theta) = \begin{bmatrix}
        \cos(\frac \theta 2) & -\sin(\frac \theta 2) \\
        \sin(\frac \theta 2) & \cos(\frac \theta 2)
    \end{bmatrix}.
\]

\begin{definition}[Trace distance]
    The trace distance of two quantum states $\rho_0, \rho_1 \in D(\mathcal{H})$ is defined as

$$ \TD (\rho_0, \rho_1) \defeq \frac 1 2 \norm{\rho_0 - \rho_1}_1 = \frac 1 2 \Tr\paren{\sqrt{(\rho_0 - \rho_1)^\dagger (\rho_0 - \rho_1)}} .$$
\end{definition}

\begin{fact}[Data processing \cite{Wil17}]\label{fact:dpi}
    Let $\rho_0,\rho_1 \in D(\mathcal{H})$ and $\Phi$ be a CPTP map. Then,
    \[
        \TD (\Phi(\rho_0), \Phi(\rho_1)) \leq \TD (\rho_0, \rho_1).
    \]
\end{fact}

\begin{fact}[{\cite[Theorem 9.3.1]{Wil17}}]\label{fact:pure-td}
    Let $\ket{\psi}, \ket{\varphi}$ be pure states. Then
    \[
        \TD (\ketbra{\psi}, \ketbra{\varphi}) = \sqrt{1 - \abs{\braket{\psi|\varphi}}^2}.
    \]
\end{fact}

\begin{fact}[{\cite[Section 9.1.4]{Wil17}}]
    Let $\rho_0, \rho_1 \in D(\mathcal{H})$ be two quantum states. For a state drawn uniformly from the set $\{\rho_0, \rho_1\}$, the optimal distinguishing probability is
    $$ \frac 1 2 + \frac 1 2 \TD(\rho_0, \rho_1). $$
\end{fact}

\begin{definition}[Diamond-norm distance {\cite[Section 9.1.6]{Wil17}}]
    Let $\cE , \cF : L(\cH_A) \to L(\cH_B)$ be quantum channels. The diamond-norm distance is defined as
    \[
        \norm{\cE - \cF}_\diamond \defeq \sup_N \max_{\rho_{R_N A}} \norm{
            (\id_{R_N} \otimes \cE_{A \to B})(\rho_{R_N A}) - (\id_{R_N} \otimes \cF_{A \to B})(\rho_{R_N A})
        }_1,
    \]
    where $\rho_{R_N A} \in D(\cH_{R_N} \otimes \cH_A)$ and $R_N$ is the reference system with dimension $N$.

    The diamond-norm distance is the measure of distinguishability between channels $\cE$ and $\cF$.
\end{definition}

\begin{definition}[Choi representation {\cite[Section 2.2.2]{Wat18}}]
    Let $\cE: L(\cH_A) \to L(\cH_B)$ be a quantum channel, where $\cH_A = \C^N$. The Choi representation of $\cE$ is the operator
    \[
        J(\cE) \defeq \sum_{i,j \in [N]} \cE(\ketbraD{i}{j}) \otimes \ketbraD{i}{j}.
    \]
\end{definition}

\begin{fact}[{\cite[Section 3.4]{Wat18}}]\label{fact:diamond-choi}
    Let $\cE, \cF: L(\cH_A) \to L(\cH_B)$ be quantum channels. Then
    \[
        \norm{\cE - \cF}_\diamond \leq \norm{J(\cE) - J(\cF)}_1.
    \]
\end{fact}

\begin{definition}[Haar measure]The Haar measure on $S(\cH)$ is the probability measure that is invariant under the action of every unitary operator,
and we denote it by $\mu(S(\cH))$. In this paper, we simplify $\mu(S(\C^{2^n}))$ to $\mu_n$, which is the Haar measure over the $n$-qubit pure states set.
\end{definition}

\begin{definition}[Haar-random states]
    A Haar-random state $\ket{\psi} \in S(\cH)$ is a quantum state chosen uniformly at random from $S(\cH)$, according to the Haar measure $\mu(S(\cH))$.
\end{definition}

\begin{definition}[Standard Gaussian random vector]\label{def:gaussian-vector}
    Let $A_i \sim \cN(0, \frac 1 2)$ and $B_i \sim \cN(0,\frac 1 2)$ independently for all $ {i \in [N]}$. We say that the complex random vector
    \[
        \vec{v} = (A_0+iB_0, A_1+iB_1, \ldots, A_{N-1}+iB_{N-1})\in \C^{N}
    \]
    is a standard (complex) Gaussian random vector.
\end{definition}

In this paper, when we say standard Gaussian random vector, we always mean the standard complex Gaussian random vector.
For simplicity, we slightly abuse terminology by also referring to the case $A_i, B_i \sim \cN(0, \sigma^2)$ for any fixed $\sigma \in \R_{>0}$ as a standard Gaussian random vector, since they become equivalent after normalization.

\begin{fact}[{\cite[Section 7.2.1]{Wat18}}]\label{fact:haar}
    Given a standard Gaussian random vector $\vec{v} \in \C^N$, the normalized random vector $\vec{v}/\norm{\vec{v}}_2$ is a Haar-random state.
\end{fact}

\begin{definition}[Quantum oracle]
    A quantum oracle $\oracle$ is a quantum map. We denote a quantum circuit $\alg$ with a black-box access to $\oracle$ by $\alg^{\oracle}$ (each invocation of $\oracle$ is counted unit time). For oracles $\oracle_1$ and $\oracle_2$ and a number $j$ we denote $\alg^{\oracle_1, \oracle_2,j}$ as a circuit which does $j$ black-box access to $\oracle_1$ followed by several black-box access to $\oracle_2$.

    For a function $f: \mathcal{X} \mapsto \mathcal{Y}$, we denote its unitary oracle
    $$U_f \defeq \sum_{x \in \mathcal{X},y \in \mathcal{Y}}  \ket{x}\ket{y \oplus f(x)} \bra{x}\bra{y}.$$
\end{definition}

\begin{definition}[Quantum-secure one-way functions \cite{Zha21}]\label{def:zha-owf}
    A one-way function is a function $F: \cX \to \cY$, where $\cX$ and $\cY$ are implicitly functions of a security parameter $\lambda$. We require that $F$ can be evaluated by a classical algorithm in time polynomial in $\lambda$.

    A one-way function $F$ is secure if no efficient quantum adversary $\alg$ can invert $F$. Formally, for any efficient quantum adversary A, there exists a negligible function $\varepsilon = \varepsilon(\lambda)$ such that
    \[
            \Pr_{x \gets \cX}[F(\alg(F(x))) = F(x)] < \varepsilon.
    \]
\end{definition}

\begin{definition}[Quantum-secure pseudorandom functions (QS-PRFs) \cite{Zha21}]\label{def:zha-prf}
    A $\prf$ is a function $\prf : \cK \times \cX \to \cY$, where $\cK$ is the key-space, and $\cX$ and $\cY$ are the domain and range. $\cK$, $\cX$ and $\cY$ are implicitly functions of the security parameter $\lambda$.
    We require that $\prf$ can be evaluated by a classical algorithm in time polynomial in $\lambda$. We write $y = \prf_k(x)$.

    A pseudorandom function $\prf$ is quantum-secure if no efficient quantum adversary $\alg$ making quantum queries can distinguish between a truly random function and the function $\prf_k$ for a random k. That is, for every such $\alg$, there exists a negligible function $\varepsilon = \varepsilon (\lambda)$ such that
    \[
        \abs{
            \Pr_{k \gets \cK}[ \alg^{U_{\prf_k}}(1^\lambda) = 1]
            -
            \Pr_{f \gets \cY^\cX}[ \alg^{U_f}(1^\lambda) = 1]
        } < \varepsilon.
        \footnote{Although in Zhandry's original definition the adversary $\alg$ takes no input, we add this input $1^\lambda$ to make it clear that an efficient quantum adversary means the running time is polynomially bounded in the security parameter $\lambda$.}
    \]
\end{definition}

\begin{definition}[Scalable quantum-secure pseudorandom function generator]\label{def:scalable-prf}
    Let $\lambda$ be the security parameter. A pair of algorithms $(\KeyGen, \Eval)$ is called a scalable adaptive quantum-secure pseudorandom function generator if for any $n(\lambda), m(\lambda) \in \poly(\lambda)$, the following holds:
    \begin{enumerate}
        \item Key generation: $\KeyGen(1^n, 1^m, 1^{\lambda})$ is a PPT algorithm which outputs a key $k$.
        \item Evaluation: $\Eval(1^n, 1^\lambda, k,x \in \{0,1\}^m)$ is PT algorithm that outputs $f_k(x) \in \{0,1\}^n$.
        \item Adaptive pseudorandomness: For any (non-uniform) QPT distinguisher $\alg$
        \begin{align*}
            \abs{
                \Pr_{k\gets \KeyGen(1^n, 1^m, 1^\lambda)}\sparen{\alg^{U_{f_k}}(1^\lambda) = 1}
                -
                \Pr_{f \gets \cY^\cX}\sparen{\alg^{U_f}(1^\lambda) = 1}
            } = \negl(\lambda) ,
        \end{align*}
        where $\cX = \{0,1\}^m, \cY = \{0,1\}^n$.
    \end{enumerate}
\end{definition}

\begin{fact}[{\cite[Corollary 1.3]{Zha21}}]\label{fact:zha-origin-prf}
    Assuming quantum-secure one-way functions (\Cref{def:zha-owf}) exist, quantum-secure pseudorandom functions (\Cref{def:zha-prf}) exist.
\end{fact}

In \cite{Zha21}, Zhandry proves \Cref{fact:zha-origin-prf} by essentially showing that the GGM construction for PRFs due to Goldreich, Goldwasser, and Micali \cite{GGM86} is quantum-secure when instantiated with a quantum-secure length-doubling PRG. The GGM construction can be adapted to yield arbitrary quantum-secure PRF with polynomially bounded input and output lengths, assuming the existence of quantum-secure polynomial-stretch PRGs. As the observation made by Aaronson and Christiano \cite{Aar09,AC12}, such PRGs exist under the assumption that quantum-secure OWFs exist. Consequently, Zhandry's result implicitly implies the following fact.

\begin{fact}[{\cite{Zha21}}]\label{fact:zha-scalable-prf}
    Assuming quantum-secure one-way functions (\Cref{def:zha-owf}) exist, scalable quantum-secure pseudorandom function generators (\Cref{def:scalable-prf}) exist.
\end{fact}


\begin{definition}[Pseudorandom quantum state (PRS's) \cite{JLS18}]
\label{def:PRSJLS}
Let $\lambda$ be the security parameter. Let $\cH$ be a Hilbert space and $\cK$ the key space, both parameterized by $\lambda$. A keyed family of quantum states $\{\ket{\phi_k} \in S(\cH)\}_{k \in \cK}$ is pseudorandom, if the following two conditions hold:
    \begin{enumerate}
        \item Efficient generation: There is a polynomial-time quantum algorithm $G$ that generates state $\ket{\phi_k}$ on input $k$. That is, for all $k \in \cK$, $G(k) = \ket{\phi_k}$.
        \item Pseudorandomness: Any polynomially many copies of $\ket{\phi_k}$ with the same random $k \in \cK$ is computationally indistinguishable from the same number of copies of a Haar-random state. More precisely, for any efficient quantum algorithm $\alg$ and any $t \in \poly(\lambda)$,
        \[
            \abs{
                \Pr_{k \gets \cK}\sparen{\alg(\ket{\phi_k}^{\otimes t}) = 1}
                -
                \Pr_{\ket{\psi} \gets \mu}\sparen{\alg(\ket{\psi}^{\otimes t}) = 1}} = \negl(\lambda),
        \]
        where $\mu$ is the Haar measure on $S(\cH)$.
    \end{enumerate}
\end{definition}

\begin{definition}[Scalable pseudorandom quantum state (PRS) generator]   \label{def:PRSour}
Let $\lambda$ be the security parameter. A pair of algorithms $(\KeyGen, \CircuitGen)$ is a scalable pseudorandom quantum state generator if for any $n(\lambda)\in \poly(\lambda)$
\footnote{
    Technically, $n(\lambda)$ can be any function in our construction. The indistinguishability still holds even if $n(\lambda)$ is superpolynomial, since the distinguisher $\alg$ is only allowed to run in polynomial time in $\lambda$.
    However, we do not discuss this case since it is irrelevant for practical cryptography schemes.
}
the following holds:
    \begin{enumerate}
        \item Key generation: $\KeyGen(1^{n}, 1^\lambda)$ is a PPT algorithm  which outputs a key $k$.
        \item Circuit generation: $\CircuitGen(1^n, 1^\lambda, k)$ is a PT algorithm whose output is a quantum circuit for an isometry $V_k: \C \to \C^{2^{n}}$.
        \item Pseudorandomness: For any (non-uniform) QPT distinguisher $\alg$,
        \[
            \abs{
                \Pr_{k\gets \KeyGen(1^n, 1^\lambda)}
                \left[\alg^{V_k}(1^\lambda)=1\right]
                -
                \Pr_{\ket{\psi} \gets \mu_n}\left[\alg^{W_{\ket{\psi}}}(1^\lambda)=1\right]
            } = \negl(\lambda),
        \]
    \end{enumerate}
    where $W_{\ket{\psi}}: \C \to \C^{2^n}$ is the isometry $W_{\ket{\psi}} \defeq \ket{\psi}$.

\end{definition}

\begin{claim}
    Let $(\KeyGen, \CircuitGen)$ be a scalable pseudorandom quantum state generator as in \Cref{def:PRSour}, then for any $n(\lambda) \in \poly(\lambda)$, it gives a family of $n$-qubit pseudorandom quantum states as in \Cref{def:PRSJLS}.
\end{claim}
\vspace{-0.2em}The proof is provided in \Cref{sec:scalable-prs-proof}.

\begin{definition}[Scalable and adaptive pseudorandom function-like quantum state (PRFS) generator]\label{def:prfs}
    Let $\lambda$ be the security parameter.
    A pair of algorithms $(\KeyGen,\CircuitGen)$ is a scalable adaptive pseudorandom function-like state generator if for any $n(\lambda), m(\lambda) \in \poly(\lambda)$, the following holds:

    \begin{enumerate}
        \item Key generation: $\KeyGen(1^n, 1^m, 1^{\lambda})$ is a PPT algorithm which outputs a key $k$.
        \item Circuit generation: $\CircuitGen(1^n, 1^m, 1^\lambda, k)$ is a PT algorithm whose output is a quantum circuit for an isometry $V_k: \C^{2^m} \to \C^{2^{n+m}}$.
        \item Adaptive pseudorandomness: For any (non-uniform) QPT distinguisher $\alg$,
    \begin{equation*}
        \abs{
            \Pr_{\substack{k \gets \KeyGen(1^n, 1^m, 1^{\lambda})}}
                \sparen{\alg^{V_k}(1^{\lambda}) = 1}
            -
            \Pr_{\{\ket{\psi_x} \gets \mu_n\}_{x \in \{0,1\}^m}}
                \sparen{\alg^{W_{\left\{\ket{\psi_x}\right\}_{x}}}(1^{\lambda}) = 1}
        }
        = \negl(\lambda),
    \end{equation*}
    where $W_{\left\{\ket{\psi_x}\right\}_{x}}:\C^{2^m} \to \C^{2^{n+m}}$ is the isometry
    $$ W_{\left\{\ket{\psi_x}\right\}_{x}} \defeq  \sum_{x=0}^{2^m-1}  (\ket{x} \ket{\psi_{x}}) \bra{x}.$$
    \end{enumerate}
\end{definition}

\section{Isometric asymptotically random state generator}\label{sec:ARS}

In this section, we show an isometric procedure which, given quantum oracle access to a random function, efficiently generates a quantum state that can be arbitrarily close to a Haar-random state statistically. More specifically, we show how to construct a unitary $\operatorname{RS}^{U_f}_{n, \lambda}$ such that $\{\operatorname{RS}^{U_f}_{n, \lambda}\ket{0}^{\otimes n}\}_f$ is statistically close to Haar-random states , with $\negl(\lambda)$ distance.

The following theorem, as the main result of this section, is a direct corollary of \Cref{thm:ars}.

\begin{theorem} \label{thm:ars-main}
    For every $\lambda \in \N^+$ security parameter, $n(\lambda) \in \poly(\lambda)$ number of qubits,
    there exists a unitary $\operatorname{RS}^{U_f}_{n, \lambda}$ which can be implemented efficiently given quantum oracle access to $U_f$,
    where $f: \cX \rightarrow \cY$,
    $\cX = \{0, 1\}^{n+1}$ and $\cY = \{0,1\}^{\poly(n + \lambda)}$.
    For any (non-uniform, quantum-output) quantum algorithm $\alg$ that can access an oracle at most $\poly(\lambda)$ times,
    \[
        \TD \left (
            \ex_{f \gets \cY^\cX}
                \sparen{\alg^{V_f}(1^\lambda)},
            \ex_{\ket{\psi} \gets \mu_n}
                \sparen{\alg^{W_{\ket{\psi}}}(1^\lambda)}
            \right )
        = \negl(\lambda) ,
    \]
    where
    \begin{itemize}
        \item $V_f: \C \to \C^{2^n}$ is the isometry $V_f \defeq \operatorname{RS}_{n, \lambda}^{U_f} \ket{0}^{\otimes n}$.
        \item $W_{\ket{\psi}}: \C \to \C^{2^n}$ is the isometry $W_{\ket{\psi}} \defeq \ket{\psi}$.
    \end{itemize}
\end{theorem}

\subsection{Random amplitudes procedure}

Consider a standard Gaussian random vector (\Cref{def:gaussian-vector})
\begin{align*}
    v = (A_0+i B_0, A_1+i B_1, \ldots, A_{N-1}+i B_{N-1})\in \C^{N},
\end{align*}
where $A_i, B_i \sim \cN(0, 1)$ independently.
The squared norm of each item $C_i = A_i^2 + B_i^2$ follows $\chi_2^2 = \dGamma(1, 2)$ (\Cref{fact:chi-sq-gamma}).

\begin{definition}[Random amplitudes quantum state]\label{def:ramp-state}
    Let random variables $C_i \sim \chi_2^2$ i.i.d. The random state
    \[
        \ket{\xi} = \frac 1 {\sqrt{\sum_{i=0}^{N-1} C_i}} \sum_{i=0}^{N-1} \sqrt{C_i} \ket{i}
    \]
    is called a random amplitudes quantum state.
\end{definition}

The first and most critical step in our ARS / PRS algorithm is to prepare such a random amplitudes quantum state.

\begin{lemma}\label{lma:beta-ramp}
    For any $n \in \N^+$, given $B_{t, z} \sim \dBeta(2^{n-t-1}, 2^{n-t-1})$ all independent, where $t \in [n], z \in \{0, 1\}^t$, the state prepared by the following procedure is an $n$-qubit random amplitudes quantum state.
    \begin{enumerate}
        \item Initialize with $\ket{0}^{\otimes n}$.
        \item For $k := 1 \text{ to } n$
        \begin{itemize}
            \item Map the first $k$ qubits from $\ket{z}\ket{0}$ to $\ket{z}(\sqrt{B_{k-1, z}} \ket{0} + \sqrt{1-B_{k-1, z}}\ket{1})$ for each $z \in \{0, 1\}^{k-1}$. i.e., apply the unitary
            $$U_{k}: \ket{z}\ket{0} \mapsto \ket{z}(\sqrt{B_{k-1, z}} \ket{0} + \sqrt{1-B_{k-1, z}}\ket{1})$$
            on the first $k$ qubits.
        \end{itemize}
    \end{enumerate}
    (Earlier, we have shown the example for $n=3$ in \Cref{fig:ramp-3-qubits}.)
\end{lemma}
\vspace{-0.2em}The proof is provided in \Cref{sec:ram-proof}.

Let $\cX = \{0, 1\}^{n+1}$ and $\cY = \{0,1\}^{\poly(n + \lambda)}$.
Given the unitary oracle $U_f: \ket{x}\ket{y} \mapsto \ket{x}\ket{y \oplus f(x)}$ corresponding to a classical random function $f: \cX \to \cY$, we now show a unitary algorithm $\operatorname{RA}_{n, \lambda}^{U_f}$ to implement the above random amplitudes procedure efficiently. That is, $\ket{\xi_f} = \operatorname{RA}_{n, \lambda}^{U_f} \ket{0}^{\otimes n}$ is ($\negl(\lambda)$-close to) a random amplitudes quantum state when $f \gets \cY^\cX$.

Define $B_f^{(t)}$ to be the unitary that computes (with $(n+\lambda)$-bit precision)
\[
    \theta_{t, z} = \frac {\arccos \sqrt{\alg_{RB}(1^{n+\lambda}, 2^{n-t-1}, f(2^t + z))}} {2 \pi}
    \in (0, 1),
\]
where $\alg_{RB}$ is the sampling algorithm in \Cref{lma:smp-beta} to sample from $\dBeta_{R(2^{-n-\lambda})}(2^{n-t-1}, 2^{n-t-1})$. That is,
\[
    B_f^{(t)}: \ket{z}\ket{0} \mapsto \ket{z} \ket{\theta_{t, z}}.
\]

\begin{algorithm}[H]
    \caption{Random amplitudes procedure $\operatorname{RA}_{n, \lambda}^{U_f}$}\label{alg:ramp}

    \begin{algorithmic}[1]
        \INPUT $n$ qubits, denoted by $q_i$
        \State Prepare an $(n+\lambda)$-qubit ancilla register $q_r$, initialized to $\ket{0}^{\otimes (n+\lambda)}$
        \For {$t := 0$ to $n-1$}
            \State Apply $B_f^{(t)}$ on $q_i[1, \cdots, t]$ and $q_r$
            \State Apply controlled rotation on $q_i[t+1]$, controlled on $q_r$
            \State Apply $B_f^{(t)}$ on $q_i[1, \cdots, t]$ and $q_r$
        \EndFor
        \State Output $q_i$.
    \end{algorithmic}
\end{algorithm}

The implementation details of $B_f^{(t)}$ from $U_f$ are shown in \Cref{fig:bt}.
The implementation details of controlled rotation $R: \ket{0}\ket{\theta} \mapsto (\cos (2 \pi\theta) \ket{0} + \sin (2 \pi \theta) \ket{1})\ket{\theta}$ are shown in \Cref{fig:ur}. Notice that each controlled-$R_Y(2 \pi / 2^i)$ gate will rotate $2 \pi / 2^{i+1}$ if the corresponding bit of $\theta$ is $1$, so $R$ indeed rotates $2 \pi \theta$ in the end.
The entire procedure for \Cref{alg:ramp} is shown in \Cref{fig:ARS-ramp}.

In each iteration, the second $B_f^{(t)}$ uncomputes $q_r$, so $\operatorname{RA}_{n, \lambda}^{U_f}$ is indeed a unitary.

\begin{figure}[ht]
    \centering
    \scalebox{\circuitscale}{
    \begin{quantikz}[transparent]
        \lstick{$\ket{z}$} & \qwbundle{t}  &[0.75cm] \gate[4][2.75cm]{U_f}\gateinput[3]{$2^t+z$}\gateoutput[3]{$2^t+z$} & & \gate[4][2.75cm]{U_f} & \rstick{$\ket{z}$} \\
        \lstick{$\ket{1}$} & & & & & \rstick{$\ket{1}$} \\
        \lstick{$\ket{0}$} & \qwbundle{n-t} & & & & \rstick{$\ket{0}$} \\
        \lstick{$\ket{0}$} & \qwbundle{\poly(n + \lambda)} & \gateinput{$0$}\gateoutput{$f(2^t + z)$} & \gate[2][6cm]{\frac{\arccos \sqrt{\alg_{RB}(1^{n + \lambda}, 2^{n-t-1}, \cdot)}} {2 \pi}} & & \rstick{$\ket{0}$} \\
        \lstick{$\ket{0}$} & \qwbundle{n+\lambda} & & \gateinput{$0$}\gateoutput{$\theta_{t, z}$} & & \rstick{$\ket{\theta_{t, z}}$} \\
    \end{quantikz}
    }
    \caption{
        $B_f^{(t)}: \ket{z}\ket{0} \mapsto \ket{z}\ket{\theta_{t, z}}$
    }
    \label{fig:bt}
\end{figure}

\begin{figure}[ht]
    \centering
    \scalebox{\circuitscale}{
    \begin{quantikz}[transparent]
        \lstick[4]{$\ket{\theta}$} & & \ctrl{4} & & \ \ldots \ & & \\
        & & & \ctrl{3} & \ \ldots \ & & \\
        & \lstick{\vdots} \setwiretype{n} \\
        & & & & \ \ldots \ & \ctrl{1} & \\
        \lstick{$\ket{0}$} & & \gate{R_Y(\pi / 2^{\kappa - 2})} & \gate{R_Y(\pi / 2^{\kappa - 1})} & \ \ldots \ & \gate{R_Y(2 \pi)} & \\
    \end{quantikz}
    }
    \caption{
     Controlled rotation $R: \ket{\theta}\ket{0} \mapsto \ket{\theta}(\cos (2 \pi \theta) \ket{0} + \sin(2 \pi \theta)\ket{1})$,
        where $\theta \in (0, 1)$, presented by $\kappa$ qubits, is the angle of turn.
    }
    \label{fig:ur}
\end{figure}

\begin{figure}[ht]
    \centering
    \scalebox{\circuitscale}{
    \begin{quantikz}[transparent]
    \lstick{$\ket{0}$} & &[0.1cm] & & & & & & \ \ldots\  & & \gate{R} & & \rstick[5]{$\ket{\xi}$}\\
    \lstick{$\ket{0}$} & & & & & & & & \ \ldots\  & \gate[5][1.4cm]{B_f^{(n-1)}} & & \gate[5][1.4cm]{B_f^{(n-1)}} & \\
    & \lstick{\vdots} \setwiretype{n}  \\[0.5cm]
    \lstick{$\ket{0}$} & & & & & & \gate{R} & & \ \ldots\  & & & &\\
    \lstick{$\ket{0}$} & & & \gate{R} & & \gate[2][1.4cm]{B_f^{(1)}} & & \gate[2][1.4cm]{B_f^{(1)}} & \ \ldots\  & & & & \\
    \lstick{$\ket{0}$} & \qwbundle{n+\lambda} &
        \gate[1][1.4cm]{B_f^{(0)}}\gateinput{$0$}\gateoutput{$\theta$} & \ctrl{-1} & \gate[1][1.4cm]{B_f^{(0)}}\gateinput{$\theta$}\gateoutput{$0$} &
        \gateinput{$0$}\gateoutput{$\theta$} & \ctrl{-2} & \gateinput{$\theta$}\gateoutput{$0$} &
        \ \ldots \
        & \gateinput{$0$}\gateoutput{$\theta$} & \ctrl{-5} & \gateinput{$\theta$}\gateoutput{$0$} &
        \rstick{$\ket{0}$}\\
    \end{quantikz}
    }
    \caption{The random amplitudes procedure for ARS/PRS. }
    \label{fig:ARS-ramp}
\end{figure}

\subsection{Random phases procedure}

Given quantum oracle access to a random classical function $f: \{0, 1\}^{n+1} \to \{0, 1\}^{\poly(n + \lambda)}$, now we know how to prepare a random amplitudes state, denoted as $\ket{\xi} = \sum_{z=0}^{2^n-1} \alpha_z \ket{z}$. The next step is to apply an independent random phase for each $\alpha_z \ket{z}$.

\begin{lemma}\label{lma:phase-gaussian}
    Let $C \sim \chi_2^2, U \sim \cU(0, 1)$ independently, then both real and imaginary parts of $\sqrt{C} e^{2\pi i U}$ follow standard normal distribution independently
\end{lemma}
\vspace{-0.2em}The proof is provided in \Cref{sec:phase-gaussian-proof}.

We now show a unitary algorithm $\operatorname{RP}_{n, \lambda}^{U_f}$ to implement the random phases procedure efficiently.
To simplify the notation,
for $f: \{0, 1\}^{n+1} \to \{0, 1\}^{\poly(n + \lambda)}$ we define the function $\check{f}: \{0, 1\}^n \to \{0, 1\}^\lambda$ as $$\check{f}(z) = f(2^n + z)[1, \dots, \lambda].$$
Here we add $2^n$ to ensure that the input of $f$ is different from what has already been used in the random amplitudes procedure.

\begin{algorithm}[H]
    \caption{Random phases procedure $\operatorname{RP}_{n, \lambda}^{U_f}$}\label{alg:rp}
    \begin{algorithmic}[1]
        \INPUT $n$ qubits, denoted by $q_i$
        \State Prepare a $\lambda$-qubit ancilla register $q_r$, initialized to $\ket{0}^{\otimes \lambda}$
        \State Apply $U_{\check{f}}$ on $q_i$ and $q_r$.
        \For {$k := 1$ to $\lambda$}
            \State Apply phase shift gate $P(2 \pi / 2^k)$ on $q_r[k]$.
        \EndFor
        \State Apply $U_{\check{f}}$ on $q_i$ and $q_r$.
        \State Output $q_i$.
    \end{algorithmic}
\end{algorithm}

Let the input be $\sum_{z=0}^{2^n-1} \alpha_z \ket{z}$. After applying the first $U_{\check{f}}$, the state is
\(
    \sum_{z=0}^{2^n-1} \alpha_z \ket{z} \ket{\check{f}(z)}.
\)
Since each phase shift gate $P(2 \pi / 2^k)$ applies a phase $2 \pi / 2^k$ to $\ket{\check{f}(z)}$ when the $k$-th bit of $\check{f}(z)$ is $1$, in total they apply a phase $2 \pi \check{f}(z) / 2^\lambda$ to $\ket{\check{f}(z)}$,
and the entire state becomes
\(
    \sum_{z=0}^{2^n-1} \alpha_x \ket{z} e^{2 \pi i \check{f}(z) / 2^\lambda} \ket{\check{f}(z)}.
\)
After uncomputing the ancilla register in the last step, the output is
\(
    \sum_{z=0}^{2^n-1} \alpha_z e^{2 \pi i \check{f}(z) / 2^\lambda} \ket{z}.
\)

\subsection{Asymptotically random states}

By combining the random amplitudes procedure \Cref{alg:ramp} and the random phases procedure \Cref{alg:rp}, we define our unitary random state procedure
\begin{equation}
    \operatorname{RS}_{n, \lambda}^{U_f} \defeq \operatorname{RP}_{n, \lambda}^{U_f}  \operatorname{RA}_{n, \lambda}^{U_f}. \label{random-u}
\end{equation}

\Cref{fact:haar}, \Cref{lma:beta-ramp}, and \Cref{lma:phase-gaussian} imply that the ideal output state (infinite precision and a perfect sampling algorithm) of the random amplitudes procedure followed by the random phases procedure is a Haar-random state. We now show that the distance is negligible in $\lambda$ when considering the actual algorithm output.

\begin{definition}\label{def:state-b-u}
    Let $n \in \N^+$ be the number of qubits. For parameters $\bB = (b_{t, z})_{t \in [n], z \in \{0, 1\}^t}$ and $\bU = (u_z)_{z \in \{0, 1\}^n}$, where $b_{t, z} \in [0, 1]$ and $u_z \in [0, 1)$, we define $\ket{\psi_{\bB, \bU}}$ as the quantum state prepared by the following procedure:
    \begin{enumerate}
        \item Initialize with $\ket{0}^{\otimes n}$.
        \item For $t := 0 \text{ to } n-1$
        \begin{itemize}
            \item Map the first $t+1$ qubits from $\ket{z}\ket{0}$ to $\ket{z}(\sqrt{b_{t, z}} \ket{0} + \sqrt{1-b_{t, z}}\ket{1})$ for each $z \in \{0, 1\}^{t}$. i.e., apply the unitary
            $$U^{(t)}: \ket{z}\ket{0} \mapsto \ket{z}(\sqrt{b_{t, z}} \ket{0} + \sqrt{1-b_{t, z}}\ket{1})$$
            on the first $t+1$ qubits.
        \end{itemize}
        \item Map $\ket{z}$ to $e^{2\pi i u_z} \ket{z}$ for each $z \in \{0, 1\}^n$. i.e., apply the unitary
        $$U_P: \ket{z} \mapsto e^{2 \pi i u_z} \ket{z}$$
        on all $n$ qubits.
    \end{enumerate}
\end{definition}

\begin{lemma}\label{lma:state-b-u-close}
    Let $n \in \N^+$ the number of qubits. If
    \begin{align*}
        &\bB = (b_{t, z})_{t \in [n], z \in \{0, 1\}^t}, \bU = (u_z)_{z \in \{0, 1\}^n}, \\
        &\bB' = (b'_{t, z})_{t \in [n], z \in \{0, 1\}^t}, \bU' = (u'_z)_{z \in \{0, 1\}^n}
    \end{align*}
    satisfy $b_{t, z}, b'_{t, z} \in [0, 1] ,u_z, u'_z \in [0, 1)$, $\abs{b_{t, z} - b'_{t, z}} < \delta_1$ and $\abs{u_z - u'_z} < \delta_2$ for all $t, z$, then
    \[
        \TD(\ketbra{\psi_{\bB, \bU}}, \ketbra{\psi_{\bB', \bU'}}) < \sqrt{2^n n \delta_1} + 2 \pi \delta_2.
    \]
\end{lemma}
\begin{proof}
        This result follows directly from \Cref{lma:iso-b-u-close} with $m=0$.
\end{proof}

\begin{theorem} \label{thm:ars}
    For every $n \in \N^+$ number of qubits, $\lambda \in \N^+$ security parameter,
    let $\cX = \{0, 1\}^{n+1}$, $\cY = \{0,1\}^{\poly(n + \lambda)}$,
    and $\operatorname{RS}_{n, \lambda}^{U_f}$ be defined as in \cref{random-u}.
    For any (non-uniform, quantum-output) quantum algorithm $\alg$ that can access an oracle at most $l$ times,
    \[
        \TD \left (
            \ex_{f \gets \cY^\cX}
                \sparen{\alg^{V_f}(1^\lambda)},
            \ex_{\ket{\psi} \gets \mu_n}
                \sparen{\alg^{W_{\ket{\psi}}}(1^\lambda)}
            \right )
        < (1 + \sqrt{2 \pi}) l \sqrt{n} 2^{-\frac \lambda 2} + 2 \pi l 2^{-\lambda} + 2^{-\lambda} ,
    \]
    where
    \begin{itemize}
        \item $V_f: \C \to \C^{2^n}$ is the isometry $V_f \defeq \operatorname{RS}_{n, \lambda}^{U_f} \ket{0}^{\otimes n}$.
        \item $W_{\ket{\psi}}: \C \to \C^{2^n}$ is the isometry $W_{\ket{\psi}} \defeq \ket{\psi}$.
    \end{itemize}
\end{theorem}

\begin{proof}
    This result follows directly from \Cref{thm:arfs} with $m=0$, using \cref{eq:rs-m-zero}.
    For clarity, we also provide an independent proof in \Cref{sec:ars-proof}.
\end{proof}

\section{Pseudorandom quantum state generators}\label{sec:PRS}

\Cref{thm:ars} shows an ARS generator which generates asymptotically random states from real random functions. By replacing the real random functions with quantum-secure pseudorandom functions, we get a pseudorandom quantum state generator.

\begin{theorem}\label{thm:prs}
    Let $(\KeyGen_\prf, \Eval_\prf)$ be a scalable quantum-secure pseudorandom function generator as in \Cref{def:scalable-prf}. The following is a scalable pseudorandom quantum state generator:
    \begin{itemize}
        \item $\KeyGen(1^n, 1^\lambda)$: Let $k \gets \KeyGen_\prf (1^{\poly(n + \lambda)}, 1^{n+1}, 1^\lambda)$. Output $k$.
        \item $\CircuitGen(1^n, 1^\lambda, k)$: Let function $\prf_k: \{0, 1\}^{n+1} \to \{0, 1\}^{\poly(n + \lambda)}$ be defined as $$\prf_k(\cdot)\defeq \Eval_\prf(1^{\poly(n + \lambda)}, 1^\lambda, k, \cdot).$$
        Output the circuit for $\operatorname{RS}_{n, \lambda}^{U_{\prf_k}} \ket{0}^{\otimes n}$ where $\operatorname{RS}_{n, \lambda}^{U_{\prf_k}}$ is defined in \cref{random-u}.
    \end{itemize}
\end{theorem}
This result follows by replacing the random function in \Cref{thm:ars} with a quantum-secure pseudorandom function, which is a standard technique. A full proof is provided in \Cref{sec:ar_to_pr_proof}.

\section{(Pseudo)random function-like quantum state generators}\label{sec:ARFS}

In this section, we modify our ARS algorithm $\operatorname{RS}_{n, \lambda}^{U_f}$ (\cref{random-u}) to accept an additional input state, so that we can get the construction for ARFS / PRFS.

For a classical function $f: \{0, 1\}^{n+m+1} \to \{0, 1\}^{\poly(n+m+\lambda)}$, we define the quantum oracle $U_f^{\gets m}$ as follows: Prepare $U_f$ and keep an $m$-qubit register. Each time the algorithm $\operatorname{RS}_{n, \lambda}^{U_f^{\gets m}}$ uses this oracle, this register is used as the first $m$ qubits, and $\operatorname{RS}_{n, \lambda}^{U_f^{\gets m}}$ only provides the rest $n+1$ qubits (see \Cref{fig:hidem} with $a=n+m+1$ and $b=\poly(n+m+\lambda)$).

\begin{figure}[h!]
    \centering
    \scalebox{\circuitscale}{
    \begin{quantikz}[transparent]
        & \qwbundle{m} && \gate[3][2.7cm]{U_f}\gateinput{$x$}\gateoutput{$x$}
        &&& \gate[3][2.7cm]{U_f}\gateinput{$x$}\gateoutput{$x$}
        &&& \gate[3][2.7cm]{U_f}\gateinput{$x$}\gateoutput{$x$} &  \\
        & \qwbundle{a-m} &&\gateinput{$z_1$}\gateoutput{$z_1$}
        &&\wireoverride{n}&\gateinput{$z_2$}\gateoutput{$z_2$}
        &&\wireoverride{n}&\gateinput{$z_3$}\gateoutput{$z_3$} & \\
        & \qwbundle{b} &&\gateinput{$y_1$}\gateoutput{$y_1 \oplus f(x||z_1)$}
        &&\wireoverride{n}&\gateinput{$y_2$}\gateoutput{$y_2 \oplus f(x||z_2)$}
        &&\wireoverride{n}&\gateinput{$y_3$}\gateoutput{$y_3 \oplus f(x||z_3)$}& \\
    \end{quantikz}
    }
    \caption{How quantum oracle $U_f^{\gets m}$ works, where $f: \{0, 1\}^a \to \{0, 1\}^b$.}
    \label{fig:hidem}
\end{figure}

Similarly to $\operatorname{RS}_{n, \lambda}^{U_f}$, we define
\begin{equation}
    \operatorname{RS}_{n, \lambda}^{U_f^{\gets m}} \defeq \operatorname{RP}_{n, \lambda}^{U_f^{\gets m}}  \operatorname{RA}_{n, \lambda}^{U_f^{\gets m}}. \label{random-f-u}
\end{equation}

For the random amplitudes procedure, define $B_f'^{(t)}$ to compute
\[
    \theta_{t,z}^{(x)}=
    \frac{\arccos\sqrt{\alg_{RB}(1^{n+\lambda},1^{n+m+\lambda},2^{n-t-1},f(x\cdot 2^{n+1}+2^t+z))}}{2\pi}
\]
with $(n+\lambda)$-bit precision. The first parameter of $\alg_{RB}$ determines the precision of the rounded Beta sample, while the second makes its statistical distance at most $2^{-n-m-\lambda}$; thus, the latter does not enlarge the angle register. The algorithm $\operatorname{RS}_{n, \lambda}^{U_f^{\gets m}}$ is still a unitary (on $n+m$ qubits). Its random amplitudes procedure $\operatorname{RA}_{n,\lambda}^{U_f^{\gets m}}$ is shown in \Cref{fig:ARFS-ramp}. After this, we apply random phases as before (i.e. $\operatorname{RP}_{n,\lambda}^{U_f^{\gets m}}$).

When $m=0$, the input register is trivial and $U_f^{\gets 0}=U_f$. Moreover, $x=0$, and the equality of the two sampler interfaces in \Cref{lma:smp-beta} gives $B_f'^{(t)}=B_f^{(t)}$. Consequently,
\begin{equation}
    \operatorname{RS}_{n,\lambda}^{U_f^{\gets 0}}=\operatorname{RS}_{n,\lambda}^{U_f}.
    \label{eq:rs-m-zero}
\end{equation}

\begin{figure}[H]
    \centering
    \scalebox{\circuitscale}{
    \begin{quantikz}[transparent]
    \lstick{$\ket{0}$} & &[0.1cm] & & & & & & \ \ldots\  & & \gate{R} & & \rstick[5]{$\ket{\xi_x}$}\\
    \lstick{$\ket{0}$} & & & & & & & & \ \ldots\  & \gate[6][1.4cm]{B_f'^{(n-1)}} & & \gate[6][1.4cm]{B_f'^{(n-1)}} & \\
    & \lstick{\vdots} \setwiretype{n}  \\[0.5cm]
    \lstick{$\ket{0}$} & & & & & & \gate{R} & & \ \ldots\  & & & &\\
    \lstick{$\ket{0}$} & & & \gate{R} & & \gate[3][1.4cm]{B_f'^{(1)}} & & \gate[3][1.4cm]{B_f'^{(1)}} & \ \ldots\  & & & & \\
    \lstick{$\ket{x}$} & \qwbundle{m} & \gate[2][1.4cm]{B_f'^{(0)}} & & \gate[2][1.4cm]{B_f'^{(0)}} & & & & \ \ldots \ & & & & \rstick{$\ket{x}$}\\
    \lstick{$\ket{0}$} & \qwbundle{n+\lambda} &
        \gateinput{$0$}\gateoutput{$\theta$} & \ctrl{-2} & \gateinput{$\theta$}\gateoutput{$0$} &
        \gateinput{$0$}\gateoutput{$\theta$} & \ctrl{-3} & \gateinput{$\theta$}\gateoutput{$0$} &
        \ \ldots \
        & \gateinput{$0$}\gateoutput{$\theta$} & \ctrl{-6} & \gateinput{$\theta$}\gateoutput{$0$} &
        \rstick{$\ket{0}$}\\
    \end{quantikz}
    }
    \caption{The random amplitudes procedure $\operatorname{RA}_{n,\lambda}^{U_f^{\gets m}}$ for ARFS/PRFS, where $B_f'^{(t)}$ takes the additional $m$-qubit input register into account.}
    \label{fig:ARFS-ramp}
\end{figure}

\begin{definition}\label{def:iso-b-u}
    Let $n \in \N^+$ and $m \in \N$. For parameters $\bB = (b_{t, z}^{(x)})_{t \in [n], z \in \{0, 1\}^t, x \in \{0, 1\}^m}$ and $\bU = (u_z^{(x)})_{z \in \{0, 1\}^n, x \in \{0, 1\}^m}$, where $b_{t, z}^{(x)} \in [0, 1]$ and $u_z^{(x)} \in [0, 1)$, we define isometry $V_{\bB, \bU}: \C^{2^m} \to \C^{2^{n+m}}$ as
    \[
        V_{\bB, \bU} \defeq \sum_{x=0}^{2^m-1} (\ket{x} \ket{\psi_{\bB_x, \bU_x}}) \bra{x},
    \]
    where $\bB_x = (b_{t, z}^{(x)})_{t \in [n], z \in \{0, 1\}^t}$, $\bU_x = (u_z^{(x)})_{z \in \{0, 1\}^n}$, and $\ket{\psi_{\bB_x, \bU_x}}$ is defined in \Cref{def:state-b-u}. We use $\cE_{\bB, \bU}$ to denote $V_{\bB, \bU}$ as a channel, i.e., $\cE_{\bB, \bU}(\rho) = V_{\bB, \bU} \rho V_{\bB, \bU}^\dagger$.
\end{definition}

\begin{lemma}\label{lma:iso-b-u-close}
    Let $n \in \N^+$ and $m \in \N$.
    If
    \begin{align*}
        &\bB = (b_{t, z}^{(x)})_{t \in [n], z \in \{0, 1\}^t, x \in \{0, 1\}^m}, \bU = (u_z^{(x)})_{z \in \{0, 1\}^n, x \in \{0, 1\}^m},\\
        &\bB' = ({b'}_{t, z}^{(x)})_{t \in [n], z \in \{0, 1\}^t, x \in \{0, 1\}^m}, \bU' = ({u'}_z^{(x)})_{z \in \{0, 1\}^n, x \in \{0, 1\}^m}
    \end{align*}
    satisfy $b_{t, z}^{(x)}, {b'}_{t, z}^{(x)} \in [0, 1], u_z^{(x)}, {u'}_z^{(x)} \in [0, 1)$, $\abs{b_{t, z}^{(x)} - {b'}_{t, z}^{(x)}} < \delta_1$ and $\abs{u_z^{(x)} - {u'}_z^{(x)}} < \delta_2$ for all $t, z, x$, then the following bound holds:
    \footnote{In the earlier conference version, the bound in this lemma was obtained through the trace norm of the Choi representation and incurred a factor of $2^m$. Although this loss could be absorbed using the scalability of our PRS construction, the present proof avoids it by bounding the diamond-norm distance via the operator-norm distance.}
    \begin{align*}
        \frac 1 2 \norm{\cE_{\bB, \bU} - \cE_{\bB', \bU'}}_\diamond < \sqrt{2^n n \delta_1} + 2 \pi \delta_2.
    \end{align*}
\end{lemma}
\vspace{-0.2em}The proof is provided in \Cref{sec:diamond-proof}.

\begin{theorem}\label{thm:arfs}
    For every $n,\lambda \in \N^+$ and $m\in\N$, let $\cX = \{0, 1\}^{n+m+1}$ and $\cY = \{0,1\}^{\poly(n + m + \lambda)}$.
    With $\operatorname{RS}_{n, \lambda}^{U_f^{\gets m}}$ defined as in \cref{random-f-u},
    for any (non-uniform, quantum-output) quantum algorithm $\alg$ that can access an oracle at most $l$ times,
    \begin{align*}
        \TD \left (
            \ex_{f \gets \cY^{\cX}}
                \sparen{\alg^{V_f}(1^\lambda)},
            \ex_{\{\ket{\psi_x} \gets \mu_n\}_{x \in \{0,1\}^m}}
                \sparen{\alg^{W_{\left\{\ket{\psi_x}\right\}_{x}}}(1^{\lambda})}
            \right ) \\
        < (1 + \sqrt{2 \pi}) l \sqrt{n} 2^{-\frac \lambda 2} + 2 \pi l 2^{-\lambda} + 2^{-\lambda},
    \end{align*}
    where
    \begin{itemize}
        \item $V_f: \C^{2^m} \to \C^{2^{n+m}}$ is the isometry
        $$
            V_f \defeq \operatorname{RS}_{n, \lambda}^{U_f^{\gets m}} (I_m \otimes \ket{0}^{\otimes n}).
        $$
        \item $W_{\left\{\ket{\psi_x}\right\}_{x}}:\C^{2^m} \to \C^{2^{n+m}}$ is the isometry
        $$
            W_{\left\{\ket{\psi_x}\right\}_{x}} \defeq \sum_{x=0}^{2^m-1}  (\ket{x} \ket{\psi_{x}}) \bra{x}.
        $$
    \end{itemize}
\end{theorem}

\begin{proof}
    Let $\varepsilon_1 = 2^{-n-\lambda}$ and $\varepsilon_2 = 2^{-\lambda}$.
    We use $\fB$ to denote the joint distribution of random variables $(B_{t, z}^{(x)})_{t \in [n], z \in \{0, 1\}^t, x \in \{0, 1\}^m}$ where $B_{t, z}^{(x)} \sim \dBeta(2^{n-t-1}, 2^{n-t-1})$ independently. That is, $\bB \gets \fB$ is short for $$\bB = \left(b_{t, z}^{(x)} \gets \dBeta(2^{n-t-1}, 2^{n-t-1}) \right)_{t \in [n], z \in \{0, 1\}^t, x \in \{0, 1\}^m}.$$
    Similarly, we use $\fU$ to denote the joint distribution of random variables $(U_z^{(x)})_{z \in \{0, 1\}^n, x \in \{0, 1\}^m}$ where $U_z^{(x)} \sim \cU(0, 1)$ independently.
    We use notations $\bB_x$, $\bU_x$, $V_{\bB, \bU}$ and $\cE_{\bB, \bU}$ for given $\bB, \bU$ as in \Cref{def:iso-b-u}.

    Consider the following hybrids:
    \begin{itemize}
        \item[$P_0$:]
        \(
            \ex_{\{\ket{\psi_x} \gets \mu_n\}_{x \in \{0,1\}^m}}
                \sparen{\alg^{W_{\left\{\ket{\psi_x}\right\}_{x}}}(1^{\lambda})} .
        \)
        \item[$P_1$:]
        \(
            \ex_{
                \bB \gets \fB \\ \bU \gets \fU
            } \sparen{\alg^{V_{\bB, \bU}}(1^\lambda)}.
        \)
        \item[$P_2$:]
        \(
            \ex_{
                \bB \gets \fB , \bU \gets \fU
            } \sparen{\alg^{V_{\tilde{\bB}, \tilde{\bU}}}(1^\lambda)},
        \)
        where $\tilde{\bB} = (\round{b_{t, z}^{(x)} / \varepsilon_1 } \varepsilon_1)_{t \in [n], z \in \{0, 1\}^t, x \in \{0, 1\}^m}, \tilde{\bU} = (\floor{u_z^{(x)} / \varepsilon_2} \varepsilon_2)_{z \in \{0, 1\}^n, x \in \{0, 1\}^m}$.
        \item[$P_3$:]
        \(
            \ex_{
                f \gets \cY^\cX
            } \sparen{\alg^{V_{\bB_f, \bU_f}}(1^\lambda)},
        \)
        where
        \begin{align*}
            &\bB_f = \big(b_{t, z}^{(x)} := \alg_{RB}(1^{n + \lambda},1^{n+m+\lambda}, 2^{n-t-1}, f(x \cdot 2^{n+1} + 2^t + z))\big)_{t \in [n], z \in \{0, 1\}^t, x \in \{0, 1\}^m} ,\\
            &\bU_f = \big(u_z^{(x)} := f(x \cdot 2^{n+1} + 2^n + z)[1, \dots, \lambda] / 2^{\lambda} \big)_{z \in \{0, 1\}^n, x \in \{0, 1\}^m} .
        \end{align*}
        \item[$P_4$:]
        \(
            \ex_{f \gets \cY^\cX}
                \sparen{\alg^{V_f}(1^\lambda)}.
        \)
    \end{itemize}

    We now prove that each pair of consecutive distributions are close in trace distance.
    \begin{itemize}
        \item $\TD(P_0, P_1) = 0$.\\
            From \Cref{lma:beta-ramp} and \Cref{lma:phase-gaussian}, we know that for $\bB, \bU$ in $S_1$, for each $x$ the coefficients of $\ket{\psi_{\bB_x, \bU_x}}$ follow the distribution of the normalized standard Gaussian random vector. Thus, from \Cref{fact:haar}, we know it's a Haar-random state. And for different $x$, they are all independent.
            Since
            $$
                W_{\left\{\ket{\psi_x}\right\}_{x}} \defeq \sum_{x=0}^{2^m-1}  (\ket{x} \ket{\psi_{x}}) \bra{x}
            $$
            in $P_0$ and
            $$
                V_{\bB, \bU} \defeq \sum_{x=0}^{2^m-1} (\ket{x} \ket{\psi_{\bB_x, \bU_x}}) \bra{x}
            $$
            in $P_1$ have exactly the same distribution, we have $\TD(P_0, P_1) = 0$.

        \item $\TD(P_1, P_2) < l (\sqrt{n} 2^{-\frac \lambda 2} + 2 \pi 2^{-\lambda})$.\\
            Given $\bB = (b_{t, z}^{(x)})_{t \in [n], z \in \{0, 1\}^t, x \in \{0, 1\}^m}$ and $\bU = (u_z^{(x)})_{z \in \{0, 1\}^n, x \in \{0, 1\}^m}$, we define hybrid $P_{1.j}|_{\bB, \bU}$ as
            \(
                \alg^{V_{\tilde{\bB}, \tilde{\bU}}, V_{\bB, \bU}, j}(1^\lambda).
            \)
            That is, the adversary $\alg$ accesses the oracle $V_{\tilde{\bB}, \tilde{\bU}}$ for the first $j$ queries, and then accesses the oracle $V_{\bB, \bU}$ for the remaining queries.

            Since $\abs{b_{t, z}^{(x)} - \round{b_{t, z}^{(x)} / \varepsilon_1 } \varepsilon_1} < \varepsilon_1$ and $\abs{u_z^{(x)} - \floor{u_z^{(x)} / \varepsilon_2} \varepsilon_2} < \varepsilon_2$, from \Cref{lma:iso-b-u-close} we have
            \begin{align*}
                \frac 1 2 \norm{\cE_{\bB, \bU} - \cE_{\tilde{\bB}, \tilde{\bU}}}_\diamond
                &< \sqrt{2^n n \varepsilon_1} + 2 \pi \varepsilon_2 \\
                &= \sqrt{n} 2^{-\frac \lambda 2} + 2 \pi 2^{-\lambda} .
            \end{align*}
            By \Cref{fact:dpi}, we have
            \begin{align}
                \TD(P_{1.j}|_{\bB, \bU}, P_{1.j+1}|_{\bB, \bU}) < \sqrt{n} 2^{-\frac \lambda 2} + 2 \pi 2^{-\lambda} ,
                \label{eq:p-j-close}
            \end{align}
            because the only difference between them is in the $(j+1)$-th query: one of them applies $V_{\bB, \bU}$ and the other applies $V_{\tilde{\bB}, \tilde{\bU}}$, and the trace distance between them after this point is at most $\frac 1 2 \norm{\cE_{\bB, \bU} - \cE_{\tilde{\bB}, \tilde{\bU}}}_\diamond$.

            Then we have
            \begin{align*}
                \TD(P_1, P_2)
                &= \frac 1 2 \norm{P_1 - P_2}_1 \\
                &= \frac 1 2 \norm{
                    \ex_{\substack{
                        \bB \gets \fB \\ \bU \gets \fU
                    }} \sparen{P_{1.0}|_{\bB, \bU}}
                    -
                    \ex_{\substack{
                        \bB \gets \fB \\ \bU \gets \fU
                    }} \sparen{P_{1.l}|_{\bB, \bU}}
                }_1\\
                &\leq \ex_{\substack{
                        \bB \gets \fB \\ \bU \gets \fU
                    }} \sparen{ \frac 1 2 \norm{P_{1.0}|_{\bB, \bU} - P_{1.l}|_{\bB, \bU}}_1}
                        \quad \quad \text{(Triangle inequality)}\\
                &< \ex_{\substack{
                        \bB \gets \fB \\ \bU \gets \fU
                    }} \sparen{l (\sqrt{n} 2^{-\frac \lambda 2} + 2 \pi 2^{-\lambda})}
                        \quad \quad \text{(From \cref{eq:p-j-close} and triangle inequality)}\\
                &= l (\sqrt{n} 2^{-\frac \lambda 2} + 2 \pi 2^{-\lambda}) .
            \end{align*}
        \item $\TD(P_2, P_3) < 2^{-\lambda}$.\\
            In $P_2$, since $b_{t, z}^{(x)} \gets \dBeta(2^{n-t-1}, 2^{n-t-1})$ and $u_z^{(x)} \gets \cU(0, 1)$, we know that $\round{b_{t, z}^{(x)} / \varepsilon_1 } \varepsilon_1$ follows the rounded Beta distribution $\dBeta_{R(\varepsilon_1)}(2^{n-t-1}, 2^{n-t-1})$ (see \Cref{def:rounded-beta}) and $\floor{u_z^{(x)} / \varepsilon_2}$ is uniformly distributed in $\{0, 1\}^{\lambda}$ (recall that $\varepsilon_2 = 2^{-\lambda}$).

            We use $\tilde{\fB}$ to denote the joint distribution of discrete random variables $(\tilde{B}_{t, z}^{(x)})_{t \in [n], z \in \{0, 1\}^t, x \in \{0, 1\}^m}$ where $\tilde{B}_{t, z}^{(x)} \sim \dBeta_{R(\varepsilon_1)}(2^{n-t-1}, 2^{n-t-1})$ independently, and $\tilde{\fU}$ to denote the joint distribution of discrete random variables $(\tilde{U}_z^{(x)})_{z \in \{0, 1\}^n, x \in \{0, 1\}^m}$ where $\tilde{U}_z^{(x)} \sim \cU_{\lambda} / 2^{\lambda}$ independently. Then $P_2$ can be rewritten as
            \[
                P_2 = \ex_{\substack{
                    \bB \gets \tilde{\fB} \\ \bU \gets \tilde{\fU}
                }} \sparen{\alg^{V_{\bB, \bU}}(1^\lambda)}.
            \]

            In $P_3$, since $f \gets \cY^\cX$ is a random function, each $f(x)$ is an independent uniformly random string for $x \in \cX$.
            We use $\dot{\fB}$ to denote the joint distribution of discrete random variables
            $$(\alg_{RB}(1^{n + \lambda},1^{n+m+\lambda}, 2^{n-t-1}, \tilde{U}_{t, z}^{(x)}))_{t \in [n], z \in \{0, 1\}^t, x \in \{0, 1\}^m}$$
            where $\tilde{U}_{t, z}^{(x)} \sim \cU_{\poly(n+m+\lambda)}$ independently.
            Then $P_3$ can be rewritten as
            \[
                P_3 = \ex_{\substack{
                    \bB \gets \dot{\fB}\\
                    \bU \gets \tilde{\fU}
                }} \sparen{\alg^{V_{\bB, \bU}}(1^\lambda)}.
            \]

            \Cref{lma:smp-beta} shows that
            $$\Delta(\alg_{RB}(1^{n + \lambda},1^{n+m+\lambda}, 2^{n-t-1}, \tilde{U}), \tilde{B}) < 2^{-n-m-\lambda},$$
            where $\tilde{U} \sim \cU_{\poly(n+m+\lambda)}$ and $\tilde{B} \sim \dBeta_{R(2^{-n-\lambda})}(2^{n-t-1}, 2^{n-t-1})$. Since $\tilde{\fB}$ (and $\dot{\fB}$) has $2^m(2^n-1)$ random variables, by triangle inequality we have
            \begin{align*}
                \Delta(\tilde{\fB} ,\dot{\fB})
                    &\leq 2^m (2^n - 1) 2^{-n-m-\lambda} \\
                    &< 2^m \cdot 2^n \cdot 2^{-n-m-\lambda}\\
                    &= 2^{-\lambda} .
            \end{align*}

            Then we have
            \begin{align*}
                \TD(P_2, P_3)
                &= \TD \left(
                    \ex_{\substack{
                        \bB \gets \tilde{\fB} \\ \bU \gets \tilde{\fU}
                    }} \sparen{\alg^{V_{\bB, \bU}}}
                    ,
                    \ex_{\substack{
                        \bB \gets \dot{\fB} \\ \bU \gets \tilde{\fU}
                    }} \sparen{\alg^{V_{\bB, \bU}}}
                \right)\\
                &\leq \Delta \left( \tilde{\fB}, \dot{\fB} \right)
                    \quad \quad \text{(Data processing inequality)}\\
                &< 2^{-\lambda} .
            \end{align*}
        \item $\TD(P_3, P_4) < l \sqrt{2 \pi n} 2^{-\frac \lambda 2}$.\\
            Given $f \in \cY^\cX$, we define hybrid $P_{3.j}|_f$ as
            \[
                \alg^{V_f, V_{\bB_f, \bU_f}, j}(1^\lambda).
            \]
            That is, the adversary $\alg$ accesses the oracle $V_f$ for the first $j$ queries, and then accesses the oracle $V_{\bB_f, \bU_f}$ for the remaining queries.

            For each classical input $\ket{x}$,
            we now check the output of $\operatorname{RS}_{n, \lambda}^{U_f^{\gets m}} (\ket{x} \otimes \ket{0}^{\otimes n})$. In \Cref{fig:ARFS-ramp}, we compute $\theta_{t, z}^{(x)}$ with $(n+\lambda)$-bit precision, which is $\floor{\theta_{t, z}^{(x)}/ \varepsilon_1} \varepsilon_1$, and then use it to do rotation. In \Cref{alg:rp}, we apply the phase exactly as $e^{2 \pi i f(x \cdot 2^{n+1} + 2^n + z)[1, \cdots, \lambda] / 2^{\lambda}}$ for each $\ket{z}$. Thus, $\operatorname{RS}_{n, \lambda}^{U_f^{\gets m}} (\ket{x} \otimes \ket{0}^{\otimes n})$ can be written as
            \[
                \operatorname{RS}_{n, \lambda}^{U_f^{\gets m}} (\ket{x} \otimes \ket{0}^{\otimes n}) = \ket{\psi_{\hat{\bB}_{f, x}, \bU_{f, x}}},
            \]
            where
            \[
                \hat{\bB}_{f, x} = \paren{ \hat{b}_{t, z}^{(x)} := \cos^2 \paren{2 \pi \floor{\theta_{t, z}^{(x)}/ \varepsilon_1} \varepsilon_1} }_{t \in [n], z \in \{0, 1\}^t} ,
            \]
            \[
                \theta_{t, z}^{(x)} = \frac {\arccos \sqrt{\alg_{RB}(1^{n + \lambda},1^{n+m+\lambda}, 2^{n-t-1}, f(x \cdot 2^{n+1} + 2^t + z))}} {2 \pi}.
            \]
            The isometry $V_f$ can thus be written as
            \begin{align*}
                V_f \defeq \operatorname{RS}_{n, \lambda}^{U_f^{\gets m}} (I_m \otimes \ket{0}^{\otimes n})
                    = \sum_{x=0}^{2^m-1} (\ket{x} \ket{\psi_{\hat{\bB}_{f, x}, \bU_{f, x}}}) \bra{x}
                    = V_{\hat{\bB}_f, \bU_f}. \quad \mbox{(\Cref{def:iso-b-u})}
            \end{align*}

            Since $\abs{\diff{\cos^2 \theta}{\theta}} \leq 1$, we have $\abs{\hat{b}_{t, z}^{(x)} - b_{t, z}^{(x)}} \leq 2 \pi \varepsilon_1$,
            where \[b_{t, z}^{(x)} = \alg_{RB}(1^{n + \lambda},1^{n+m+\lambda}, 2^{n-t-1}, f(x \cdot 2^{n+1} + 2^t + z))\] as defined in $P_3$.
            Then from \Cref{lma:iso-b-u-close} we have
            \begin{align*}
                \frac 1 2 \norm{\cE_{\bB_f, \bU_f} - \cE_{\hat{\bB}_f, \bU_f}}_\diamond
                < \sqrt{2^n n \cdot 2 \pi \varepsilon_1}
                = \sqrt{2 \pi n } 2^{-\frac \lambda 2} .
            \end{align*}

            By \Cref{fact:dpi}, we have
            \begin{align}
                \TD(P_{3.j}|_f, P_{3.j+1}|_f) < \sqrt{2 \pi n} 2^{-\frac \lambda 2} \label{eq:p3-j-close}
            \end{align}
            because the only difference between them is in the $(j+1)$-th query: one of them applies $V_{\bB_f, \bU_f}$ and the other applies $V_{\hat{\bB}_f, \bU_f}$, and the trace distance between them after this point is at most $\frac 1 2 \norm{\cE_{\bB_f, \bU_f} - \cE_{\hat{\bB}_f, \bU_f}}_\diamond$.

            Then we have
            \begin{align*}
                \TD(P_3, P_4)
                &= \frac 1 2 \norm{P_3 - P_4}_1 \\
                &= \frac 1 2 \norm{
                    \ex_{f \gets \cY^\cX} \sparen{P_{3.0}|_f}
                    -
                    \ex_{f \gets \cY^\cX} \sparen{P_{3.l}|_f}
                }_1\\
                &\leq \ex_{f \gets \cY^\cX} \sparen{ \frac 1 2 \norm{P_{3.0}|_f - P_{3.l}|_f}_1}
                    \quad \quad \text{(Triangle inequality)} \\
                &< \ex_{f \gets \cY^\cX} \sparen{l \sqrt{2 \pi n} 2^{-\frac \lambda 2}}
                    \quad \quad \text{(From \cref{eq:p3-j-close} and triangle inequality)}\\
                &= l \sqrt{2 \pi n} 2^{-\frac \lambda 2} .
            \end{align*}
    \end{itemize}

\end{proof}

\begin{theorem}\label{thm:prfs}
    Let $(\KeyGen_\prf, \Eval_\prf)$ be a scalable quantum-secure pseudorandom function generator as in \Cref{def:scalable-prf}.
    Then the following is a scalable and adaptive pseudorandom function-like quantum state generator:
    \begin{itemize}
        \item $\KeyGen(1^n, 1^m, 1^\lambda)$: Let $k \gets \KeyGen_\prf(1^{\poly(n+m+\lambda)}, 1^{n+m+1}, 1^\lambda)$, output $k$.
        \item $\CircuitGen(1^n, 1^m, 1^\lambda, k)$: Let function $\prf_k: \{0, 1\}^{n+m+1} \to \{0, 1\}^{\poly(n + m + \lambda)}$ be defined as $$\prf_k(\cdot)\defeq \Eval_\prf(1^{\poly(n + m + \lambda)}, 1^\lambda, k, \cdot).$$
        Output the circuit for $\operatorname{RS}_{n, \lambda}^{U_{\prf_k}^{\gets m}} (I_m \otimes \ket{0}^{\otimes n})$, where $\operatorname{RS}_{n, \lambda}^{U_{\prf_k}^{\gets m}}$ is defined in \cref{random-f-u}.
    \end{itemize}
\end{theorem}
This result follows by replacing the random function in \Cref{thm:arfs} with a quantum-secure pseudorandom function, which is a standard technique. A full proof is provided in \Cref{sec:ar_to_pr_proof}.

\subsection*{Acknowledgements}
We thank Kishor Bharti, Srijita Kundu, and Minglong Qin for helpful discussions.

This project is supported by the National Research Foundation, Singapore through the National Quantum Office, hosted in A*STAR, under its Centre for Quantum Technologies Funding Initiative (S24Q2D0009) and its Advanced Quantum Algorithms and Solutions Funding Initiative (S25Q9DA001 and S25Q9DA002).

\newpage

\bibliographystyle{alpha}
\bibliography{main/pseudorandom}


\appendix
\section{Alternative proof of ARS}

Here we show an independent proof of \Cref{thm:ars}.

\label{sec:ars-proof}

\begin{proof}
Let $\varepsilon_1 = 2^{-n-\lambda}$, and $\varepsilon_2 = 2^{-\lambda}$.
    We use $\fB$ to denote the joint distribution of random variables $(B_{t, z})_{t \in [n], z \in \{0, 1\}^t}$ where $B_{t, z} \sim \dBeta(2^{n-t-1}, 2^{n-t-1})$ independently. That is, $$\bB \gets \fB$$ is short for $$\bB = \left(b_{t, z} \gets \dBeta(2^{n-t-1}, 2^{n-t-1}) \right)_{t \in [n], z \in \{0, 1\}^t}.$$ 
    Similarly, we use $\fU$ to denote the joint distribution of random variables $(U_z)_{z \in \{0, 1\}^n}$ where $U_z \sim \cU(0, 1)$ independently. We use notation $\ket{\psi_{\bB, \bU}}$ as in \Cref{def:state-b-u}.
    
    Consider the following hybrids:
    \begin{itemize}
        \item[$S_0$:]
        \[
            \ex_{\ket{\psi} \gets \mu_n} \sparen{\alg^{W_{\ket{\psi}}}(1^\lambda)}.
        \]
        \item[$S_1$:]
        \begin{align*}    
            \ex_{\substack{
                \bB \gets \fB \\ \bU \gets \fU
            }} \sparen{\alg^{W_{\ket{\psi_{\bB, \bU}}}}(1^\lambda)}.
        \end{align*}
        \item[$S_2$:]
        \begin{align*}
            \ex_{\substack{
                \bB \gets \fB \\ \bU \gets \fU
            }} \sparen{\alg^{W_{\ket{\psi_{\tilde{\bB}, \tilde{\bU}}}}}(1^\lambda)},
        \end{align*}
        where $\tilde{\bB} = (\round{b_{t, z} / \varepsilon_1 } \varepsilon_1)_{t \in [n], z \in \{0, 1\}^t}, \tilde{\bU} = (\floor{u_z / \varepsilon_2} \varepsilon_2)_{z \in \{0, 1\}^n}$.
        \item[$S_3$:]
        \begin{align*}
            \ex_{
                f \gets \cY^\cX
            } \sparen{\alg^{W_{\ket{\psi_{\bB_f, \bU_f}}}}(1^\lambda)},
        \end{align*}
        where
        \begin{align*}
            &\bB_f = \big(b_{t, z} := \alg_{RB}(1^{n + \lambda}, 2^{n-t-1}, f(2^t + z))\big)_{t \in [n], z \in \{0, 1\}^t} ,\\
            &\bU_f = \big(u_z := f(2^n + z)[1, \dots, \lambda] / 2^\lambda \big)_{z \in \{0, 1\}^n} .
        \end{align*}
        \item[$S_4$:]
        \begin{align*}
            \ex_{f \gets \cY^\cX}
                \sparen{\alg^{V_f}(1^\lambda)}.
        \end{align*}
    \end{itemize}

    We now prove that each pair of consecutive distributions are close in trace distance.
    \begin{itemize}
        \item $\TD(S_0, S_1) = 0$.\\
            From \Cref{lma:beta-ramp} and \Cref{lma:phase-gaussian}, we know that the coefficients of $\ket{\psi_{\bB, \bU}}$ in $S_1$ follow the distribution of the normalized standard Gaussian random vector. Thus, from \Cref{fact:haar}, we know it's a Haar-random state. Since $\ket{\psi}$ in $S_0$ and $\ket{\psi_{\bB, \bU}}$ in $S_1$ have exactly the same distribution, we have $\TD(S_0, S_1) = 0$.

        \item $\TD(S_1, S_2) < l (\sqrt{n} 2^{-\frac \lambda 2} + 2 \pi 2^{-\lambda})$.\\
            Given $\bB = (b_{t, z})_{t \in [n], z \in \{0, 1\}^t}$ and $\bU = (u_z)_{z \in \{0, 1\}^n}$, we define hybrid $S_{1.j}|_{\bB, \bU}$ as
            \[
                \alg^{W_{\ket{\psi_{\tilde{\bB}, \tilde{\bU}}}}, W_{\ket{\psi_{\bB, \bU}}}, j}(1^\lambda).
            \]
            That is, the adversary $\alg$ accesses the oracle $W_{\ket{\psi_{\tilde{\bB}, \tilde{\bU}}}}$ for the first $j$ queries, and then accesses the oracle $W_{\ket{\psi_{\bB, \bU}}}$ for the remaining queries.

            Since $\abs{b_{t, z} - \round{b_{t, z} / \varepsilon_1 } \varepsilon_1} < \varepsilon_1$ and $\abs{u_z - \floor{u_z / \varepsilon_2} \varepsilon_2} < \varepsilon_2$, from \Cref{lma:state-b-u-close} we have
            \begin{align*}
                &\TD(\ketbra{\psi_{\bB, \bU}}, \ketbra{\psi_{\tilde{\bB}, \tilde{\bU}}}) \\
                &< \sqrt{2^n n \varepsilon_1} + 2 \pi \varepsilon_2 \\
                &= \sqrt{n} 2^{-\frac \lambda 2} + 2 \pi 2^{-\lambda} .
            \end{align*}
            By \Cref{fact:dpi}, we have
            \begin{align}    
                \TD(S_{1.j}|_{\bB, \bU}, S_{1.j+1}|_{\bB, \bU}) < \sqrt{n} 2^{-\frac \lambda 2} + 2 \pi 2^{-\lambda} ,
                \label{eq:s-j-close}
            \end{align}
            because the only difference between them is in the $(j+1)$-th query one of them gets $\ket{\psi_{\bB, \bU}}$ and the other gets $\ket{\psi_{\tilde{\bB}, \tilde{\bU}}}$.

            Then we have
            \begin{align*}
                &\TD(S_1, S_2)\\
                &= \frac 1 2 \norm{S_1 - S_2}_1 \\
                &= \frac 1 2 \norm{
                    \ex_{\substack{
                        \bB \gets \fB \\ \bU \gets \fU
                    }} \sparen{S_{1.0}|_{\bB, \bU}}
                    -
                    \ex_{\substack{
                        \bB \gets \fB \\ \bU \gets \fU
                    }} \sparen{S_{1.l}|_{\bB, \bU}}
                }_1\\
                &\leq \ex_{\substack{
                        \bB \gets \fB \\ \bU \gets \fU
                    }} \sparen{ \frac 1 2 \norm{S_{1.0}|_{\bB, \bU} - S_{1.l}|_{\bB, \bU}}_1} && \text{(Triangle inequality)}\\
                &< \ex_{\substack{
                        \bB \gets \fB \\ \bU \gets \fU
                    }} \sparen{l (\sqrt{n} 2^{-\frac \lambda 2} + 2 \pi 2^{-\lambda})}
                    && \text{(From \cref{eq:s-j-close} and triangle inequality)}\\
                &= l (\sqrt{n} 2^{-\frac \lambda 2} + 2 \pi 2^{-\lambda}) .
            \end{align*}
        \item $\TD(S_2, S_3) < 2^{-\lambda}$.\\
            In $S_2$, since $b_{t, z} \gets \dBeta(2^{n-t-1}, 2^{n-t-1})$ and $u_z \gets \cU(0, 1)$, we know that $\round{b_{t, z} / \varepsilon_1 } \varepsilon_1$ follows the rounded Beta distribution $\dBeta_{R(\varepsilon_1)}(2^{n-t-1}, 2^{n-t-1})$ (see \Cref{def:rounded-beta}) and $\floor{u_z / \varepsilon_2}$ is uniformly distributed in $\{0, 1\}^\lambda$ (recall that $\varepsilon_2 = 2^{-\lambda}$).

            We use $\tilde{\fB}$ to denote the joint distribution of discrete random variables $(\tilde{B}_{t, z})_{t \in [n], z \in \{0, 1\}^t}$ where $\tilde{B}_{t, z} \sim \dBeta_{R(\varepsilon_1)}(2^{n-t-1}, 2^{n-t-1})$ independently, and $\tilde{\fU}$ to denote the joint distribution of discrete random variables $(\tilde{U}_z)_{z \in \{0, 1\}^n}$ where $\tilde{U}_z \sim \cU_\lambda / 2^\lambda$ independently. Then $S_2$ can be rewritten as
            \[
                S_2 = \ex_{\substack{
                    \bB \gets \tilde{\fB} \\ \bU \gets \tilde{\fU}
                }} \sparen{\alg^{W_{\ket{\psi_{\bB, \bU}}}}(1^\lambda)}.
            \]

            In $S_3$, since $f \gets \cY^\cX$ is a random function, each $f(x)$ is an independent uniformly random string for $x \in \cX$.
            We use $\dot{\fB}$ to denote the joint distribution of discrete random variables
            $$(\alg_{RB}(1^{n + \lambda}, 2^{n-t-1}, \tilde{U}_{t, z}))_{t \in [n], z \in \{0, 1\}^t}$$
            where $\tilde{U}_{t, z} \sim \cU_{\poly(n + \lambda)}$ independently.
            Then $S_3$ can be rewritten as
            \[
                S_3 = \ex_{\substack{
                    \bB \gets \dot{\fB}\\
                    \bU \gets \tilde{\fU}
                }} \sparen{\alg^{W_{\ket{\psi_{\bB, \bU}}}}(1^\lambda)}.
            \]

            \Cref{lma:smp-beta} shows that
            $$\Delta(\alg_{RB}(1^{n + \lambda}, 2^{n-t-1}, \tilde{U}), \tilde{B}) < 2^{-n-\lambda},$$
            where $\tilde{U} \sim \cU_{\poly(n + \lambda)}$ and $\tilde{B} \sim \dBeta_{R(2^{-n-\lambda})}(2^{n-t-1}, 2^{n-t-1})$. Since $\tilde{\fB}$ (and $\dot{\fB}$) has $2^n-1$ random variables, by triangle inequality we have
            \[
                \Delta(\tilde{\fB} ,\dot{\fB}) \leq (2^n - 1) 2^{-n-\lambda} < 2^{-\lambda} .
            \] 

            Then we have
            \begin{align*}
                &\TD(S_2, S_3) \\
                &= \TD \left(
                    \ex_{\substack{
                        \bB \gets \tilde{\fB} \\ \bU \gets \tilde{\fU}
                    }} \sparen{\alg^{W_{\ket{\psi_{\bB, \bU}}}}}
                    ,
                    \ex_{\substack{
                        \bB \gets \dot{\fB} \\ \bU \gets \tilde{\fU}
                    }} \sparen{\alg^{W_{\ket{\psi_{\bB, \bU}}}}}
                \right)\\
                &\leq \Delta \left( \tilde{\fB}, \dot{\fB} \right)
                    && \text{(Data processing inequality)}\\
                &< 2^{-\lambda} .
            \end{align*}
        \item $\TD(S_3, S_4) < l \sqrt{2 \pi n} 2^{-\frac \lambda 2}$.\\
            Given $f \in \cY^\cX$, we define hybrid $S_{3.j}|_f$ as
            \[
                \alg^{V_f, W_{\ket{\psi_{\bB_f, \bU_f}}}, j}(1^\lambda).
            \]
            That is, the adversary $\alg$ accesses the oracle $V_f$ for the first $j$ queries, and then accesses the oracle $W_{\ket{\psi_{\bB_f, \bU_f}}}$ for the remaining queries.

            We now check the output of $V_f = \operatorname{RS}_{n, \lambda}^{U_f} \ket{0}^{\otimes n}$. In \Cref{alg:ramp}, we compute $\theta_{t, z}$ with $(n+\lambda)$-bit precision, which is $\floor{\theta_{t, z}/ \varepsilon_1} \varepsilon_1$, and then use it to do rotation. In \Cref{alg:rp}, we apply the phase exactly as $e^{2 \pi i f(2^n + z)[1, \cdots, \lambda] / 2^\lambda}$ for each $\ket{z}$. Thus, $\operatorname{RS}_{n, \lambda}^{U_f} \ket{0}^{\otimes n}$ can be written as
            \[
                \operatorname{RS}_{n, \lambda}^{U_f} \ket{0}^{\otimes n} = \ket{\psi_{\hat{\bB}_f, \bU_f}},
            \]
            where
            \[
                \hat{\bB}_f = \paren{ \hat{b}_{t, z} := \cos^2 \paren{2 \pi \floor{\theta_{t, z}/ \varepsilon_1} \varepsilon_1} }_{t \in [n], z \in \{0, 1\}^t} ,
            \]
            \[
                \theta_{t, z} = \frac {\arccos \sqrt{\alg_{RB}(1^{n+\lambda}, 2^{n-t-1}, f(2^t + z))}} {2 \pi}.
            \]
            That is,
            \[
                V_f = W_{\ket{\psi_{\hat{\bB}_f, \bU_f}}}.
            \]

            Since $\abs{\diff{\cos^2 \theta}{\theta}} \leq 1$, we have $\abs{\hat{b}_{t, z} - b_{t, z}} \leq 2 \pi \varepsilon_1$, where $b_{t, z} = \alg_{RB}(1^{n + \lambda}, 2^{n-t-1}, f(2^t + z)$ as defined in $S_3$. Then from \Cref{lma:state-b-u-close} we have
            \begin{align*}
                &\TD(\ketbra{\psi_{\bB_f, \bU_f}}, \ketbra{\psi_{\hat{\bB}_f, \bU_f}}) \\
                &< \sqrt{2^n n \cdot 2 \pi \varepsilon_1} \\
                &= \sqrt{2 \pi n } 2^{-\frac \lambda 2} .
            \end{align*}

            By \Cref{fact:dpi}, we have
            \begin{align}
                \TD(S_{3.j}|_f, S_{3.j+1}|_f) < \sqrt{2 \pi n} 2^{-\frac \lambda 2} \label{eq:s3-j-close}
            \end{align}
            because the only difference between them is in the $(j+1)$-th query one of them gets $\ket{\psi_{\bB_f, \bU_f}}$ and the other gets $\ket{\psi_{\hat{\bB}_f, \bU_f}}$.

            Then we have
            \begin{align*}
                &\TD(S_3, S_4) \\
                &= \frac 1 2 \norm{S_3 - S_4}_1 \\
                &= \frac 1 2 \norm{
                    \ex_{f \gets \cY^\cX} \sparen{S_{3.0}|_f}
                    -
                    \ex_{f \gets \cY^\cX} \sparen{S_{3.l}|_f}
                }_1\\
                &\leq \ex_{f \gets \cY^\cX} \sparen{ \frac 1 2 \norm{S_{3.0}|_f - S_{3.l}|_f}_1}  && \text{(Triangle inequality)}\\
                &< \ex_{f \gets \cY^\cX} \sparen{l \sqrt{2 \pi n} 2^{-\frac \lambda 2}} 
                    && \text{(From \cref{eq:s3-j-close} and triangle inequality)}\\
                &= l \sqrt{2 \pi n} 2^{-\frac \lambda 2} .
            \end{align*}
    \end{itemize}
\end{proof}

\section{Proof of asymptotically random to pseudorandom} \label{sec:ar_to_pr_proof}

\subsection*{Proof of \Cref{thm:prs}}

\begin{proof}
    $\prf_k$ can be implemented efficiently classically, so $U_{\prf_k}$ can be implemented efficiently. 
    Since we already show the efficient implementation of \Cref{alg:ramp} and \Cref{alg:rp},
    implementation of $\operatorname{RS}_{n, \lambda}^{U_{\prf_k}}$ can be output efficiently.

    Let $\cX = \{0, 1\}^{n+1}$, $\cY = \{0,1\}^{\poly(n + \lambda)}$ and $f \in \cY^{\cX}$.
    We define
    \begin{itemize}
        \item $V_f: \C \to \C^{2^n}$ is the isometry $V_f \defeq \operatorname{RS}_{n, \lambda}^{U_f} \ket{0}^{\otimes n}$.
        \item $W_{\ket{\psi}}: \C \to \C^{2^n}$ is the isometry $W_{\ket{\psi}} \defeq \ket{\psi}$.
    \end{itemize}
    Thus, $\CircuitGen(1^n, 1^\lambda, k)$ outputs $V_{\prf_k}$.

    To show the computational indistinguishability, we start from truly random functions.
    For any (non-uniform) QPT distinguisher $\alg$, it can only access the oracle polynomially many times (polynomial in $\lambda$).
    When $n$ is also in $\poly(\lambda)$, from \Cref{thm:ars} we have
    \begin{align}\label{eq:distingusher1}
        \abs{
            \Pr_{f \gets \cY^\cX }
                \sparen{\alg^{V_{f}}(1^\lambda) = 1}
            - \Pr_{\ket{\psi} \gets \mu_n}
                \sparen{\alg^{W_{\ket{\psi}}}(1^\lambda) = 1}
        } = \negl(\lambda) .
    \end{align}

    Assuming $\alg$ can distinguish $\{V_{\prf_k}\}_{k}$ and $\{V_{f}\}_{f}$ with non-negligible advantage $\nu(\lambda)$, then we can construct a QPT distinguisher $\mathcal{B}$ to distinguish $\{U_{\prf_k}\}_{k}$ and $\{U_f\}_{f}$:
    \begin{enumerate}
        \item Construct $\operatorname{RS}_{n, \lambda}^{\oracle}$ using the given oracle $\oracle$.
        \item Run $\alg$, using $\operatorname{RS}_{n, \lambda}^{\oracle} \ket{0}^{\otimes n}$ as the oracle, and output the result.
    \end{enumerate}
    This gives
    \[
        \abs{
            \Pr_{k \gets \KeyGen(1^n, 1^\lambda)}
                \left[\mathcal{B}^{U_{\prf_k}}(1^\lambda) = 1\right]
            - \Pr_{f \gets \cY^\cX }
                \left[\mathcal{B}^{U_f}(1^\lambda) = 1\right]
        } = \nu(\lambda) ,
    \]
    which contradicts that $\{\prf_k\}_k$ is quantum-secure pseudorandom. Thus, we have
        \begin{align}\label{eq:distingusher2}
        \abs{
            \Pr_{k \gets \KeyGen(1^n, 1^\lambda)}
                \left[\alg^{V_{\prf_k}}(1^\lambda) = 1\right]
            -
            \Pr_{f \gets \cY^\cX }
                \sparen{\alg^{V_{f}}(1^\lambda) = 1}
        } = \negl(\lambda).
    \end{align}
    From \cref{eq:distingusher1} and \cref{eq:distingusher2}, we get
    \begin{align*}
        \abs{
            \Pr_{k \gets \KeyGen(1^n, 1^\lambda)}
                \left[\alg^{V_{\prf_k}}(1^\lambda) = 1\right]
            - \Pr_{\ket{\psi} \gets \mu_n}\left[\alg^{W_{\ket{\psi}}}(1^\lambda) = 1\right]
        } = \negl(\lambda) .
    \end{align*}
\end{proof}

\subsection*{Proof of \Cref{thm:prfs}}

\begin{proof}
    $\prf_k$ can be implemented efficiently classically, so $U_{\prf_k}$ can be implemented efficiently.
    By \Cref{lma:smp-beta}, the two-parameter sampler used in $B_f'^{(t)}$ is efficient. Together with the efficient implementations of \Cref{alg:ramp} and \Cref{alg:rp},
    the circuit of $\operatorname{RS}_{n, \lambda}^{U_{\prf_k}^{\gets m}}$ can be output efficiently.

    Let $\cX = \{0, 1\}^{n+m+1}, \cY = \{0,1\}^{\poly(n+m+\lambda)}$ and $f \in \cY^\cX$, we define
    \begin{itemize}
        \item $V_f: \C^{2^m} \to \C^{2^{n+m}}$ is the isometry
        $$
            V_f \defeq \operatorname{RS}_{n, \lambda}^{U_f^{\gets m}} (I_m \otimes \ket{0}^{\otimes n}).
        $$
        \item $W_{\left\{\ket{\psi_x}\right\}_{x}}:\C^{2^m} \to \C^{2^{n+m}}$ is the isometry 
        $$
            W_{\left\{\ket{\psi_x}\right\}_{x}} \defeq  \sum_{x=0}^{2^m-1}  (\ket{x} \ket{\psi_{x}}) \bra{x}.
        $$
    \end{itemize}
    Thus, $\CircuitGen(1^n, 1^m, 1^\lambda, k, \cdot)$ output $V_{\prf_k}$.

    To show computational indistinguishability, we start from truly random functions.
    For any (non-uniform) quantum polynomial-time distinguisher $\alg$, it can only access the oracle polynomially many times (polynomial in $\lambda$). When $n$ and $m$ are also in $\poly(\lambda)$, from \Cref{thm:arfs} we have
    \begin{align*}
        \abs{
            \Pr_{f \gets \cY^\cX }
                [\alg^{V_f}(1^\lambda) = 1]
            - \Pr_{\{\ket{\psi_x} \gets \mu_n\}_{x \in \{0,1\}^m}}\sparen{\alg^{W_{\left\{\ket{\psi_x}\right\}_{x}}}(1^{\lambda}) = 1}
        } = \negl(\lambda) .
    \end{align*}
Since $\prf$ is a family of quantum-secure pseudorandom functions, similar to \Cref{thm:prs}, we have
    \begin{align*}
        \abs{
            \Pr_{k \gets \KeyGen(1^n, 1^m, 1^\lambda) }
                [\alg^{V_{\prf_k}}(1^\lambda) = 1]
            - \Pr_{\{\ket{\psi_x} \gets \mu_n\}_{x \in \{0,1\}^m}}\sparen{\alg^{W_{\left\{\ket{\psi_x}\right\}_{x}}}(1^{\lambda}) = 1}
        } = \negl(\lambda) .
    \end{align*}
\end{proof}

\section{Proof of the random amplitudes quantum state from Beta variables} \label{sec:ram-proof}

\subsection*{Proof of \Cref{lma:beta-ramp}}

\begin{proof}
    We proceed by induction on \( t \) to show that the state after the \( t \)-th iteration, denoted by \( \ket{\phi_t} \), can be written as
    \[
        \ket{\phi_t} = \frac 1 {\sqrt{\sum_{z=0}^{2^t-1} W_{t, z}}} \sum_{z=0}^{2^t-1} \sqrt{W_{t, z}} \ket{z} \ket{0}^{\otimes n-t},
    \]
    where each $W_{t, z} \sim \chi^2_{2^{n-t+1}}$ i.i.d. and is also independent of all $B_{t', z'}$ for $t' \geq t$.
    We take $\ket{\phi_0} = \ket{0}^{\otimes n}$ as the initial state before the first iteration. 

    \begin{itemize}
        \item Base case: This holds for $t=0$ since $\ket{\phi_0} = \ket{0}^{\otimes n}$.
        \item Inductive step: 
        Assume this holds for $t$, then the state after the $(t+1)$-th iteration is
        \begin{align*}
            & \ket{\phi_{t+1}} \\
            &= \frac 1 {\sqrt{\sum_{z=0}^{2^t-1} W_{t, z}}} \sum_{z=0}^{2^t-1} \sqrt{W_{t,z}} \ket{z} \paren{\sqrt{B_{t, z}} \ket{0} + \sqrt{1-B_{t, z}} \ket{1}} \ket{0}^{\otimes n-t-1}\\
            &= \frac 1 {\sqrt{\sum_{z=0}^{2^t-1} (W_{t, z} B_{t, z} + W_{t, z} (1-B_{t, z}))}} \sum_{z=0}^{2^t-1} \paren{\sqrt{W_{t, z} B_{t, z}} \ket{z} \ket{0} + \sqrt {W_{t, z} (1-B_{t, z})} \ket{z} \ket{1}} \ket{0}^{\otimes n-t-1}\\
            &= \frac 1 {\sqrt{\sum_{z=0}^{2^{t+1}-1}} W_{t+1, z}} \sum_{z=0}^{2^{t+1}-1} \sqrt{W_{t+1, z}} \ket{z} \ket{0}^{\otimes n-(t+1)},
        \end{align*}
        where $W_{t+1, 2z} = W_{t, z} B_{t, z}, W_{t+1, 2z+1} = W_{t, z} (1-B_{t, z})$.

        Notice that
        $$W_{t, z} = W_{t+1, 2z} + W_{t+1, 2z+1},$$
        and
        $$B_{t, z} = W_{t+1, 2z} / (W_{t+1, 2z} + W_{t+1, 2z+1}).$$
        From \Cref{fact:gamma-beta} we have
        \begin{align}
            &W_{t+1, 2z}, W_{t+1, 2z+1} \sim \chi^2_{2^{n-t}} \text{ independently} \nonumber\\
            &\Leftrightarrow W_{t, z} \sim \chi^2_{2^{n-t+1}}, B_{t, z} \sim \dBeta(2^{n-t-1}, 2^{n-t-1}) \text{ independently}. \label{eq:induc-step} 
        \end{align}

        Since for all $z \in \{0, 1\}^t$, $W_{t, z} \sim \chi^2_{2^{n-t}}$ (inductive hypothesis) and $B_{t, z} \sim \dBeta(2^{n-t-1}, 2^{n-t-1})$, \cref{eq:induc-step} gives that for all $z \in \{0, 1\}^{t+1}$,
        $W_{t+1, z} \sim \chi^2_{2^{n-t}}$ i.i.d. and is independent of all $B_{t', z'}$ for $t' \geq t+1$.
        That is, the statement also holds for $t+1$.
    \end{itemize}

    The output state of this procedure is
    \[
        \ket{\phi_n} = \frac 1 {\sqrt{\sum_{z=0}^{2^n-1} W_{n, z}}} \sum_{z=0}^{2^n-1} \sqrt{W_{n, z}} \ket{z}
    \]
    with $W_{n, z} \sim \chi^2_2$ i.i.d., which is a random amplitudes quantum state (\Cref{def:ramp-state}).
\end{proof}

\section{Proof of that random phase gives standard Gaussian} \label{sec:phase-gaussian-proof}

\subsection*{Proof of \Cref{lma:phase-gaussian}}

\begin{proof}
    Let $U_1 = e^{-C / 2} \in (0, 1]$.
    PDF of $\chi_2^2$ is $f_C(c) = \frac 1 2 e^{-c/2}, c \in [0, +\infty)$ (\Cref{def:chi-sq}),
    then from \Cref{fact:trans} we know that PDF of $U_1$ is $f_U(u) = f_C(c) / \abs{ \diff{}{c} e^{-c/2}} = 1$,
    which means $U_1  \sim \cU(0, 1)$.
    Rewrite $\sqrt{C} e^{2\pi i U}$ as
    \begin{align*}
        \sqrt{C} e^{2\pi i U}
        &= \sqrt{-2 \ln U_1} e^{2\pi i U} \\
        &= \sqrt{-2 \ln U_1} \cos (2\pi U) +  i \sqrt{-2 \ln U_1} \sin (2\pi U) .
    \end{align*}
    According to Box-Muller transform (\Cref{fact:box-muller}), $\sqrt{-2 \ln U_1} \cos 2\pi U, \sqrt{-2 \ln U_1} \sin 2\pi U \sim \cN(0, 1)$ i.i.d.
\end{proof}

\section{Proof of the diamond-norm bound via operator norm} \label{sec:diamond-proof}

\subsection*{Proof of \Cref{lma:iso-b-u-close}}

\begin{proof}
    We first relate the diamond-norm distance between two isometric channels to the operator-norm distance between their isometries. Let $V$ and $W$ be isometries and let $\cE_V$ and $\cE_W$ be their corresponding channels. For every reference register $R$ and every $\rho \in D(\cH_R \otimes \cH_A)$,
    \begin{align*}
        &\norm{(\id_R \otimes \cE_V)(\rho) - (\id_R \otimes \cE_W)(\rho)}_1 \\
        &\quad = \norm{
            (I_R \otimes (V-W)) \rho (I_R \otimes V^\dagger)
            + (I_R \otimes W) \rho (I_R \otimes (V^\dagger-W^\dagger))
        }_1 \\
        &\quad \leq 2 \norm{V-W}_{\mathrm{op}},
    \end{align*}
    where we use $\norm{AXB}_1 \leq \norm{A}_{\mathrm{op}}\norm{X}_1\norm{B}_{\mathrm{op}}$, $\norm{\rho}_1=1$, and the fact that an isometry has operator norm one. Taking the supremum over $R$ and $\rho$ gives
    \begin{equation}
        \frac 1 2 \norm{\cE_V-\cE_W}_\diamond \leq \norm{V-W}_{\mathrm{op}}.
        \label{eq:diamond-op}
    \end{equation}

    Let $N=2^n$. From the preparation procedure in \Cref{def:state-b-u}, for every $x \in \{0,1\}^m$ we can write
    \[
        \ket{\psi_{\bB_x,\bU_x}}
        = \sum_{z \in [N]} e^{2\pi i u_z^{(x)}}\sqrt{c_z^{(x)}}\ket{z},
        \qquad
        \ket{\psi_{\bB'_x,\bU'_x}}
        = \sum_{z \in [N]} e^{2\pi i {u'}_z^{(x)}}\sqrt{{c'}_z^{(x)}}\ket{z},
    \]
    where
    \[
        c_z^{(x)} = \prod_{t \in [n]} \operatorname{Flip}(b_{t,z[1,\dots,t]}^{(x)},z[t+1]),
        \qquad
        {c'}_z^{(x)} = \prod_{t \in [n]} \operatorname{Flip}({b'}_{t,z[1,\dots,t]}^{(x)},z[t+1]),
    \]
    and $\operatorname{Flip}(a,d)=a$ if $d=0$ and $\operatorname{Flip}(a,d)=1-a$ if $d=1$. Since all the factors lie in $[0,1]$, a hybrid over the $n$ factors gives
    \begin{equation}
        \abs{c_z^{(x)}-{c'}_z^{(x)}} < n\delta_1.
        \label{eq:dis-c}
    \end{equation}
    It follows that changing the amplitudes while keeping the phases fixed changes each output vector by at most
    \begin{align*}
        \norm{\ket{\psi_{\bB_x,\bU_x}}-\ket{\psi_{\bB'_x,\bU_x}}}_2^2
        &= \sum_{z \in [N]} \paren{\sqrt{c_z^{(x)}}-\sqrt{{c'}_z^{(x)}}}^2 \\
        &\leq \sum_{z \in [N]} \abs{c_z^{(x)}-{c'}_z^{(x)}} \\
        &< 2^n n\delta_1,
    \end{align*}
    where we use $(\sqrt{a}-\sqrt{b})^2 \leq \abs{a-b}$ for $a,b\geq 0$. Similarly, changing only the phases gives
    \begin{align*}
        \norm{\ket{\psi_{\bB'_x,\bU_x}}-\ket{\psi_{\bB'_x,\bU'_x}}}_2^2
        &= \sum_{z \in [N]} {c'}_z^{(x)}
            \abs{e^{2\pi i u_z^{(x)}}-e^{2\pi i {u'}_z^{(x)}}}^2 \\
        &< (2\pi\delta_2)^2,
    \end{align*}
    since $\abs{e^{ia}-e^{ib}} \leq \abs{a-b}$ and $\sum_z {c'}_z^{(x)}=1$. Hence, by the triangle inequality, for every $x$,
    \begin{equation}
        \norm{\ket{\psi_{\bB_x,\bU_x}}-\ket{\psi_{\bB'_x,\bU'_x}}}_2
        < \sqrt{2^n n\delta_1}+2\pi\delta_2.
        \label{eq:state-vector-close}
    \end{equation}

    Finally, the classical input register makes the isometries block diagonal. Therefore,
    \begin{align*}
        \norm{V_{\bB,\bU}-V_{\bB',\bU'}}_{\mathrm{op}}
        &= \max_{x \in \{0,1\}^m}
            \norm{\ket{\psi_{\bB_x,\bU_x}}-\ket{\psi_{\bB'_x,\bU'_x}}}_2 \\
        &< \sqrt{2^n n\delta_1}+2\pi\delta_2.
    \end{align*}
    Combining this with \cref{eq:diamond-op} proves the claim without a dimension-dependent loss in $m$.
\end{proof}

\section{Proof of claim about the scalable pseudorandom state generator} \label{sec:scalable-prs-proof}

\subsection*{Proof of: \Cref{def:PRSour} implies \Cref{def:PRSJLS}}

\begin{proof}
    We write $\KeyGen(1^n, 1^\lambda)$ as $\KeyGen(1^n, 1^\lambda; r)$, where $r$ is the randomness of this PPT algorithm. For given $\lambda$ and $n$, let $\cK$ be the set of all possible $r$. Then the family of quantum states $\{\ket{\phi_r} \in S(\C^{2^n})\}_{r \in \cK}$ is given by:
    
    \begin{enumerate}
        \item[$G(r)$:]
        \item $k_r = \KeyGen(1^n, 1^\lambda; r)$.
        \item $C_r = \CircuitGen(1^n, 1^\lambda, k_r)$.
        \item Output $\ket{\phi_r} = C_r()$.
    \end{enumerate}
    Since $\CircuitGen$ is a PT algorithm, the size of the quantum circuit $C_r$ is polynomially bounded in $\lambda$. The above algorithm $G$ is indeed a polynomial-time quantum algorithm.

    To show the pseudorandomness, assume that there exists $t \in \poly(\lambda)$ and a QPT distinguisher $\alg$ such that
    \[
        \abs{
            \Pr_{r \gets \cK}\sparen{\alg(\ket{\phi_r}^{\otimes t}) = 1}
            -
            \Pr_{\ket{\psi} \gets \mu_n}\sparen{\alg(\ket{\psi}^{\otimes t}) = 1}
        } = \varepsilon(\lambda) \notin \negl(\lambda).
    \]
    Then we can construct a QPT distinguisher $\mathcal{B}$ by using $V_{k_r}$ $t$ times to get $t$ copies of $\ket{\phi_r}$, and then running $\alg$. This gives
    \[
        \abs{
            \Pr_{\substack{r \gets \cK \\ k_r = \KeyGen(1^n, 1^\lambda; r)}}\sparen{\mathcal{B}^{V_{k_r}}(1^\lambda) = 1}
            -
            \Pr_{\ket{\psi} \gets \mu_n}\sparen{\mathcal{B}^{W_{\ket{\psi}}}(1^\lambda) = 1}
        } = \varepsilon(\lambda),
    \]
    and contradicts the pseudorandomness of $(\KeyGen, \CircuitGen)$ according to \Cref{def:PRSour}. Thus, $\{\ket{\phi_r}\}_{r \in \cK}$ is pseudorandom according to \Cref{def:PRSJLS}.

\end{proof}

\section{Classical sampling algorithms}

The main result of this section is an efficient deterministic classical algorithm to sample from the (finite precision) Beta distribution with given randomness as input (\Cref{lma:smp-beta}).

First, without considering precision, we show that, given a random variable from the normal distribution and a random variable from the uniform distribution, we can efficiently sample from $\dGamma(\alpha, 1)$ for $\alpha \geq 1$, with failure probability bounded by some constant.

\begin{lemma}\label{lma:smp-gamma}
    There exists a deterministic algorithm $\alg_G$, such that for any $\alpha \geq 1$, with $X \sim \cN(0, 1)$ and $U \sim \cU(0, 1)$ independently, the following holds:
    \begin{itemize}
        \item $\alg_G(\alpha, X, U)|_{\alg_G(\alpha, X, U) \neq \bot} \sim \dGamma(\alpha, 1)$.
        \item $\Pr[\alg_G(\alpha, X, U) = \bot] < 0.05$.
    \end{itemize}
\end{lemma}

\begin{proof}
    We use Marsaglia and Tsang's algorithm \cite{MT00}, and provide a rigorous proof of its correctness.

    \begin{algorithm}[H]
        \caption{$\alg_G$: Sample from standard Gamma distribution}
        \label{alg:g}
        \begin{algorithmic}[1]
            \INPUT $\alpha \geq 1$, $x \gets \mathcal{N}(0, 1)$, $u \gets \cU(0, 1)$
            \State $d := \alpha -\frac 1 3, c := \frac 1 {\sqrt{9 d}}$.
            \State $v := (1 + c x)^3$.
            \If {$v > 0$ and $\ln(u) < \frac 1 2 x^2 + d - d v + d \ln(v)$}
                \State return $d v$
            \Else
                \State return $\bot$
            \EndIf
        \end{algorithmic}
    \end{algorithm}

    First, we show that the output follows $\dGamma(\alpha, 1)$ when this algorithm succeeds.
    Let $h(x) = d(1 + c x)^3$ for $x \in (- \frac 1 c, +\infty)$. If a random variable $X$ follows the distribution with PDF $h(x)^{\alpha-1} e^{-h(x)} h'(x) / \Gamma(\alpha)$, then PDF of $Y = h(X)$ will be $y^{\alpha-1} e^{-y} / \Gamma(\alpha)$ (\Cref{fact:trans}, since $h(x)$ is monotonic and differentiable, $h'(x) > 0$), which means $Y$ follows $\dGamma(\alpha, 1)$ (recall that $\dGamma(\alpha, \theta)$ has PDF $f(y) = \frac 1 {\Gamma(\alpha) \theta^\alpha } y^{\alpha-1} e^{-y/\theta} $).

    We have
    \begin{align}
        &h(x)^{\alpha-1} e^{-h(x)} h'(x) / \Gamma(\alpha) \nonumber\\
        &= d^{\alpha-1} (1 + c x)^{3(\alpha-1)} e^{-d(1 + c x)^3} 3 c d (1 + c x)^2 / \Gamma(\alpha) \nonumber\\
        &= 3 c d^{\alpha} / \Gamma(\alpha) \cdot e^{-d} \cdot e^{
            d \ln (1 + c x)^3 - d(1+ c x)^3 +d
        }.
            \quad \quad \text{(Since $d = \alpha - \frac 1 3$)}
        \label{eq:target_pdf}
    \end{align}

    Let $g(x) = d \ln (1 + c x)^3 - d(1+ c x)^3 +d, x \in (- \frac 1 c, +\infty)$.
    Rewrite \cref{eq:target_pdf} as $ C e^{g(x)}$, where the constant $C = 3 c d^{\alpha} / \Gamma(\alpha) e^{-d}$.
    Our target PDF for $x$ is
    \begin{align*}
        f_0(x) = \begin{cases}
            0 , & x \in (- \infty , -\frac 1 c]\\
            C e^{g(x)}, & x \in (-\frac 1 c, + \infty)
        \end{cases} .
    \end{align*}
    From \Cref{fun-lma:s-lesser-than-0} we know $g(x) \leq -\frac  1 2 x^2$,
    so we can sample $x$ from $\cN(0, 1)$, whose PDF is $f_1(x) = \frac 1 {\sqrt{2 \pi}} e^{-\frac 1 2 x^2}$,
    to do the rejection sampling (\Cref{fact:rej-samp}, with $M = \sqrt{2 \pi} C$) ,
    accepting $x$ with probability
    \[
    p_\text{accept} (x) =
    \begin{cases}
        0, & x \in (- \infty , -\frac 1 c]\\
        e^{g(x)} / e^{-\frac 1 2 x^2}, & x \in (-\frac 1 c, + \infty)
    \end{cases}.
    \]
    In $\alg_G$ this acceptance probability is simulated by checking whether $x > - \frac 1 c$ and $u < e^{g(x) + \frac 1 2 x^2}$ for $u \gets \cU(0, 1)$.

    Above rejection sampling ensures that $x$ follows the target PDF $f_0(x)$ when it is accepted, so that the output $d v = d(1 + cx)^3 = h(x)$ follows $\dGamma(\alpha, 1)$, i.e.
    \[
        \alg_G(\alpha, X, U)|_{\alg_G(\alpha, X, U) \neq \bot} \sim \dGamma(\alpha, 1) .
    \]

    Denote the success probability of this rejection sampling by $p_{\alpha}$, which is also the success probability of this algorithm. Then
    \begin{align*}
        p_{\alpha}
        &= \frac 1 M && \text{(\Cref{fact:rej-samp})}\\
        &= \frac 1 {
            \sqrt{2 \pi} \cdot
            3 c d^\alpha  e^{-d} / \Gamma(\alpha)
        }\\
        &= \frac {e^d \Gamma(\alpha)} {3 c d^\alpha \sqrt{2 \pi}}\\
        &= \frac {e^{\alpha - \frac 1 3} \sqrt{\alpha-\frac 1 3} \Gamma(\alpha)} {\sqrt{2 \pi} (\alpha - \frac 1 3)^\alpha}.
    \end{align*}

    From \Cref{fun-lma:p-large}, we know $p_{\alpha} > 0.95$, so $\Pr\left[\alg_G(\alpha, X, U) = \bot\right] < 0.05$.
\end{proof}

\begin{definition}[Rounded Gaussian distribution \cite{BS20}]
    We define the rounded standard Gaussian distribution $\cN_{R(\varepsilon, B)}(0, 1)$ to be the output distribution of the following process: Sample $z$ from $\cN(0, 1)$, if $\abs{z} > B$ then output $0$, otherwise output $\floor{z /\varepsilon} \varepsilon$.
\end{definition}

\begin{fact} [Efficient Rounded Gaussian Sampling {\cite[Fact 3.14]{BS20}}] \label{fact:smp-norm}
    There is a sampling algorithm $\alg_{RN}$ that takes $1^m, B$ and a random tape as input, runs in polynomial time, i.e.\ $\poly(m, \log B)$, and samples from a distribution that has statistical distance at most $2^{-m}$ from the rounded standard Gaussian distribution $\cN_{R(2^{-m}, B)}(0, 1)$.
\end{fact}

\begin{definition}[Rounded Beta distribution]\label{def:rounded-beta}
    We define the rounded Beta distribution $\dBeta_{R(\varepsilon)}(\alpha, \beta)$ to be the output distribution of the following process: Sample $z$ from $\dBeta(\alpha, \beta)$, output $\round{z /\varepsilon} \varepsilon$.
\end{definition}

\begin{lemma}\label{lma:smp-beta}
    There exists a deterministic polynomial sampling algorithm $\alg_{RB}$ such that for any $m, \alpha \in \N^+$, with $U \sim \cU_{\poly(m)}$ uniformly random, $\alg_{RB}(1^m, \alpha, U)$ has statistical distance at most $2^{-m}$ from the rounded Beta distribution $\dBeta_{R(2^{-m})}(\alpha, \alpha)$.
    More generally, the output precision and the statistical accuracy can be chosen independently: for any $p,s,\alpha\in\N^+$, with $U\sim\cU_{\poly(p+s)}$ uniformly random,
    \[
        \alg_{RB}(1^p,1^s,\alpha,U)
    \]
    has statistical distance at most $2^{-s}$ from $\dBeta_{R(2^{-p})}(\alpha,\alpha)$.
    The two interfaces agree when the parameters are equal:
    \[
        \alg_{RB}(1^p,1^p,\alpha,U)=\alg_{RB}(1^p,\alpha,U).
    \]
\end{lemma}
\begin{proof}
    Without considering precision, we first show a deterministic algorithm $\alg_{B1}$ to sample from $\dBeta(\alpha, \alpha)$ with a polynomial number of standard normal random variables and uniform random variables as inputs.

    \begin{algorithm}[H]
        \caption{$\alg_{B1}$}
        \label{alg:beta1}
        \begin{algorithmic}[1]
            \INPUT $1^m, \alpha \in \N^+$, $x_1, \cdots, x_{2m} \gets \mathcal{N}(0, 1)$, $u_1, \cdots, u_{2m} \gets \cU(0, 1)$
            \For{$i := 1 \text{ to } m$}
                \State $a := \alg_G(\alpha, x_{2i-1}, u_{2i-1})$
                \If{$a \neq \bot$} \label{alg:beta1:line:check-a}
                    \Break
                \EndIf
            \EndFor
            \For{$i := 1 \text{ to } m$}
                \State $b := \alg_G(\alpha, x_{2i}, u_{2i})$
                \If{$b \neq \bot$} \label{alg:beta1:line:check-b}
                    \Break
                \EndIf
            \EndFor
            \If{$a = \bot$ or $b = \bot$} \label{alg:beta1:line:ab}
                \State return $\bot$
            \Else
                \State return $a / (a+b)$
            \EndIf
        \end{algorithmic}
    \end{algorithm}

    Let the values of $a$ and $b$ at line~\ref{alg:beta1:line:ab} of \Cref{alg:beta1} be $a^*$ and $b^*$. We have
    \begin{align*}
        &\Pr[a^* = \bot ] \leq 0.05^{m} < 2^{-4m}, \\
        &\Pr[b^* = \bot ] \leq 0.05^{m} < 2^{-4m}. &&\text{(From \Cref{lma:smp-gamma})}
    \end{align*}
    Using the union bound, we have
    \begin{align}
        \Pr[a^* = \bot \text{ or } b^* = \bot] < 2^{-4m+1}. \label{eq:b1}
    \end{align}
    This gives
    \[
        \Pr[\alg_{B1}(1^m, \alpha, X_1, \cdots, X_{2m}, U_1, \dots, U_{2m}) = \bot]
        < 2^{-4m+1},
    \]
    where random variables $X_i \sim \cN(0, 1), U_i \sim \cU(0, 1)$ independently.

    From \Cref{fact:gamma-beta}, we know that if both $a^*$ and $b^*$ follow $\dGamma(\alpha, 1)$ independently, then $a^*/(a^* + b^*)$ follows $\dBeta(\alpha, \alpha)$. Thus, we have
    \[
        \alg_{B1}(1^m, \alpha, X_1, \cdots, X_{2m}, U_1, \dots, U_{2m})|_{
            \alg_{B1}(1^m, \alpha, X_1, \cdots, X_{2m}, U_1, \dots, U_{2m}) \neq \bot
        }
        \sim \dBeta(\alpha, \alpha) .
    \]

    We now modify $\alg_{B1}$ to produce an output with $m$-bit precision.
    We define $\tilde{\alg}_G^{\angbra{m_1}}$ to be the rounded version of $\alg_G$, which only reads the first $m_1$ bits after the binary point of inputs (always round down even if the input is negative); i.e.
    \[
        \tilde{\alg}_G^{\angbra{m_1}}(\alpha, x, u) \defeq \alg_G(\alpha, \floor{2^{m_1} x} 2^{-m_1} , \floor{2^{m_1} u} 2^{-m_1}).
    \]
    We need to consider when the output (after rounding to $m$-bit precision) of the original $\alg_G$ and the new one using $\tilde{\alg}_G^{\angbra{m_1}}$ differ. There are only two possible scenarios:
    \begin{enumerate}
        \item $\alg_G(\alpha, x_i, u_i)$ and $\tilde{\alg}_G^{\angbra{m_1}}(\alpha, x_i, u_i)$ behave differently in some iteration: One accepts and the other rejects.
        \item Each pair of them behave the same, but the final output is different under $m$-bit precision, i.e. $\round{2^m a / (a+b)}2^{-m} \neq \round{2^m \tilde{a} / (\tilde{a} + \tilde{b})} 2^{-m}$, where $\tilde{a}, \tilde{b}$ are the outputs of $\tilde{\alg}_G^{\angbra{m_1}}$ in the same iteration as $a, b$.
    \end{enumerate}
    To detect the first one, we define a checker $\cC$ as follows (and choose $\varepsilon = 2^{-m_1}$):
    \begin{itemize}
        \item $\cC(\varepsilon, \alpha, x, u)$:
        \begin{enumerate}
            \item Let $d = \alpha - \frac 1 3, c = \frac 1 {\sqrt{9 d}}$, and $\tilde{x} = \floor{x/\varepsilon} \varepsilon, \tilde{u} = \floor{u/\varepsilon} \varepsilon$.
            \item If $\ind{(1+c x)^3 > 0} \neq \ind{(1+c \tilde{x})^3 > 0}$, output false.
            \item Let $s(x) = \frac 1 2 x^2 + d - d (1+c x)^3 + d \ln((1+c x)^3)$. If $\ind{u < e^{s(x)}} \neq \ind{\tilde{u} < e^{s(\tilde{x})}}$, output false.
            \item Otherwise, output true.
        \end{enumerate}
    \end{itemize}
    This checker checks whether $\alg_G$ behaves the same, both accepting or both rejecting, when run on $(\alpha, x, u)$ versus $(\alpha, \tilde{x}, \tilde{u})$ (see $\alg_G$ \Cref{alg:g}). Thus, as long as $\cC(2^{-m_1}, \alpha, x_{2i-1}, u_{2i-1})$ outputs true, the branch taken by the if statement at line~\ref{alg:beta1:line:check-a} is unaffected by replacing $\alg_G$ with $\tilde{\alg}_G^{\angbra{m_1}}$, and similarly for the branch at line~\ref{alg:beta1:line:check-b}.

    We also exclude the case that $\abs{x}$ is too large, because later we will use $x$ sampled from the rounded Gaussian distribution, whose range is bounded. Now we have our new algorithm $\alg_{B2}$.

    \newcommand{\mm}{\ensuremath{m_1}}

    Let $\mm = 3m + 3 + \ceil{\log_2 (\eta +3)}$, where $\eta$ is the constant from \Cref{fun-lma:esx-d-small}.

    \begin{algorithm}[H]
        \caption{$\alg_{B2}$}
        \label{alg:beta2}
        \begin{algorithmic}[1]
            \INPUT $1^m, \alpha \in \N^+$, $x_1, \cdots, x_{2m} \gets \mathcal{N}(0, 1)$, $u_1, \cdots, u_{2m} \gets \cU(0, 1)$
            \For{$i := 1 \text{ to } m $}
                \If{ $\abs{x_{2 i -1}} > 2m $ or not $\cC(2^{-\mm}, \alpha, x_{2 i -1}, u_{2 i -1})$}
                    \State return $\bot_R$
                \EndIf
                \State $\tilde{a} := \tilde{\alg}_G^{\angbra{\mm}}(1^m, \alpha, x_{2i-1}, u_{2i-1})$
                \State $a := \alg_G(1^m, \alpha, x_{2i-1}, u_{2i-1})$
                \If{$\tilde{a} \neq \bot$}
                    \Break
                \EndIf
            \EndFor
            \For{$i := 1 \text{ to } m $}
                \If{$\abs{x_{2 i}} > 2m$ or not $\cC(2^{-\mm}, \alpha, x_{2 i}, u_{2 i})$}
                    \State return $\bot_R$
                \EndIf
                \State $\tilde{b} := \tilde{\alg}_G^{\angbra{\mm}}(1^m, \alpha, x_{2i}, u_{2i})$
                \State $b := \alg_G(1^m, \alpha, x_{2i}, u_{2i})$
                \If{$\tilde{b} \neq \bot$}
                    \Break
                \EndIf
            \EndFor
            \If{$\tilde{a} = \bot$ or $\tilde{b} = \bot$}
                \State return $\bot$
            \ElsIf {$\round{2^m a / (a+b)} 2^{-m} \neq \round{2^m \tilde{a} / (\tilde{a} + \tilde{b})} 2^{-m}$}
                \State return $\bot_R$ \label{alg:beta2:line:bot-R-at-end}
            \Else
                \State return $\round{2^m \tilde{a} / (\tilde{a} + \tilde{b})} 2^{-m}$
            \EndIf
        \end{algorithmic}
    \end{algorithm}\

    We need to check the probability that $\alg_{B2}$ outputs $\bot_R$.
    From \Cref{fact:mills}, we have
    \begin{align}
        \Pr_{x \gets \cN(0, 1)}[\abs{x} > 2m] < e^{-\frac {(2m)^2} {2}}/(2m) . \label{eq:check1}
    \end{align}
    We also have
    \begin{align}
        &\Pr_{x \gets \cN(0, 1)}\sparen{\ind{(1+cx)^3 > 0} \neq \ind{(1+c\floor{x/\varepsilon}\varepsilon)^3 > 0}} \nonumber \\
        &= \Pr_{x \gets \cN(0, 1)}\sparen{ \ind{x > -\frac 1 c} \neq \ind{\floor{x/\varepsilon}\varepsilon > -\frac 1 c}}\nonumber\\
        &\leq \Pr_{x \gets \cN(0, 1)}\sparen{x \in (-\frac 1 c, -\frac 1 c + \varepsilon)} \nonumber\\
        & < \varepsilon.
            \quad \quad \text{(Since PDF of $\cN(0, 1)$ is $\frac{1}{\sqrt{2 \pi}} e^{-\frac 1 2 x^2} < 1$)}
        \label{eq:v}
    \end{align}
    Since $\abs{(e^{s(x)})'} \leq \eta$ (\Cref{fun-lma:esx-d-small}) and $\abs{x -  \floor{x/\varepsilon}\varepsilon} < \varepsilon $, we know
    \begin{align}
        \abs{e^{s(\floor{x/\varepsilon}\varepsilon)} - e^{s(x)}} < \eta \varepsilon. &&\text{(Mean value inequality)}
        \label{eq:esx-diff}
    \end{align}
    We also know that $\abs{u - \floor{u/\varepsilon} \varepsilon} < \varepsilon$. Then, when $u \leq \min\left\{
                    e^{s(x_0)},
                    e^{s(\floor{x_0/\varepsilon}\varepsilon)}
                \right\}$,
    $$
        \ind{u < e^{s(x_0)}}
        =
        \ind{\floor{u/\varepsilon}\varepsilon < e^{s(\floor{x_0/\varepsilon}\varepsilon)}} = 1,
    $$
    and when $u \geq \max\left\{
                    e^{s(x_0)},
                    e^{s(\floor{x_0/\varepsilon}\varepsilon)}
                \right\} + \varepsilon$,
    $$
        \ind{u < e^{s(x_0)}}
        =
        \ind{\floor{u/\varepsilon}\varepsilon < e^{s(\floor{x_0/\varepsilon}\varepsilon)}} = 0.
    $$
    Thus, for a fixed $x_0$ and uniformly random $u$,
    \begin{align*}
        &\Pr_{u \gets \cU(0, 1)} \sparen{
            \ind{u < e^{s(x_0)}}
            \neq
            \ind{\floor{u/\varepsilon}\varepsilon < e^{s(\floor{x_0/\varepsilon}\varepsilon)}}
        }\\
        &\leq \Pr_{u \gets \cU(0, 1)} \sparen{
            u \in \paren{
                \min\left\{
                    e^{s(x_0)},
                    e^{s(\floor{x_0/\varepsilon}\varepsilon)}
                \right\},
                \max\left\{
                    e^{s(x_0)},
                    e^{s(\floor{x_0/\varepsilon}\varepsilon)}
                \right\} + \varepsilon
            }
        }\\
        &= \max\left\{
            e^{s(x_0)},
            e^{s(\floor{x_0/\varepsilon}\varepsilon)}
        \right\} + \varepsilon - \min\left\{
            e^{s(x_0)},
            e^{s(\floor{x_0/\varepsilon}\varepsilon)}
        \right\} \\
        &< (\eta + 1) \varepsilon ,
            \quad \quad \text{(from \cref{eq:esx-diff})}
    \end{align*}
    so
    \begin{align}
        \Pr_{u \gets \cU(0, 1), x \gets \cN(0, 1)} \sparen{\ind{u < e^{s(x)}} \neq \ind{\floor{u/\varepsilon}\varepsilon < e^{s(\floor{x/\varepsilon}\varepsilon)}}} < (\eta + 1) \varepsilon . \label{eq:check2}
    \end{align}

    The final point at which $\alg_{B_2}$ may output $\bot_R$ is in line~\ref{alg:beta2:line:bot-R-at-end}.
    We show that the probability of this is small enough by discussing the following two cases:
    \begin{enumerate}
        \item At least one of $a$ and $b$ is not so close to $0$. In this case, $\abs{a / (a+b) - \tilde{a} / (\tilde{a} + \tilde{b})}$ can be bounded, then the probability that they differ after rounding can be bounded.
        \item Both $a$ and $b$ are very close to $0$. The probability that this happens is small enough.
    \end{enumerate}
    $\alg_{B_2}$ reaches line~\ref{alg:beta2:line:bot-R-at-end} only when both $a$ and $b$ are not $\bot$. Let $x_a, x_b$ be the value of input $x$ of $\alg_G$ for $a$ and $b$ respectively; i.e. $$a = d(1+c x_a)^3, b = d(1+c x_b)^3,$$ where $d, c$ are defined in $\alg_G$. Let $\tilde{x}_a, \tilde{x}_b$ be the rounded version of $x_a, x_b$ respectively; i.e. $$\tilde{x}_a = \floor{x_a / 2^{-\mm}} 2^{-\mm}, \tilde{x}_b = \floor{x_b / 2^{-\mm}} 2^{-\mm} .$$
    We have $x_a - \tilde{x}_a \in [0, 2^{-\mm})$ and $x_b - \tilde{x}_b \in [0, 2^{-\mm})$, so $(1+c x_a) - (1+c \tilde{x}_a) \in [0, c 2^{-\mm})$ and $(1+c x_b) - (1+c \tilde{x}_b) \in [0, c 2^{-\mm})$.

    First consider the case that $x_a > -\frac 1 c + 2^{-m} + 2^{-\mm}$ or $x_b > -\frac 1 c + 2^{-m} + 2^{-\mm}$, which gives $$1+c x_a, 1+c \tilde{x}_a > c 2^{-m} \text{ or } 1+c x_b, 1+c \tilde{x}_b > c 2^{-m} .$$
    Define function $f(u, v) = u^3 / (u^3 + v^3)$, and let $u = 1+c x_a, v = 1+c x_b$, $\tilde{u} = 1+c \tilde{x}_a, \tilde{v} = 1+c \tilde{x}_b$.
    We have
    \begin{align}
        &\abs{\frac a {a+b} - \frac {\tilde{a}} {\tilde{a} + \tilde{b}}} \nonumber\\
        &= \abs{\frac {(1+c x_a)^3} {(1+c x_a)^3 + (1+c x_b)^3} - \frac {(1+c \tilde{x}_a)^3} {(1+c \tilde{x}_a)^3 + (1+c \tilde{x}_b)^3}}\nonumber\\
        &= \abs{f(u, v) - f(\tilde{u}, \tilde{v})} \nonumber\\
        & \leq \max \Bigg\{
            \abs{ (u-\tilde{u}) \min_{\substack{
                u > c 2^{-m} \lor \\ v > c 2^{-m}}
             } \diffp{f}{u}(u, v)
             + (v-\tilde{v}) \min_{\substack{
                u > c 2^{-m} \lor \\ v > c 2^{-m}}
             } \diffp{f}{v}(u, v)
            }, \nonumber\\
            & \quad \quad \quad
            \abs{ (u-\tilde{u}) \max_{\substack{
                u > c 2^{-m} \lor \\ v > c 2^{-m}}
             } \diffp{f}{u}(u, v)
             + (v-\tilde{v}) \max_{\substack{
                u > c 2^{-m} \lor \\ v > c 2^{-m}}
             } \diffp{f}{v}(u, v)
            }
        \Bigg\}
            &&\text{(Mean value inequality)} \nonumber\\
        &< c 2^{-\mm} \cdot \frac 1 {c 2^{-m}} &&\text{(Using \Cref{fun-lma:partial-small})}\nonumber\\
        &= 2^{-\mm + m}.
        \label{eq:case-ab-large}
    \end{align}
    In the penultimate line, a factor of 2 is not needed because $\diffp{f}{u}(u, v)$ and $\diffp{f}{v}(u, v)$ always have opposite signs.

    For the case that $x_a, x_b \in (-\frac 1 c, -\frac 1 c + 2^{-m} + 2^{-\mm}]$, we have
    \begin{align}
        a = d(1+c x_a)^3
        &\in \left(0, d c^3 (2^{-m} + 2^{-\mm})^3\right] \nonumber\\
        &= \left(0, \frac 1 9 c (2^{-m} + 2^{-\mm})^3\right]
            &&\text{(Since $c^2 d = \frac 1 9$)} \nonumber\\
        &\subseteq \left(0, c \left(\frac {2^{-m} + 2^{-\mm}} {2}\right)^3\right) \nonumber\\
        &\subseteq \left(0, c 2^{-3m}\right). &&\text{(Since $\mm > m$)}
        \label{eq:case-ab-small}
    \end{align}
    Similarly, we have $b \in (0, c 2^{-3m})$.

    We already show that $a, b$ follow $\dGamma(\alpha, 1)$ and $a / (a+b)$ follows $\dBeta(\alpha, \alpha)$ when $a, b \neq \bot$, if we don't check and output $\bot_R$ earlier.
    Thus,
    \begin{align}
        &\Pr\left[\text{ $\alg_{B2}$ outputs $\bot_R$ at line~\ref{alg:beta2:line:bot-R-at-end} }\right] \nonumber\\
        &\leq \Pr\left[\text{$\alg_{B_2}$ reaches line~\ref{alg:beta2:line:bot-R-at-end} and $x_a, x_b \in (-\frac 1 c, -\frac 1 c + 2^{-m} + 2^{-\mm}]$}\right] \nonumber\\ &\quad \quad
            + \Pr\left[\text{$\alg_{B2}$ outputs $\bot_R$ at line~\ref{alg:beta2:line:bot-R-at-end} } | \text{ $\alg_{B_2}$ reaches line~\ref{alg:beta2:line:bot-R-at-end} and $x_a$ or $x_b$ larger than $-\frac 1 c + 2^{-m} + 2^{-\mm}$}\right] \nonumber\\
        &<\Pr_{a, b \gets \dGamma(\alpha, 1)}\sparen{a, b \in (0 , c 2^{-3m})}
            \quad \quad \text{(From \cref{eq:case-ab-small})} \nonumber\\
            &\quad \quad
            + \Pr_{r \gets \dBeta(\alpha, \alpha)}\sparen{ \exists \tilde{r}: \abs{r - \tilde{r}} < 2^{-\mm+m}, \round{2^m r} 2^{-m} \neq \round{2^m \tilde{r}} 2^{-m}}
            \quad \quad \text{(From \cref{eq:case-ab-large})}
            \nonumber\\
        &\leq (c 2^{-3m})^2 \quad \quad \text{(Since PDF of $\dGamma(\alpha, 1)$ is $f(x) = \frac 1 {\Gamma(\alpha)} x^{\alpha - 1} e^{-x} \leq 1$)}
            \nonumber\\ &\quad \quad
            + \Pr_{r \gets \dBeta(\alpha, \alpha)}\sparen{ \exists \tilde{r}: \abs{r - \tilde{r}} < 2^{-\mm+m}, \round{2^m r} 2^{-m} \neq \round{2^m \tilde{r}} 2^{-m}}  \nonumber\\
        &= c^2 2^{-6m}
            + \Pr_{r \gets \dBeta(\alpha, \alpha)}\sparen{ r \in \left((k+ \frac 1 2) 2^{-m} - 2^{-\mm+m}, (k+ \frac 1 2) 2^{-m} + 2^{-\mm+m}\right) \text{ for some } k \in \Z } \nonumber\\
        & = c^2 2^{-6m}
            +  \Pr_{r \gets \dBeta(\alpha, \alpha)}\sparen{ r \in \left((k+ \frac 1 2) 2^{-m} - 2^{-\mm+m}, (k+ \frac 1 2) 2^{-m} + 2^{-\mm+m}\right) \text{ for some } k \in \Z, k 2^{-m} < \frac 1 2}
            \nonumber\\ &\quad \quad \quad \quad
            + \Pr_{r \gets \dBeta(\alpha, \alpha)}\sparen{ r \in \left((k+ \frac 1 2) 2^{-m} - 2^{-\mm+m}, (k+ \frac 1 2) 2^{-m} + 2^{-\mm+m}\right) \text{ for some } k \in \Z, k 2^{-m} \geq \frac 1 2}
            \nonumber\\
        &< c^2 2^{-6m}
            + \frac {2 \cdot 2^{-\mm+m}} {\frac 1 2 \cdot 2^{-m}} \Bigg( \Pr_{r \gets \dBeta(\alpha, \alpha)}\sparen{ r \in \left((k+\frac 1 2) 2^{-m}, (k+1) 2^{-m}\right) \text{ for some } k \in \Z, k 2^{-m} < \frac 1 2}
            \nonumber\\ &\quad \quad \quad \quad
            + \Pr_{r \gets \dBeta(\alpha, \alpha)}\sparen{ r \in \left(k 2^{-m}, (k+\frac 1 2) 2^{-m}\right) \text{ for some } k \in \Z, k 2^{-m} \geq \frac 1 2} \Bigg)
            \tag{$\ast$} \nonumber \\
            & \quad \quad  \text{(Since PDF of $\dBeta(\alpha, \alpha)$ is monotonic in $(0, \frac 1 2)$ and $(\frac 1 2, 1)$ respectively, see \Cref{fig:fr-left} and \Cref{fig:fr-right})} \nonumber\\
        &< c^2 2^{-6m} + 4 \cdot 2^{-\mm + 2m} \cdot 1 \nonumber\\
        &< 2^{-6m} + 4 \cdot 2^{-\mm + 2m}.
            \quad \quad \text{(Since $c = \frac 1 {\sqrt{9 \alpha - 3}} \leq \frac {1} {\sqrt{6}} < 1$)}
        \label{eq:check-round}
    \end{align}

    \begin{figure}[H]
\centering
\begin{tikzpicture}
    \begin{axis}[
        axis lines=middle,
        xlabel=$r$,
        ylabel={$f(r)$},
        xtick={0.25, 0.375, 0.5},
        xticklabels={$kd$, $(k+\frac{1}{2})d$, $(k+1)d$},
        ymin=0,
        ymax=1.7,
        domain=0.25-0.02:0.5+0.02,
        ]
        
        \addplot[blue, opacity=0.2, fill=blue, domain=0.375:0.5] {6*x*(1-x)} \closedcycle;

        \addplot[
            pattern=north west lines, 
            pattern color=blue!60!black, 
            draw=none, 
            domain={(0.375-0.03)}:{(0.375+0.03}
        ] {6*x*(1-x)} \closedcycle;

        \addplot[blue, thick, samples=100] {6*x*(1-x)};

        \node[font=\small] at (axis cs:0.37, 1) {$S_1$};
        \node[font=\small] at (axis cs:0.45, 1) {$S_2$};

        \draw[thin, gray] (axis cs:0.375-0.03, 0) -- (axis cs:0.375-0.03, 1.5);
        \draw[thin, gray] (axis cs:0.375, 0) -- (axis cs:0.375, 1.5);
        \draw[thin, gray] (axis cs:0.375+0.03, 0) -- (axis cs:0.375+0.03, 1.5);

        \draw[<->, thick] (axis cs:0.375-0.03, 1.5) -- (axis cs:0.375, 1.5) node[midway, above] {$a$};
        \draw[<->, thick] (axis cs:0.375, 1.5) -- (axis cs:0.375+0.03, 1.5) node[midway, above] {$a$};
    \end{axis}
\end{tikzpicture}
\begin{align*}
        S_1 &< \frac {2 a} {\frac 1 2 d} S_2\\
        \Pr\left[r \in \left((k+\frac 1 2)d - a, (k+\frac 1 2)d - a \right)\right] &< \frac {2 a} {\frac 1 2 d} \Pr\left[r \in \left((k+\frac 1 2)d, (k+1)d \right)\right]
\end{align*}    
\caption{The case $kd < \frac 1 2$}
\label{fig:fr-left}
\end{figure}

\begin{figure}[ht]
\centering
\begin{tikzpicture}
    \begin{axis}[
        axis lines=middle,
        xlabel=$r$,
        ylabel={$f(r)$},
        xtick={0.5, 0.625, 0.75},
        xticklabels={$kd$, $(k+\frac{1}{2})d$, $(k+1)d$},
        ymin=0,
        ymax=1.7,
        domain=0.5-0.02:0.75+0.02,
        ]
        
        \addplot[blue, opacity=0.2, fill=blue, domain=0.5:0.625] {6*x*(1-x)} \closedcycle;

        \addplot[
            pattern=north west lines, 
            pattern color=blue!60!black, 
            draw=none, 
            domain={(0.625-0.03)}:{(0.625+0.03}
        ] {6*x*(1-x)} \closedcycle;

        \addplot[blue, thick, samples=100] {6*x*(1-x)};

        \node[font=\small] at (axis cs:0.62, 1) {$S_1$};
        \node[font=\small] at (axis cs:0.55, 1) {$S_2$};

        \draw[thin, gray] (axis cs:0.625-0.03, 0) -- (axis cs:0.625-0.03, 1.5);
        \draw[thin, gray] (axis cs:0.625, 0) -- (axis cs:0.625, 1.5);
        \draw[thin, gray] (axis cs:0.625+0.03, 0) -- (axis cs:0.625+0.03, 1.5);

        \draw[<->, thick] (axis cs:0.625-0.03, 1.5) -- (axis cs:0.625, 1.5) node[midway, above] {$a$};
        \draw[<->, thick] (axis cs:0.625, 1.5) -- (axis cs:0.625+0.03, 1.5) node[midway, above] {$a$};
    \end{axis}
\end{tikzpicture}
\begin{align*}
        S_1 &< \frac {2 a} {\frac 1 2 d} S_2\\
        \Pr\left[r \in \left((k+\frac 1 2)d - a, (k+\frac 1 2)d - a \right)\right] &< \frac {2 a} {\frac 1 2 d} \Pr\left[r \in \left(kd, (k+\frac 1 2)d \right)\right]
\end{align*}    
\caption{The case $kd \geq \frac 1 2$}
\label{fig:fr-right}
\end{figure}

    Notice that checker $\cC$ runs at most $2 m$ times with independent inputs $x_i, u_i$.
    Combining \cref{eq:check1}, \cref{eq:v}, \cref{eq:check2} and \cref{eq:check-round} we have
    \begin{align}
        &\Pr[\alg_{B2}(1^m, \alpha, X_1, \cdots, X_{2m}, U_1, \dots, U_{2m}) = \bot_R] \nonumber\\
        &< 2m\paren{e^{-2m^2}/(2m) + (\eta + 2) \cdot 2^{-\mm}} + 2^{-6m} + 4 \cdot 2^{-\mm + 2m} \nonumber \\
        &= e^{-2m^2} + 2(\eta + 2)m \cdot 2^{-\mm} + 2^{-6m} + 4 \cdot 2^{-\mm + 2m}, \label{eq:b2}
    \end{align}
    where random variables $X_i \sim \cN(0, 1), U_i \sim \cU(0, 1)$ independently.

    When the output is not $\bot_R$, $\alg_{B 2}$ outputs the rounded result of $\alg_{B1}$ under the same inputs. Combining \cref{eq:b1} and \cref{eq:b2}, we have
    \begin{align}
        &\Delta\paren{
            \alg_{B2}(1^m, \alpha, X_1, \cdots, X_{2m}, U_1, \dots, U_{2m}),
            \dBeta_{R(2^{-m})}(\alpha, \alpha)
        } \nonumber \\
        &< 2^{-4m+1} + e^{-2m^2} + 2(\eta +2)m \cdot 2^{-\mm} + 2^{-6m} + 4 \cdot 2^{-\mm + 2m}. \label{eq:b2-dis}
    \end{align}

    By modifying $\alg_{B2}$ to use rounded inputs and avoid output $\bot$ or $\bot_R$, we obtain $\alg_{B3}$.
    \begin{algorithm}[H]
        \caption{$\alg_{B3}$}
        \begin{algorithmic}[1]
            \INPUT $1^m, \alpha \in \N^+$, $x_1, \cdots, x_{2m} \gets \cN_{R(2^{-\mm}, 2m)}(0, 1)$, $u_1, \cdots, u_{2m} \gets \{0, 1\}^{\mm}$
            \For{$i := 1 \text{ to } m $}
                \State $\tilde{a} := \tilde{\alg}_G^{\angbra{\mm}}(1^m, \alpha, x_{2i-1}, u_{2i-1})$
                \If{$\tilde{a} \neq \bot$}
                    \Break
                \EndIf
            \EndFor
            \For{$i := 1 \text{ to } m $}
                \State $\tilde{b} := \tilde{\alg}_G^{\angbra{\mm}}(1^m, \alpha, x_{2i}, u_{2i})$
                \If{$\tilde{b} \neq \bot$}
                    \Break
                \EndIf
            \EndFor
            \If{$\tilde{a} = \bot$ or $\tilde{b} = \bot$}
                \State return $\frac 1 2$
            \Else
                \State return $\round{2^m \tilde{a} / (\tilde{a} + \tilde{b})} 2^{-m}$
            \EndIf
        \end{algorithmic}
    \end{algorithm}

    Notice that in $\alg_{B2}$ we have the following:
    \begin{itemize}
        \item Only the first $\mm$ bits (after the binary point) of $x_i$ and $u_i$ are used for computing $\tilde{a}$ and $\tilde{b}$.
        \item If any $\abs{x_i} > 2m$, then it will output $\bot_R$.
        \item $\bot$ and $\bot_R$ are not in the support of $\dBeta_{R(2^{-m})}(\alpha, \alpha)$.
    \end{itemize}
    Thus, the rounded version of the output of $\alg_{B2}$ is the same as the output of $\alg_{B3}$ with corresponding inputs, when it is not $\bot$ or $\bot_R$.
    From \cref{eq:b2-dis} we have
    \begin{align}
        &\Delta\paren{
            \alg_{B3}(1^m, \alpha, \tilde{X}_1, \cdots, \tilde{X}_{2m}, \tilde{U}),
            \dBeta_{R(2^{-m})}(\alpha, \alpha)
        } \nonumber \\
        &< 2^{-4m+1} + e^{-2m^2} + 2(\eta + 2)m \cdot 2^{-\mm} + 2^{-6m} + 4 \cdot 2^{-\mm + 2m}, \label{eq:b3-dis}
    \end{align}
    where random variables $\tilde{X}_i \sim \cN_{R(2^{-\mm}, 2m)}(0, 1)$, $\tilde{U} \sim \cU_{\mm \cdot 2m}$ independently.

    Finally, algorithm $\alg_{RB}$ is defined as follows:
    \begin{enumerate}
        \item Run $\alg_{RN}$ in \Cref{fact:smp-norm} multiple times to sample each $x_i$ with fresh randomness.
        \item Run $\alg_{B3}$ with those $x_i$'s and a fresh uniform random $u_i$.
    \end{enumerate}
    More formally,
    \begin{align*}
        &\alg_{RB}(1^m, \alpha, \tilde{U}'_1||\tilde{U}'_2||\cdots||\tilde{U}'_{2m}||\tilde{U}_1||\tilde{U}_2||\cdots||\tilde{U}_{2m}) \\
        &\defeq
        \alg_{B3}\paren{
            1^m, \alpha, \alg_{RN}(1^{\mm}, 2m, \tilde{U}'_1), \cdots, \alg_{RN}(1^{\mm}, 2m, \tilde{U}'_{2m}), \tilde{U}_1, \cdots, \tilde{U}_{2m}
        }.
    \end{align*}
    From \Cref{fact:smp-norm}, we have
    \[
        \Delta\paren{
            \alg_{RN}(1^{\mm}, 2m, \tilde{U}),
            \cN_{R(2^{-\mm}, 2m)}(0, 1)
        } \leq 2^{-\mm} ,
    \]
    where $\tilde{U} \sim \cU_{\poly(m)}$. Then, from the triangle inequality and the data processing inequality, we have
    \begin{align*}
        \Delta&\Big(\\
            &\alg_{B3}\paren{
                1^m, \alpha, \alg_{RN}(1^{\mm}, 2m, \tilde{U}'_1), \cdots, \alg_{RN}(1^{\mm}, 2m, \tilde{U}'_{2m}), \tilde{U}_1, \cdots, \tilde{U}_{2m}
            },\\
            &\alg_{B3}\paren{
                1^m, \alpha, \tilde{X}_1, \cdots, \tilde{X}_{2m}, \tilde{U}_1, \cdots, \tilde{U}_{2m}
            }\\
        \Big)& \leq 2m \cdot 2^{-\mm},
    \end{align*}
    which gives
    \begin{align}
        \Delta\paren{
            \alg_{RB}(
                1^m, \alpha, \tilde{U}
            ),
            \alg_{B3}(
                1^m, \alpha, \tilde{X}_1, \cdots, \tilde{X}_{2m}, \tilde{U}_1, \cdots, \tilde{U}_{2m}
            )
        } \leq 2m \cdot 2^{-\mm}, \label{eq:rb}
    \end{align}
    where $\tilde{U}=\tilde{U}'_1\|\tilde{U}'_2\|\cdots\|\tilde{U}'_{2m}\|\tilde{U}_1\|\tilde{U}_2\|\cdots\|\tilde{U}_{2m} \sim \cU_{\poly(m)}$.

    Combining \cref{eq:b3-dis} and \cref{eq:rb}, we have
    \begin{align*}
        &\Delta\paren{
            \alg_{RB}(1^m, \alpha, \tilde{U}),
            \dBeta_{R(2^{-m})}(\alpha, \alpha)
        } \\
        &< 2^{-4m+1} + e^{-2m^2} + 2(\eta + 2)m \cdot 2^{-\mm} + 2^{-6m} + 4 \cdot 2^{-\mm + 2m} + 2m \cdot 2^{-\mm} \nonumber \\
        &< 2^{-m} .
    \end{align*}
    For the last step, notice that $\mm = 3m + 3 + \ceil{\log_2 (\eta +3)} > 3m + 4$, so
    $$2(\eta + 2)m \cdot 2^{-\mm} + 2m \cdot 2^{-\mm} = 2(\eta + 3)m \cdot 2^{-\mm} \leq (\eta + 3) \cdot 2^m 2^{-\mm} < \frac 1 8 \cdot 2^{-m}$$
    and
    $$4 \cdot 2^{-\mm + 2m} = 2^{- \mm + 2m + 2} < \frac 1 4 \cdot 2^{-m}.$$

    To obtain the more general guarantee, first suppose that $s\leq p$. In this case, define
    \[
        \alg_{RB}(1^p,1^s,\alpha,U)\defeq\alg_{RB}(1^p,\alpha,U).
    \]
    Its statistical distance from $\dBeta_{R(2^{-p})}(\alpha,\alpha)$ is less than $2^{-p}\leq 2^{-s}$. In particular, this gives the claimed equality of the two interfaces when $s=p$.

    It remains to consider $s>p$. Set $q=p+s+3$ and write $Q_r(y)=\round{2^r y}2^{-r}$. Define
    \[
        \alg_{RB}(1^p,1^s,\alpha,U)
        \defeq Q_p\left(\alg_{RB}(1^q,\alpha,U)\right),
    \]
    where the algorithm on the right is the original three-argument algorithm. If $X\sim\dBeta(\alpha,\alpha)$, the first part of the proof and the data processing inequality give
    \[
        \Delta\left(\alg_{RB}(1^p,1^s,\alpha,U),Q_p(Q_q(X))\right)<2^{-q}.
    \]
    Moreover, $Q_p(Q_q(X))\neq Q_p(X)$ only if $X$ is within $2^{-q-1}$ of a rounding boundary $(k+\frac12)2^{-p}$. Since the density of $\dBeta(\alpha,\alpha)$ is nondecreasing on $(0,\frac12)$ and nonincreasing on $(\frac12,1)$, the same interval comparison used in $(\ast)$ above shows that
    \[
        \Pr[Q_p(Q_q(X))\neq Q_p(X)]
        \leq 4\cdot 2^{-q-1+p}=2^{p-q+1}.
    \]
    Hence
    \[
        \Delta\left(\alg_{RB}(1^p,1^s,\alpha,U),
        \dBeta_{R(2^{-p})}(\alpha,\alpha)\right)
        <2^{-q}+2^{p-q+1}<2^{-s}.
    \]
    The running time and the number of random bits are polynomial in $q$, and hence in $p+s$.
\end{proof}

\section{Technical lemmas on function properties}

\begin{lemma}\label{fun-lma:s-lesser-than-0}
    Let $d > 0$, $c = \frac 1 {\sqrt{9 d}}$ and $g(x) = d \ln (1 + c x)^3 - d(1+ c x)^3 +d$ for $x \in (- \frac 1 c, +\infty)$. Then
    \[
        g(x) \leq -\frac 1 2 x^2.
    \]
\end{lemma}
\begin{proof}
    Let
    \begin{align*}
        s(x)
            &= g(x) - (- \frac 1 2 x^2) \\
            &= \frac 1 2 x^2 + d - d (1 + c x)^3 + d \ln (1 + c x)^3.
    \end{align*}
    Then
    \begin{align*}
        s'(x)
        &= x -3 c d (1 + c x)^2  + \frac {3 c d }{1 + c x}\\
        &= \frac {x (-3 c^4 d x^2 - 9 c^3 d x - 9 c^2 d + cx + 1)} {1 + c x}\\
        &= - \frac {c^2 x^3} {3(1 + c x)}. \quad \quad \text{(Since $c^2 d = \frac 1 9$)}\\
    \end{align*}
    We have
    $s(0) = 0$ and $\begin{cases}
        s'(x) >0, & x \in (- \frac 1 c, 0)\\
        s'(x)     \leq 0, & x \in [0, + \infty)
    \end{cases}$, 
    so $s(x) \leq 0$ for $x \in (-\frac 1 c , + \infty$), which means $g(x) \leq -\frac 1 2 x^2$ for $x \in (-\frac 1 c , + \infty)$.
\end{proof}

\begin{lemma}\label{fun-lma:p-large}
    Let $p(\alpha) = e^{\alpha - \frac 1 3} \sqrt{\alpha-\frac 1 3} \Gamma(\alpha) / \left(\sqrt{2 \pi} (\alpha - \frac 1 3)^\alpha\right)$, $\alpha \in \N^+$. Then $p(\alpha) > 0.95$.
\end{lemma}
\begin{proof}
    Let 
    \begin{align*}
        r(\alpha) 
        &= \frac {p(\alpha+1)} {p(\alpha)} \\
        &= e \frac {\sqrt{\alpha + \frac 2 3}} {\sqrt{\alpha - \frac 1 3}} \alpha \frac {(\alpha - \frac 1 3)^\alpha} {(\alpha + \frac 2 3)^{\alpha + 1}} \\
        & = e \alpha \frac {(\alpha - \frac 1 3)^{\alpha - \frac 1 2}} {(\alpha + \frac 2 3)^{\alpha + \frac 1 2}} .
    \end{align*}
    Let
    \begin{align*}
        f(\alpha)
        &=\ln r(\alpha) \\
        &= 1 + \ln \alpha + (\alpha - \frac 1 2) \ln (\alpha - \frac 1 3) - (\alpha + \frac 1 2) \ln (\alpha + \frac 2 3).
    \end{align*}

    Notice that $f(\alpha)$ is well-defined and differentiable on $[1, +\infty)$.
    Let $\beta = \alpha + \frac 2 3$, we have
    \begin{align*}
        f'(\alpha)
        &= \frac 1 \alpha + \ln (\alpha - \frac 1 3) + \frac {\alpha - \frac 1 2} {\alpha - \frac 1 3} - \ln (\alpha + \frac 2 3) - \frac {\alpha + \frac 1 2} {\alpha + \frac 2 3} \\
        &= \frac 1 \alpha + \ln (\beta - 1) + \frac {\alpha - \frac 1 2} {\beta -1} - \ln \beta - \frac {\alpha + \frac 1 2} {\beta} \\
        &= \ln(\frac {\beta - 1} {\beta}) + \frac 1 \alpha + \frac {\alpha - \frac 1 2} {\beta - 1} - \frac {\alpha + \frac 1 2} {\beta} \\
        &= \ln(1 - \frac 1 {\beta}) + \frac 1 \alpha - \frac {1} {6 \beta (\beta - 1)} \\
        &= - \sum_{k=1}^{+\infty} \frac 1 {k \beta^k} + \frac 1 \alpha - \frac {1} {6 \beta (\beta - 1)}
         \quad \quad \text{(Taylor expansion)}\\
        &< - \sum_{k=1}^3 \frac 1 {k \beta^k} + \frac 1 \alpha - \frac {1} {6 \beta (\beta - 1)}\\
        &= - \frac {2\alpha(\beta-1) + 3 \alpha \beta (\beta-1) + 6 \alpha \beta^2(\beta-1) - 6 \beta^3(\beta-1) + \alpha \beta^2} {6 \alpha \beta^3 (\beta-1) } \\
        &= - \frac { \frac 1 {27} (3 \alpha - 4)^2 } {6 \alpha \beta^3 (\beta-1) }\\
        & \leq 0 .
    \end{align*}
    This shows that $f(\alpha)$ decreases monotonically, so $r(\alpha)$ decreases monotonically.

    We also have
    \begin{align*}
        \lim_{\alpha \to +\infty} r(\alpha)
        &= \lim_{\alpha \to +\infty} e \frac {\alpha} {\alpha - \frac 1 3} \frac {(\alpha - \frac 1 3)^{\alpha + \frac 1 2}} {(\alpha + \frac 2 3)^{\alpha + \frac 1 2}} \\
        & = e \lim_{\alpha \to +\infty} \paren{\frac {\alpha - \frac 1 3} {\alpha + \frac 2 3}}^{\alpha + \frac 1 2} \\
        & = e \lim_{\alpha \to +\infty} \paren{1 - \frac 1 {\alpha + \frac 2 3}}^{\alpha + \frac 1 2} \\
        & = e \cdot e^{
            \lim_{\alpha \to +\infty} \ln\left( 1 - \frac 1 {\alpha + \frac 2 3}\right) (\alpha + \frac 1 2)
        }\\
        & = e \cdot e^{
            \lim_{\alpha \to +\infty} - \left(\frac 1 {\alpha + \frac 2 3} + O\left(\left(\frac 1 {\alpha + \frac 2 3}\right)^2\right) \right) (\alpha + \frac 1 2)
        } \quad \quad \text{(Taylor expansion)}\\
        & = e \cdot e^{-1} \\ 
        & = 1.
    \end{align*}
    Thus, $r(\alpha) > 1$. This gives that $p(\alpha)$ increases monotonically. Then for all $\alpha \geq 1$,
    \[
        p(\alpha) \geq p(1) = e^{\frac 2 3} \sqrt{\frac 2 3} / (\sqrt{2 \pi} \cdot \frac 2 3) \approx 0.95167 > 0.95.
    \]
\end{proof}

\begin{lemma}\label{fun-lma:partial-small}
    Let $f(x, y) = \frac {x^3} {x^3 + y^3}$, $x, y \in \R^+$. Given constants $\delta > 0$, then 
    \begin{align*}
        0 < \diffp{f}{x} < \frac 1 \delta,\\
        -\frac 1 \delta < \diffp{f}{y} < 0,
    \end{align*}
    when $x > \delta$ or $y > \delta$.
\end{lemma}
\begin{proof}
    Let $u = x^3, v = y^3$, then
    \begin{align*}    
        \diffp{f}{x}
        &= \diffp{f}{u} \diff{u}{x} \\
        & = \frac {v} {(u + v)^2} \cdot 3 x^2 \\
        & = \frac {3 x^2 y^3} {(x^3 + y^3)^2} \\
        &> 0. &&\text{(Since $x, y \in \R^+$)}
    \end{align*}

    When $x > \delta$,
    \begin{align*}
        \diffp{f}{x} =
        &\frac {3 x^2 y^3} {(x^3 + y^3)^2} \\
        &\leq \frac {3 x^2 y^3} {4 x^3 y^3} &&\text{(Since $(A + B)^2 \geq 4 A B$)} \\
        &= \frac {3} {4 x} \\
        &< \frac {3} {4 \delta} \\
        &< \frac 1 \delta.
    \end{align*}

    When $y > \delta$, let $t = \frac x y$, then
    \begin{align*}
        \diffp{f}{x} 
        &= \frac {3 x^2 y^3} {(x^3 + y^3)^2} \\
        &= \frac {3 (t y)^2 y^3} {((t y)^3 + y^3)^2} \\
        &= \frac 3 y \cdot \paren{\frac {t} {t^3+1}}^2 \\
        &\leq \frac 3 y \cdot \frac {2^{\frac 4 3}} {9} &&\text{(
            $\frac t {t^3+1}$ achieves maximum value $\frac {2^{\frac 4 3}} {9}$
            at $t = 2^{-\frac 1 3}$ for $t > 0$
            )} \\
        &\leq \frac {2^{\frac 4 3}} {3 \delta} \\
        &< \frac 1 \delta.
    \end{align*}

    $f(x, y)$ can be represented as $f(x, y) = 1 - \frac {y^3} {x^3 + y^3} = 1 - f(y, x)$, which shows the result for $\diffp{f}{y}$.
\end{proof}

\begin{lemma}\label{fun-lma:esx-d-small}
    Let $\alpha \geq 1$, $d = \alpha - \frac 1 3$, $c = \frac 1 {\sqrt{9 d}}$, and $s(x) = \frac 1 2 x^2 + d - d (1 + c x)^3 + d \ln ( (1 + c x)^3)$ for $x \in (-\frac{1}{c}, +\infty)$.
    Let $f(x) = e^{s(x)}$, then there exists a constant $\eta$ (independent of $\alpha$) such that
    \[
        \abs{f'(x)} \leq \eta . 
    \]
\end{lemma}
\begin{proof}
    Same as in \Cref{fun-lma:s-lesser-than-0}, We have
    $s'(x) = - \frac {c^2 x^3} {3 (1 + c x)}$.
    Then
    \begin{align}
        f'(x)
        &= e^{s(x)} s'(x) \nonumber\\
        &= e^{\frac 1 2 x^2 + d - d (1 + c x)^3} (1+cx)^{3d} \cdot - \frac {c^2 x^3} {3 (1 + c x)} \nonumber\\
        &= - \frac {c^2} 3  x^3 (1+cx)^{3d - 1} e^{\frac 1 2 x^2 + d - d (1 + c x)^3}  . \label{eq:f-prime}
    \end{align}
    By direct computation, we have
    \begin{align*}
        &\lim_{x \to \left(-\frac{1}{c}\right)^+} f'(x) = 0 , \\
        &\lim_{x \to +\infty} f'(x) = 0 ,
    \end{align*}
    and
    \begin{align*}
        \begin{cases}
            f'(x) > 0, & x \in (-\frac{1}{c}, 0) \\
            f'(x) = 0, & x = 0 \\
            f'(x) < 0, & x \in (0, +\infty) .
        \end{cases}
    \end{align*}

    Since $f'(x)$ is continuous and differentiable, we know that $f'(x)$ achieves its maximum value at some point in $(-\frac{1}{c}, 0)$ and its minimum value at some point in $(0, +\infty)$, both only when $f''(x) = 0$.  
    Since
    $$f''(x) = e^{s(x)} \paren{s'(x)^2 + s''(x)},$$
    $f''(x) = 0$ if and only if $s'(x)^2 + s''(x) = 0$.
    We have
    \begin{align*}
        &s'(x)^2 + s''(x) \\
        &=\frac {c^4 x^6} {9 (1 + c x)^2} - \frac {c^2 x^2 (2 c x + 3)} {3 (1 + c x)^2} \\
        &= \frac {c^2 x^2 (c^2 x^4 - 6 c x - 9)} {9(c x + 1)^2} .\\
    \end{align*}
    Thus, $f''(x) = 0$ if and only if $c^2 x^4 - 6 c x - 9 = 0$, for both $x \in (-\frac{1}{c}, 0)$ and $x \in (0, +\infty)$

    Let $w(x) = c^2 x^4 - 6 c x - 9$, then
    \[
        w'(x) = 4 c^2 x^3 - 6 c
    \]
    has only one real root $x_0 = \sqrt[3]{\frac 3 {2 c}}$, and we have
    \begin{align*}
        &w(-\frac 1 c) = \frac 1 {c^2} - 3 = 9 d -3 > 0,\\
        &w(0) = -9 < 0,\\
        &\lim_{x \to +\infty} w(x) = +\infty.
    \end{align*}
    These show that $w(x)$ has exactly one real root in $(-\frac 1 c, 0)$ and one real root in $(0, +\infty)$, denote them as $x_1$ and $x_2$ respectively. Then they are the points where $f'(x)$ achieves its maximum and minimum values, respectively, i.e. $f'(x_2) \leq f'(x) \leq f'(x_1)$.
    
    This already shows that for each fixed $c$, there exists $\eta_1(c) = f'(x_1), \eta_2(c) = f'(x_2)$ as the upper and lower bounds of $f'(x)$, respectively. Now we consider the bound over all possible $c$.
    Since $\alpha \geq 1$, we have $c = \frac {1} {\sqrt{9 \alpha - 3}} \in (0, \sqrt{6}]$.

    First, we consider the area $x \in (-\frac 1 c, 0)$.
    Let $u = c x^2$, then $x = - \sqrt{\frac u c}$. $c^2 x^4 - 6 c x - 9 = 0$ can be represented as 
    \begin{align}
        u^2 + 6 \sqrt{cu} - 9 = 0 \label{eq:u-},
    \end{align}
    which has a positive root $u_1$.
    When $c=0$ the root $u_1=3$, and $p(u, c) = u^2 + 6 \sqrt{cu} - 9$ is a continuous function, so we have
    \[
        \lim_{c \to 0^+} u_1 = 3.
    \]
    We also represent $f'(x)$ in terms of $u$ by putting $x = - \sqrt{\frac u c}$ into \cref{eq:f-prime}:
    \[
        h_1(u) = \frac 1 3  \sqrt{c u^3} (1 - \sqrt{c u})^{3d - 1} e^{\frac {u} {2c} + d - d (1 - \sqrt{c u})^3}.
    \]
    Thus, we have
    \begin{align*}
        \lim_{c \to 0^+} \eta_1(c)
        &= \lim_{c \to 0^+} h_1(u_1) \\
        &= \lim_{c \to 0^+} h_1(3) \\
        &= \lim_{c \to 0^+} \frac 1 3 \sqrt{3^3 c} (1 - \sqrt{3 c})^{3d - 1} e^{\frac 3 {2c} + d - d (1 - \sqrt{3 c})^3} \\
        &= \lim_{t \to 0^+} t (1-t)^{\frac 3 {t^4}-1} e^{\frac 9 {2t^2} + \frac 1 {t^4} - \frac 1 {t^4} (1-t)^3}
            \quad \quad \text{(Let $t = \sqrt{3 c} = \frac 1 {\sqrt[4]{d}}$)}\\
        &= \lim_{t \to 0^+} t (1-t)^{\frac 3 {t^4}-1} e^{\frac 3 {t^3} + \frac 3 {2t^2} + \frac 1 {t}} \\
        &= \lim_{t \to 0^+} t e^{\left(\frac 3 {t^4} -1\right) \ln (1-t) + \frac 3 {t^3} + \frac 3 {2t^2} + \frac 1 {t}} \\
        &= \lim_{t \to 0^+} t e^{- \left(\frac 3 {t^4} -1 \right) \left( \sum_{i=1}^4 \frac {t^i} {i} + O(t^5) \right)
            + \frac 3 {t^3} + \frac 3 {2t^2} + \frac 1 {t}}
            \quad \quad \text{(Taylor expansion)}\\
        &= \lim_{t \to 0^+} t e^{-\left( \frac 3 {t^3} + \frac 3 {2 t^2} + \frac 1 t + \frac 3 4 + O(t)  \right)
            + \frac 3 {t^3} + \frac 3 {2t^2} + \frac 1 {t}}\\
        &= \lim_{t \to 0^+} t e^{-\frac 3 4 + O(t)} \\
        &= 0.
    \end{align*}

    Then we consider the area $x \in (0, +\infty)$. Let $u = c x^2$, then $x = \sqrt{\frac u c}$. $c^2 x^4 - 6 c x - 9 = 0$ can be represented as
    \begin{align}
        u^2 - 6 \sqrt{cu} - 9 = 0 \label{eq:u+},
    \end{align}
    which has root $u_2$.
    Similarly, we can get
    \begin{align*}
        \lim_{c \to 0^+} u_2 = 3,
    \end{align*}
    represent $f'(x)$ as
    \begin{align*}
        h_2(u) = \frac 1 3 \sqrt{c u^3} (1 + \sqrt{c u})^{3d - 1} e^{\frac {u} {2c} + d - d (1 + \sqrt{c u})^3},
    \end{align*}
    and get
    \begin{align*}
        \lim_{c \to 0^+} \eta_2(c) = \lim_{c \to 0^+} h_2(3) = 0 .
    \end{align*}

    Since $\lim_{c \to 0^+} \eta_1(c) = \lim_{c \to 0^+} \eta_2(c) = 0$, there exists a constant $c_0$ such that for all $c \in (0, c_0)$, we have $\abs{\eta_1(c)}, \abs{\eta_2(c)} < 1$. Then let the constant
    \[
        \eta = \max\{1, \max_{c \in [c_0, \sqrt{6}]} \abs{\eta_1(c)}, \max_{c \in [c_0, \sqrt{6}]} \abs{\eta_2(c)} \}.
    \]
    We have $\abs{\eta_1(c)}, \abs{\eta_2(c)} \leq \eta$ for all $c \in (0, \sqrt{6}]$, which means $\abs{f'(x)} \leq \eta$ always holds.
\end{proof}

\end{document}